\documentclass[12pt]{article}
\usepackage[utf8]{inputenc}
\usepackage[export]{adjustbox}
\usepackage{authblk}
\usepackage{bbm}
\usepackage{pgfpages}
\usepackage{multicol}
\usepackage{csquotes}
\usepackage{multirow}
\usepackage{booktabs} 
\usepackage{authblk}
\usepackage{caption}
\usepackage{subcaption}
\usepackage{stackengine}
\usepackage{threeparttable}
\usepackage{floatrow}

\usepackage[T1]{fontenc}

\usepackage{tikz}

\usetikzlibrary{
    arrows.meta,
    positioning,
    shapes,
    bending,
    calc,
    fit,
    backgrounds
}

\usepackage{amsmath,amssymb,amsthm,bm}

\newtheorem{proposition}{Proposition}
\usepackage{setspace}
\usepackage{graphicx}

\usepackage{xcolor}

\usepackage{comment}
\usepackage{lscape}
\usepackage{hyperref}
\hypersetup{
    colorlinks=true,
    linkcolor=blue,
    filecolor=blue,      
    urlcolor=tiblue,
    citecolor=blue,
}

\usepackage[spanish,english]{babel}

\usepackage{natbib} % bibliography
\usepackage[a4paper, margin=1in]{geometry}

\usepackage{setspace}
\usepackage{footmisc}

\usepackage{caption}
\AtBeginEnvironment{thebibliography}{\singlespacing}

\title{Bank Earnings, Credit Supply \& the Macroeconomy \\ Evidence from Canada}

\author[1]{Santiago Camara}

\author[2]{Sanaa Latif}

\affil[1]{McGill University \& Red-NIE}

\affil[2]{McGill University}

\date{\normalsize This Draft \today}

\begin{document}
    
\maketitle

\begin{abstract}
This paper studies whether news about banks' balance sheets propagates to aggregate financial conditions and macroeconomic activity. We construct high-frequency Canadian bank net-worth shocks using stock-price reactions around earnings announcements of the six large Canadian banks. Guided by a model in which higher intermediary net worth expands credit supply and lowers borrowing spreads, we use the co-movement between bank equity prices and Canadian corporate spreads to purge raw bank equity surprises from contaminating information. Favorable purged credit-supply bank net-worth shocks lower corporate spreads, raise bank valuations and broader equity prices, appreciate the Canadian dollar, and increase real activity over the medium run. The results are robust across specifications, samples, and additional outcomes, and suggest that bank earnings news is macroeconomically relevant in concentrated banking systems.

\bigskip


\medskip


\end{abstract}

%%%%%%%%%%%%%%%%%%%%%%%%%%%%%%%%%%%%%%%%%%%%%%%%%%%%%%%%%%%%%%%%%%%%%%
%%%%%%%%%%%%%%%%%%%%%%%%%%%%%%%%%%%%%%%%%%%%%%%%%%%%%%%%%%%%%%%%%%%%%%
% -------------------------------------------------
% Introduction
% -------------------------------------------------
\newpage
\section{Introduction} \label{sec:introduction}

What role do banks play in aggregate fluctuations? A large theoretical literature argues that shocks to intermediary balance sheets can affect credit supply, asset prices, and real activity. When intermediaries lose net worth, their ability to expand balance sheets deteriorates, lending spreads increase, and credit to the real economy contracts. Conversely, improvements in intermediary net worth can ease borrowing conditions and stimulate activity. These mechanisms are central to theories of financial amplification and macroeconomic crises \citep{bernanke1983non,bernanke1999financial,kiyotaki1997credit,brunnermeier2014macroeconomic,gertler2015banking}. Yet empirically identifying shocks to bank net worth, and tracing their aggregate effects, remains challenging. Bank equity prices respond to many sources of information, including news about banks' own lending capacity, expected loan demand, borrower fundamentals, aggregate macroeconomic conditions, risk premia, and monetary policy. The central empirical problem is therefore to isolate the component of changes in bank valuations that can be interpreted as credit-supply news.

This paper constructs high-frequency Canadian bank net-worth shocks using stock-price reactions around the earnings announcements of the six large Canadian banks. The empirical strategy builds on the insight of \citet{ottonello2025financial}, who use earnings-announcement windows of large U.S. financial intermediaries to identify high-frequency financial shocks. Earnings announcements are useful for identification because they are predetermined, bank-specific events that reveal information about bank profitability, future net worth, and lending capacity. Around these events, stock-price movements provide a direct market-based measure of the news released about the intermediary. We adapt this idea to the Canadian setting and use it to study the macroeconomic consequences of bank net-worth shocks.

Canada is an especially useful setting for this question. The Canadian banking sector is highly concentrated and dominated by a small number of large institutions. Earnings news about Royal Bank of Canada, Toronto-Dominion Bank, Bank of Montreal, Bank of Nova Scotia, Canadian Imperial Bank of Commerce, and National Bank of Canada therefore represents news about a quantitatively important part of the domestic intermediation sector. This feature distinguishes the Canadian setting from the U.S. banking system. While the U.S. financial sector is larger, intermediation is substantially more fragmented across banks, nonbanks, broker-dealers, and capital-market institutions. By contrast, large Canadian banks account for a central share of domestic financial intermediation and have been studied as a distinctive banking system relative to the United States \citep{allen2006canadian}. Their scale, concentration, and resilience are also emphasized in work comparing Canadian banks to international banking systems, especially around the global financial crisis \citep{ratnovski2009canadian}. This institutional structure makes Canada a natural environment in which to ask whether news about large banks' balance sheets has aggregate effects on borrowing conditions and macroeconomic activity.

The Canadian setting is also important because Canada is a small open economy exposed to global financial and trade shocks. Recent work shows that international monetary disturbances can affect the Canadian economy through both financial conditions and trade linkages.\footnote{A first example is \cite{rey2015dilemma} and \cite{rey2016international}, which present evidence on the international spillovers of US monetary policy in the Canadian economy. Second, \citep{camara2025between} shows that monetary policy shocks from the European Central Bank have significant financial and trade linkages with the Canadian economy.} This paper studies a complementary domestic source of variation in financial outcomes: news about Canadian banks' balance sheets. By focusing on bank earnings announcements, we isolate shocks originating in the domestic intermediation sector and ask whether they propagate to Canadian corporate spreads, asset prices, real activity, and prices.

Our empirical approach closely follows \citet{ottonello2025financial}, but adapts their high-frequency earnings-announcement strategy to the Canadian setting and to a more explicitly macroeconomic question. The first difference is practical. While their baseline identification uses narrow intraday windows around U.S. financial intermediaries' earnings announcements, our implementation uses timing-adjusted daily event windows, reflecting the availability of Canadian bank equity prices and Canadian corporate spread data. For each earnings release, we use the Bloomberg earnings-release timestamp to define the relevant stock-price reaction: previous close to announcement-day open for pre-market releases, open to close for intraday releases, and announcement-day close to next-day open for post-market releases. This preserves the central logic of the high-frequency design, isolating the market reaction around predetermined earnings-news events, while adapting the event window to the Canadian data environment. The second difference is substantive. Whereas \citet{ottonello2025financial} focus primarily on financial-market and cross-sectional firm-level transmission, our analysis uses the identified shocks to study aggregate macroeconomic propagation in Canada. We estimate the dynamic responses of Canadian corporate spreads, bank valuations, equity prices, exchange rates, monetary policy, real activity, prices, sectoral output, and labor market outcomes.

The paper begins with a simple theoretical framework that motivates this identification strategy. Financial intermediaries face costly balance-sheet expansion because external equity issuance is costly, and they compete strategically in the loan market. In this environment, an increase in intermediary net worth lowers the marginal cost of supplying loans, expands credit supply, reduces equilibrium lending spreads, and raises bank equity valuations. The model delivers a clear empirical implication: favorable news about banks' net worth should raise bank stock prices and lower corporate borrowing spreads. This implication guides the construction of the purged credit-supply component of bank equity surprises.

We construct the shocks using Bloomberg data from 2002:Q4 to 2026:Q1. We focus on earnings releases for the Canadian-listed tickers of the six large Canadian banks. We first show that earnings announcements are information events. Announcement days feature substantially larger absolute and squared stock-price movements than non-announcement days, and timing-adjusted event-window stock-price reactions are strongly positively related to Bloomberg earnings surprises. We then aggregate bank-level stock-price reactions using lagged quarterly market-capitalization shares. The use of lagged weights ensures that larger banks receive greater weight in the aggregate surprise while avoiding mechanical effects of the announcement-window price reaction.

The raw weighted bank equity surprise is informative, but it is not yet a clean credit-supply shock. Positive bank earnings news may reflect an improvement in banks' balance-sheet strength and lending capacity, but it may also reflect contaminating information about expected loan demand, borrower fundamentals, aggregate conditions, or risk premia. To purge this contamination, we use the model's key empirical implication: favorable credit-supply news should raise bank equity valuations and lower corporate borrowing spreads. We therefore identify the purged credit-supply shock as the component of bank equity news that moves bank equity prices and Canadian corporate spreads in opposite directions. News that raises bank valuations and corporate spreads in the same direction is treated as contaminating information rather than as the object of interest. Our benchmark implementation uses the median admissible rotational angle, and we also consider a simpler poor man's sign-restriction approach. This use of sign restrictions is related to a broader literature showing that high-frequency financial-market surprises often combine multiple economic forces, so that separating shocks by their asset-price co-movement is important for interpretation \citep{jarocinski2020deconstructing,jarocinski2022central,camara2025spillovers}.

The main empirical results show that favorable purged credit-supply bank net-worth shocks have meaningful aggregate effects. A one-standard-deviation favorable shock lowers Canadian corporate spreads on impact and over the following year. It raises the real aggregate market capitalization of Canadian banks, increases the broader Canadian equity index, and appreciates the Canadian dollar. Real activity rises gradually, with GDP peaking roughly one year after the shock. The Bank of Canada policy rate increases over the medium run, consistent with an endogenous response to improved financial conditions and stronger activity. These responses are consistent with the mechanism emphasized by the model: news that improves intermediary net worth eases corporate credit conditions and propagates to the real economy.

We provide several pieces of additional evidence. First, local projections using the raw weighted bank equity surprise generate broadly similar but less sharply interpretable responses, consistent with the idea that raw earnings news contains both credit-supply news and contaminating information. Second, responses using the poor man's sign-restriction shock are close to the benchmark results, suggesting that the findings are not driven by the specific median-rotation implementation. Third, we use the median-rotation purged credit-supply shock as an external instrument for Canadian corporate spreads. The IV exercise shows that bank-induced increases in corporate spreads lower bank valuations, reduce equity prices, and contract real activity. This result reinforces the interpretation that corporate borrowing spreads are an important transmission channel through which bank net-worth shocks affect the macroeconomy.

The results are robust across specifications and samples. The main propagation patterns survive when we use two lags rather than six lags of macro-financial controls, when we add a deterministic trend, when we restrict the sample to the pre-2020 period, and when we focus on the post-2010 period. Additional outcome exercises show that the real effects extend beyond aggregate GDP. Favorable credit-supply bank net-worth shocks raise employment and total hours worked and reduce unemployment over the medium run. Sectoral responses indicate that the real-side effects are stronger in goods-producing and resource-related sectors than in services, suggesting that the shock operates through credit-sensitive components of economic activity.

\paragraph{Related literature.}
This paper contributes to several literatures. First, it contributes to empirical work on financial intermediation and aggregate fluctuations by providing high-frequency evidence that bank net-worth shocks affect macroeconomic outcomes. Existing work has shown that intermediary balance sheets matter for asset prices, credit supply, and firm outcomes \citep{khwaja2008tracing,adrian2014financial,chodorow2014employment,siriwardane2025segmented}. Recent evidence also shows that global banks' net worth can transmit financial shocks across borders and affect emerging-market macroeconomic and microeconomic outcomes \citep{arnabal2025global}. We complement this literature by focusing on aggregate propagation in Canada, a setting in which a small number of large domestic banks play a central role in financial intermediation and the banking system differs sharply from the more fragmented U.S. system \citep{allen2006canadian,allen2007efficiency,ratnovski2009canadian}.

Second, the paper contributes to the literature using high-frequency identification to study macroeconomic shocks. Following the logic of high-frequency monetary policy identification \citep{kuttner2001monetary,gurkaynak2004actions,nakamura2018high,nakamura2018identification} and building directly on \citet{ottonello2025financial}, we use narrow windows around predetermined information events to identify shocks to the financial sector. The purging step is also related to work showing that high-frequency surprises can mix policy, information, and risk-premium components, and that their macroeconomic interpretation depends on how these components are separated \citep{jarocinski2020deconstructing,camara2025spillovers}. Third, the paper contributes to work on credit-supply shocks and macroeconomic transmission by showing that earnings-announcement news about banks affects corporate spreads, asset prices, real activity, labor markets, and sectoral production.

The rest of the paper proceeds as follows. Section \ref{sec:illustrative_model} presents the illustrative theoretical framework. Section \ref{sec:shock_construction} describes the data, validates the earnings-announcement stock-price reactions, and constructs the purged credit-supply bank net-worth shocks. Section \ref{sec:Macroeconomic_Propagation} estimates the macroeconomic propagation of the shocks using monthly local projections and provides additional evidence from raw shocks, alternative sign restrictions, and an IV spread exercise. Section \ref{sec:robustness} reports additional results and robustness checks. Section \ref{sec:conclusions} concludes.

%%%%%%%%%%%%%%%%%%%%%%%%%%%%%%%%%%%%%%%%%%%%%%%%%%%%%%%%%%%%%%%%%%%%%%
%%%%%%%%%%%%%%%%%%%%%%%%%%%%%%%%%%%%%%%%%%%%%%%%%%%%%%%%%%%%%%%%%%%%%%
\section{An Illustrative Theoretical Framework} \label{sec:illustrative_model}

In this section, we present a simple theoretical framework to motivate the identification strategy and interpret the empirical results. The goal is not to develop a fully dynamic quantitative model of financial intermediation, but rather to formalize the mechanism linking bank net worth, credit supply, lending spreads, and bank equity valuations. The model has three key ingredients. First, financial intermediaries face costly balance-sheet expansion because raising external equity is costly. Second, intermediary net worth affects the marginal cost of supplying loans and therefore equilibrium credit spreads. Third, in a concentrated banking system, news about the capitalization and profitability of large banks can affect aggregate financial conditions.

The model delivers the sign pattern that motivates the empirical strategy. Favorable bank-specific news about intermediary net worth raises the equity valuation of the bank receiving the news. At the same time, higher intermediary net worth expands aggregate credit supply and lowers equilibrium borrowing spreads. Thus, the model predicts that credit-supply news about banks should move bank equity valuations and corporate borrowing spreads in opposite directions. We use this sign pattern below as an identifying restriction to separate the credit-supply component of bank equity surprises from other information released during earnings announcements.

%%%%%%%%%%%%%%%%%%%%%%%%%%%%%%%%%%%%%%%%%%%%%%%%%%%%%%%%%%%%%%%%%%%%%%
\subsection{Environment}

Figure \ref{fig:stylized_model_economy} summarizes the structure of the economy. The economy has two periods, \(t\in\{0,1\}\), and is populated by households, nonfinancial firms, and a finite set of financial intermediaries indexed by \(i\in\{1,\dots,N\}\). Households supply deposits and equity funding to intermediaries and own bank equity. Intermediaries enter the period with net worth \(n_i\), raise deposits and external equity, and supply loans to firms. Firms borrow to finance investment projects, produce at \(t=1\), and repay intermediaries. The assumption that \(N\) is finite captures the concentrated structure of the Canadian banking system. This concentration is central for the mechanism: news about the profitability and capitalization of a small number of large banks can represent news about aggregate intermediation capacity.

%%%%%%%%%%%%%%%%%%%%%%%%%%%%%%%%%%%%%%%%%%%%%%%%%%%%%%%%%%%%%%%%%%%%%%
\begin{figure}[H]
\centering
\begin{tikzpicture}[
    font=\footnotesize,
    >=Latex,
    node distance=1.4cm and 1.6cm,
    box/.style={
        draw,
        thick,
        rounded corners=2pt,
        fill=white,
        text width=2.7cm,
        minimum height=0.9cm,
        align=center,
        inner sep=4pt
    },
    bank/.style={
        draw,
        thick,
        rounded corners=2pt,
        fill=blue!6,
        text width=3.0cm,
        minimum height=0.95cm,
        align=center,
        inner sep=4pt
    },
    result/.style={
        draw,
        thick,
        rounded corners=2pt,
        fill=gray!10,
        text width=7.4cm,
        minimum height=0.85cm,
        align=center,
        inner sep=5pt
    },
    arr/.style={->, thick},
    shock/.style={->, very thick, blue!70!black},
    lab/.style={font=\scriptsize, midway, fill=white, inner sep=1pt}
]

\node[box] (hh) {
    \textbf{Households}\\
    deposits, equity
};

\node[bank, right=of hh] (bank) {
    \textbf{Banks}\\
    net worth \(n_i\)\\
    choose loans \(\ell_i\)
};

\node[box, right=of bank] (firms) {
    \textbf{Firms}\\
    borrow, invest
};

\node[box, below=1.05cm of bank] (frictions) {
    \textbf{Frictions}\\
    \(\ell_i\leq \lambda(n_i+e_i)\)\\
    \(\Psi(e_i)=\frac{\kappa}{2}e_i^2\)
};

\node[box, below=1.05cm of firms] (market) {
    \textbf{Credit market}\\
    \(R(L)=\bar A(1-L)\)\\
    \(S=R-R_D\)
};

\node[result, below=1.15cm of frictions, xshift=1.45cm] (result) {
    \(n_i \uparrow \Rightarrow V_i\uparrow,\; L^*\uparrow,\; S^*\downarrow\)
};

\draw[arr] (hh) -- node[lab, above] {\(d_i,e_i\)} (bank);
\draw[arr] (bank) -- node[lab, above] {\(\ell_i\)} (firms);
\draw[arr] (firms) to[bend left=18] node[lab, below] {repayment} (bank);

\draw[arr] (frictions) -- node[lab, left] {balance-sheet cost} (bank);
\draw[arr] (firms) -- node[lab, right] {credit demand} (market);
\draw[arr] (market) -- node[lab, below] {\(R,S\)} (frictions);

\draw[shock] (bank) -- node[lab, right] {higher net worth} (frictions);
\draw[shock] (frictions) -- node[lab, left] {easier credit supply} (result);

\end{tikzpicture}

\caption{Stylized Model Economy}
\label{fig:stylized_model_economy}
\floatfoot{\textbf{Notes:} Households supply deposits and equity to banks. Banks enter with net worth, face leverage and external-equity-issuance frictions, and supply loans to firms. Higher bank net worth relaxes effective lending capacity, raises aggregate lending, lowers spreads, and increases the equity value of the bank receiving the net-worth shock.}
\end{figure}
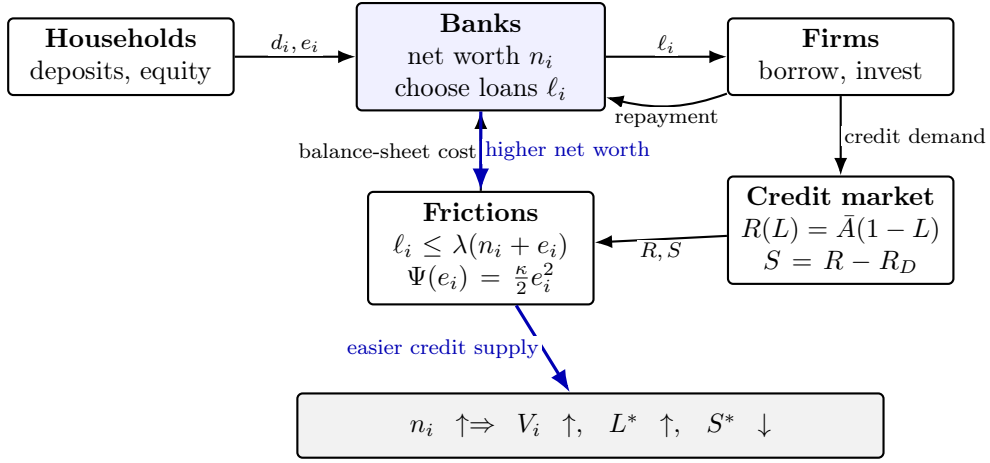
%%%%%%%%%%%%%%%%%%%%%%%%%%%%%%%%%%%%%%%%%%%%%%%%%%%%%%%%%%%%%%%%%%%%%%

\paragraph{Timing.}

At \(t=0\), intermediary \(i\) enters with net worth \(n_i\). Intermediaries raise deposits and, if needed, external equity from households. They then choose loan supply and lend to nonfinancial firms. Firms use these loans to finance investment projects. At \(t=1\), production takes place, firms repay their loans, and intermediaries distribute profits to households. Bank equity values at \(t=0\) therefore reflect the expected present discounted value of intermediary profits.

%%%%%%%%%%%%%%%%%%%%%%%%%%%%%%%%%%%%%%%%%%%%%%%%%%%%%%%%%%%%%%%%%%%%%%
\paragraph{Households.}

Households are risk-neutral and value consumption according to
\[
U=c_0+\beta c_1,
\]
where \(\beta\in(0,1)\). They supply risk-free deposits to intermediaries and own bank equity. Risk neutrality implies that the gross deposit rate is pinned down by
\[
R_D=\frac{1}{\beta}.
\]
This deposit rate is the intermediaries' risk-free funding cost. Since households own intermediary equity, bank equity valuations move with expected intermediary profits.

%%%%%%%%%%%%%%%%%%%%%%%%%%%%%%%%%%%%%%%%%%%%%%%%%%%%%%%%%%%%%%%%%%%%%%
\paragraph{Firms and credit demand.}

Nonfinancial firms require external finance to operate. Each project requires one unit of funding at \(t=0\) and produces output at \(t=1\). Projects differ in quality \(z\), where
\[
z\sim U[0,A].
\]
Only a fraction \(\lambda_f\in(0,1)\) of output can be pledged to lenders. A firm borrowing at gross loan rate \(R\) invests whenever pledgeable output is sufficient to repay the loan:
\[
R\leq \lambda_f z.
\]
Thus, lower loan rates allow more firms to borrow and invest. Aggregating across firms yields a downward-sloping inverse demand for credit:
\[
R(L)=\bar A(1-L),
\qquad
\bar A\equiv \lambda_f A,
\]
where \(L\) denotes aggregate lending. This reduced-form credit demand links changes in intermediary loan supply to equilibrium lending rates and credit spreads.

%%%%%%%%%%%%%%%%%%%%%%%%%%%%%%%%%%%%%%%%%%%%%%%%%%%%%%%%%%%%%%%%%%%%%%
\paragraph{Financial intermediaries.}

Financial intermediaries supply credit to firms. Intermediary \(i\) chooses loan supply \(\ell_i\), and aggregate lending is
\[
L=\sum_{i=1}^{N}\ell_i.
\]
Because the banking sector is concentrated, intermediaries compete as Cournot lenders in the loan market: each intermediary chooses its own loan supply taking competitors' lending as given. This strategic interaction gives individual banks market power and implies that changes in the lending capacity of large banks can affect aggregate credit conditions.

Intermediary \(i\) finances lending with deposits \(d_i\), internal net worth \(n_i\), and newly issued external equity \(e_i\):
\[
\ell_i=d_i+n_i+e_i.
\]
External equity issuance is costly:
\[
\Psi(e_i)=\frac{\kappa}{2}e_i^2,
\qquad
\kappa>0.
\]
Intermediaries are also subject to a leverage constraint:
\[
\ell_i\leq \lambda(n_i+e_i),
\qquad
\lambda>1.
\]
These two frictions make intermediary net worth valuable. Following the quantitative corporate finance literature (e.g., \cite{gomes2001financing}, \cite{hennessy2007costly}), these costs are designed to capture flotation costs and adverse-selection premia associated with raising external equity. A better-capitalized intermediary can support more lending while relying less on costly external equity. Therefore, higher net worth lowers the marginal cost of supplying credit, expands loan supply, and compresses the spread between the lending rate and the deposit rate:
\[
S(L)\equiv R(L)-R_D.
\]

This environment delivers the sign pattern used below. Favorable bank-specific news about intermediary net worth raises the announcing bank's expected profits and equity valuation. At the same time, it relaxes effective balance-sheet constraints, expands aggregate credit supply, and lowers corporate borrowing spreads. The next subsection formalizes this mechanism through the intermediary's optimization problem.

%%%%%%%%%%%%%%%%%%%%%%%%%%%%%%%%%%%%%%%%%%%%%%%%%%%%%%%%%%%%%%%%%%%%%%
\subsection{Bank Problem}

Bank \(i\) chooses loan supply \(\ell_i\) taking the lending decisions of other banks as given. Let
\[
L=\ell_i+L_{-i},
\qquad
L_{-i}\equiv\sum_{j\neq i}\ell_j,
\]
where \(L_{-i}\) denotes total lending by bank \(i\)'s competitors. The gross loan rate is determined by inverse credit demand:
\[
R(L)=\bar A(1-L).
\]
Because banks compete as Cournot lenders, each bank internalizes the effect of its own lending decision on the aggregate lending rate.

Bank \(i\)'s balance sheet is:
\[
\ell_i=d_i+n_i+e_i,
\]
where \(d_i\) denotes deposits, \(n_i\) is initial net worth, and \(e_i\) is newly issued external equity. Deposits are risk-free and pay the gross rate \(R_D\). External equity issuance is costly, with cost
\[
\Psi(e_i)=\frac{\kappa}{2}e_i^2.
\]
The bank is also subject to the leverage constraint:
\[
\ell_i\leq \lambda(n_i+e_i).
\]
We focus on the region in which this constraint binds and banks issue positive external equity. Under a binding leverage constraint,
\[
e_i=\frac{\ell_i}{\lambda}-n_i.
\]
This expression is central to the mechanism: for a given level of lending, a bank with higher net worth needs to issue less costly external equity.

Bank \(i\)'s period-1 profits are:
\[
\Pi_i
=
\left[
R(L)-R_D
\right]\ell_i
+
R_D(n_i+e_i)
-
\frac{\kappa}{2}e_i^2.
\]
The first term is net interest income from lending: the bank earns the lending spread \(R(L)-R_D\) on loans \(\ell_i\). The second term reflects the gross payoff on bank capital, \(n_i+e_i\), and the final term is the cost of raising external equity.

Substituting the binding leverage constraint into profits gives the bank's problem in terms of loan supply:
\[
\max_{\ell_i}
\quad
\Pi_i
=
\left[
R(\ell_i+L_{-i})-R_D
\right]\ell_i
+
R_D
\left(
\frac{\ell_i}{\lambda}
\right)
-
\frac{\kappa}{2}
\left(
\frac{\ell_i}{\lambda}-n_i
\right)^2.
\]
Using \(R(L)=\bar A(1-L)\), this can be written as:
\[
\max_{\ell_i}
\quad
\Pi_i
=
\left[
\bar A(1-\ell_i-L_{-i})-R_D
\right]\ell_i
+
R_D
\left(
\frac{\ell_i}{\lambda}
\right)
-
\frac{\kappa}{2}
\left(
\frac{\ell_i}{\lambda}-n_i
\right)^2.
\]
The first-order condition is:
\[
\frac{\partial \Pi_i}{\partial \ell_i}
=
\bar A(1-\ell_i-L_{-i})
-
\bar A\ell_i
-
R_D
+
\frac{R_D}{\lambda}
-
\frac{\kappa}{\lambda}
\left(
\frac{\ell_i}{\lambda}-n_i
\right)
=0.
\]
Equivalently,
\[
\bar A
-
\bar A L_{-i}
-
2\bar A\ell_i
-
R_D\left(1-\frac{1}{\lambda}\right)
-
\frac{\kappa}{\lambda}
\left(
\frac{\ell_i}{\lambda}-n_i
\right)
=0.
\]

This condition shows how bank net worth affects credit supply. Higher \(n_i\) lowers the amount of costly external equity required to support a given level of lending. This reduces the marginal cost of balance-sheet expansion and increases the bank's optimal loan supply. In equilibrium, the resulting increase in aggregate credit supply lowers the lending rate and compresses the corporate borrowing spread.

%%%%%%%%%%%%%%%%%%%%%%%%%%%%%%%%%%%%%%%%%%%%%%%%%%%%%%%%%%%%%%%%%%%%%%
\subsection{Equilibrium and Comparative Statics}

We first characterize a symmetric Cournot equilibrium in which all banks choose the same loan supply and have the same initial net worth:
\[
\ell_i=\ell,
\qquad
n_i=n,
\]
for all \(i\in\{1,\dots,N\}\). Aggregate lending is therefore
\[
L=N\ell.
\]
In a symmetric equilibrium, bank \(i\)'s competitors supply
\[
L_{-i}=(N-1)\ell.
\]
Substituting this condition into the bank's first-order condition gives
\[
\bar A
-
\bar A(N+1)\ell
-
R_D
\left(
1-\frac{1}{\lambda}
\right)
-
\frac{\kappa}{\lambda}
\left(
\frac{\ell}{\lambda}-n
\right)
=0.
\]
Solving for equilibrium lending per bank yields
\[
\ell^*
=
\frac{
\bar A
-
R_D\left(1-\frac{1}{\lambda}\right)
+
\frac{\kappa}{\lambda}n
}{
\bar A(N+1)
+
\frac{\kappa}{\lambda^2}
}.
\]
Aggregate equilibrium lending is
\[
L^*
=
N\ell^*
=
N
\frac{
\bar A
-
R_D\left(1-\frac{1}{\lambda}\right)
+
\frac{\kappa}{\lambda}n
}{
\bar A(N+1)
+
\frac{\kappa}{\lambda^2}
}.
\]
The equilibrium lending rate is pinned down by inverse credit demand:
\[
R^*
=
\bar A(1-L^*),
\]
and the equilibrium corporate borrowing spread is
\[
S^*
=
R^*-R_D.
\]

We focus on the equilibrium region in which lending is positive, banks issue positive external equity, leverage constraints bind, and equilibrium lending spreads are positive. Positive lending requires
\[
\bar A
-
R_D\left(1-\frac{1}{\lambda}\right)
+
\frac{\kappa}{\lambda}n
>0,
\]
while positive equilibrium spreads require
\[
S^*
=
\bar A(1-L^*)-R_D
>0.
\]
Appendix \ref{appendix:proofs} describes the maintained equilibrium region in more detail.

The first comparative static concerns a common increase in bank net worth. Differentiating equilibrium lending per bank with respect to the common net-worth level \(n\) gives
\[
\frac{\partial \ell^*}{\partial n}
=
\frac{
\kappa/\lambda
}{
\bar A(N+1)+\kappa/\lambda^2
}
>0.
\]
Therefore,
\[
\frac{\partial L^*}{\partial n}
=
N
\frac{\partial \ell^*}{\partial n}
>0.
\]
Higher common bank net worth increases equilibrium lending by reducing the amount of costly external equity needed to support a given balance-sheet size.

Since the equilibrium lending rate is decreasing in aggregate lending,
\[
R^*=\bar A(1-L^*),
\]
we have
\[
\frac{\partial R^*}{\partial n}
=
-\bar A
\frac{\partial L^*}{\partial n}
<0.
\]
The deposit rate \(R_D\) is fixed by households' Euler equation, so
\[
\frac{\partial S^*}{\partial n}
=
\frac{\partial R^*}{\partial n}
<0.
\]
A common increase in bank net worth therefore expands credit supply and compresses corporate borrowing spreads.

The second comparative static concerns a bank-specific increase in net worth. This is the perturbation that most closely corresponds to the empirical setting, where earnings announcements reveal news about a particular bank. Consider a marginal increase in \(n_i\), evaluated around the symmetric equilibrium. Define
\[
a
=
\bar A+\frac{\kappa}{\lambda^2},
\qquad
D
=
\bar A(N+1)+\frac{\kappa}{\lambda^2}.
\]
Solving the linear Cournot system around the symmetric equilibrium gives
\[
\frac{\partial L^*}{\partial n_i}
=
\frac{\kappa/\lambda}{D}
>0.
\]
Thus, a bank-specific increase in net worth raises aggregate lending. Since the equilibrium lending rate is \(R^*=\bar A(1-L^*)\), it follows that
\[
\frac{\partial S^*}{\partial n_i}
=
\frac{\partial R^*}{\partial n_i}
=
-\bar A
\frac{\partial L^*}{\partial n_i}
<0.
\]
A favorable bank-specific net-worth shock therefore expands aggregate credit supply and compresses corporate borrowing spreads.

The valuation effect is also positive. Optimized profits of bank \(i\) are
\[
\Pi_i^*
=
\left[
R(L^*)-R_D
\right]\ell_i^*
+
R_D
\left(
\frac{\ell_i^*}{\lambda}
\right)
-
\frac{\kappa}{2}
\left(
\frac{\ell_i^*}{\lambda}-n_i
\right)^2.
\]
For a bank-specific perturbation, the envelope theorem eliminates the effect of bank \(i\)'s own optimal lending response, but competitors' lending can also respond. The total derivative is
\[
\frac{d\Pi_i^*}{dn_i}
=
\kappa e_i^*
-
\bar A \ell_i^*
\frac{\partial L_{-i}^*}{\partial n_i},
\]
where \(L_{-i}^*\) is total lending by bank \(i\)'s competitors. In the Cournot equilibrium, loans are strategic substitutes, so
\[
\frac{\partial L_{-i}^*}{\partial n_i}
=
-(N-1)
\frac{(\kappa/\lambda)\bar A}{aD}
<0.
\]
Therefore,
\[
\frac{d\Pi_i^*}{dn_i}
=
\kappa e_i^*
+
\bar A \ell_i^*(N-1)
\frac{(\kappa/\lambda)\bar A}{aD}
>0,
\]
where the inequality uses \(e_i^*>0\) in the maintained equilibrium region. Bank equity values reflect the present discounted value of expected bank profits:
\[
V_i=\beta \mathbb E[\Pi_i^*],
\]
so \(dV_i/dn_i>0\) under the maintained equilibrium conditions.

These results deliver the model's key sign implication. Favorable bank-specific news about intermediary net worth raises the announcing bank's equity valuation, expands aggregate credit supply, and lowers corporate borrowing spreads:
\[
n_i \uparrow
\quad\Rightarrow\quad
V_i \uparrow,
\qquad
L^* \uparrow,
\qquad
S^* \downarrow.
\]
This joint movement of bank equity prices and corporate spreads motivates the empirical sign restriction used below.

%%%%%%%%%%%%%%%%%%%%%%%%%%%%%%%%%%%%%%%%%%%%%%%%%%%%%%%%%%%%%%%%%%%%%%
\subsection{Implications for Identification}

The model provides a direct interpretation for stock-price reactions around bank earnings announcements. In the model, bank equity values reflect expected future profits. Because profits depend on lending capacity, balance-sheet costs, and equilibrium lending spreads, news about bank net worth is reflected in bank stock prices.

The model also clarifies why raw bank equity surprises are not automatically clean credit-supply shocks. Earnings announcements may reveal favorable information about a bank's own balance-sheet strength and lending capacity, but they may also contain information about aggregate fundamentals, borrower quality, expected credit demand, or risk premia. These sources of news can all move bank valuations, but they do not have the same implications for credit supply.

The sign implication of the model is the joint response of bank equity prices and corporate borrowing spreads. A favorable bank-specific credit-supply shock raises the announcing bank's net worth and equity valuation. At the same time, it relaxes effective balance-sheet constraints, expands aggregate credit supply, and lowers corporate spreads. Therefore, the component of bank equity news most closely aligned with the model should move bank valuations and corporate spreads in opposite directions:
\[
\Delta V_i>0,
\qquad
\Delta S<0.
\]
Equivalently, favorable credit-supply news should be associated with positive bank stock-price reactions and declines in the Canadian corporate OAS.

Other information released during earnings announcements can also move bank valuations. News about stronger expected loan demand, for example, may raise expected bank profitability while putting upward pressure on borrowing costs. Such news moves bank equity prices and corporate spreads in the same direction and is therefore not the credit-supply object emphasized by the model. Other forms of news, such as risk-premium or aggregate-fundamentals news, may be harder to distinguish using asset-price signs alone. For this reason, the sign pattern should be interpreted as an identifying restriction motivated by the model, not as a mechanical classification rule.

This distinction motivates the empirical strategy in the next section. We first construct raw bank equity surprises from timing-adjusted stock-price reactions around Canadian bank earnings announcements. We then use the co-movement between these bank equity surprises and Canadian corporate spreads to isolate the component most closely aligned with the model's prediction: favorable bank net-worth news raises bank valuations and lowers corporate borrowing spreads.

%%%%%%%%%%%%%%%%%%%%%%%%%%%%%%%%%%%%%%%%%%%%%%%%%%%%%%%%%%%%%%%%%%%%%%

%%%%%%%%%%%%%%%%%%%%%%%%%%%%%%%%%%%%%%%%%%%%%%%%%%%%%%%%%%%%%%%%%
%%%%%%%%%%%%%%%%%%%%%%%%%%%%%%%%%%%%%%%%%%%%%%%%%%%%%%%%%%%%%%%%%
\section{Data \& Empirical Strategy} \label{sec:shock_construction}

This section describes the construction of Canadian bank net-worth shocks from earnings-announcement stock-price reactions. The empirical strategy uses Bloomberg data on bank equity prices, earnings-announcement events, market capitalization, and Canadian corporate credit spreads. We focus on the six large publicly traded Canadian banks: Royal Bank of Canada, Toronto-Dominion Bank, Bank of Montreal, Bank of Nova Scotia, Canadian Imperial Bank of Commerce, and National Bank of Canada. These institutions account for the bulk of the publicly traded Canadian banking sector's market value and closely correspond to the large intermediaries in the theoretical framework of Section \ref{sec:illustrative_model}.

The sample runs from 2002:Q4 to 2026:Q1. The starting date is determined by the availability of the Bloomberg Canadian corporate option-adjusted spread series, which we use to measure aggregate corporate credit conditions and to purge raw bank equity surprises that could contaminate the information. The construction proceeds in three steps. First, we document the concentrated structure of the Canadian banking sector and validate that earnings announcements are information events for bank equity prices. Second, we construct timing-adjusted stock-price reactions around each bank earnings announcement and aggregate them using lagged market-capitalization shares. Third, guided by the model's prediction that favorable credit-supply news raises bank valuations and lowers corporate spreads, we use the co-movement between bank equity reactions and Canadian corporate OAS changes to isolate the purged credit-supply component of bank net-worth news.

%%%%%%%%%%%%%%%%%%%%%%%%%%%%%%%%%%%%%%%%%%%%%%%%%%%%%%%%%%%%%%%%%%%%%%
\paragraph{Market capitalization shares.}

A key feature of the Canadian banking system is its high degree of concentration. Figure \ref{fig:bank_market_cap_shares} plots each bank's share of total market capitalization among the six large Canadian banks from 2002:Q4 to 2026:Q1. The figure shows that the market value of the publicly traded Canadian banking sector is persistently concentrated among a small number of institutions. RBC and TD account for the largest shares throughout most of the sample, while BMO, Scotiabank, CIBC, and National Bank account for the remainder. This concentration is central for the empirical design: an earnings surprise for a large bank represents news about a quantitatively important part of the Canadian intermediation sector.

%%%%%%%%%%%%%%%%%%%%%%%%%%%%%%%%%%%%%%%%%%%%%%%%%%%%%%%%%%%%%%%%%%%%%%
\begin{figure}[H]
\centering
\includegraphics[width=0.90\textwidth]{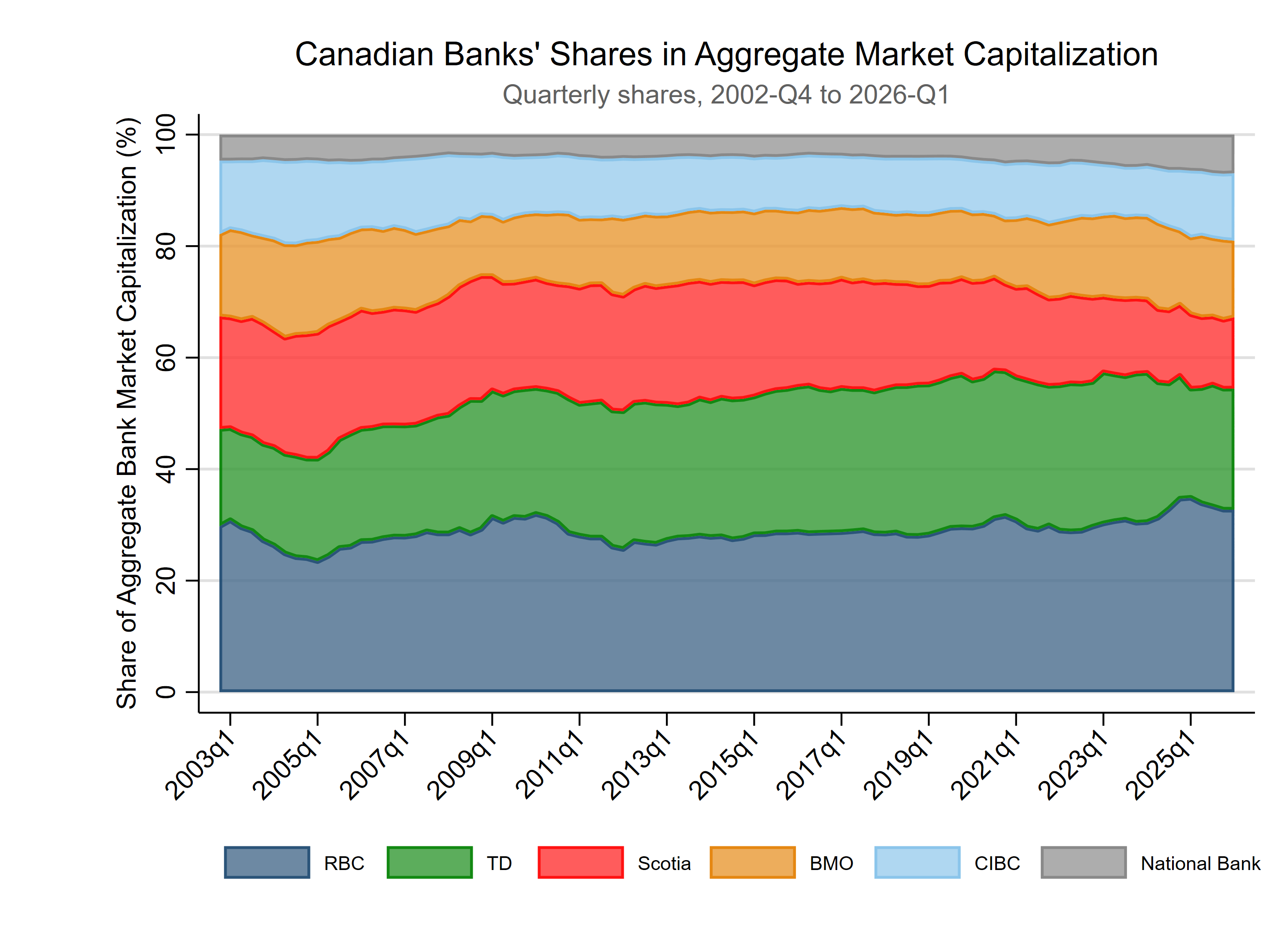}
\caption{Canadian Banks' Shares in Aggregate Market Capitalization}
\label{fig:bank_market_cap_shares}
\floatfoot{\textbf{Notes:} This figure plots quarterly market capitalization shares for the six large Canadian banks in the sample. Shares are computed as each bank's market capitalization divided by total market capitalization across RBC, TD, Scotiabank, BMO, CIBC, and National Bank.}
\end{figure}
%%%%%%%%%%%%%%%%%%%%%%%%%%%%%%%%%%%%%%%%%%%%%%%%%%%%%%%%%%%%%%%%%%%%%%

At this stage, the market-capitalization shares are used only to document the structure of the banking sector and motivate the aggregation procedure introduced below. The descriptive exercises that follow first validate that earnings announcements are information events for bank equity prices. We then use lagged market-capitalization shares to aggregate timing-adjusted bank-level stock-price reactions into a raw Canadian bank equity surprise.

%%%%%%%%%%%%%%%%%%%%%%%%%%%%%%%%%%%%%%%%%%%%%%%%%%%%%%%%%%%%%%%%%%%%%%
\paragraph{Earnings announcements and event windows.}

We collect earnings-announcement events from Bloomberg for the Canadian-listed tickers of the six large Canadian banks. We focus on events classified by Bloomberg as earnings releases, corresponding to event type ``ER.'' For each event, we retain the earnings-release date and, when available, the associated timestamp. When Bloomberg reports an earnings-release date but no announcement time, we classify the event as occurring before the market opens.

We distinguish between two return windows. For the descriptive comparison between announcement and non-announcement days, we use the same close-to-close stock-price change for both samples:
\[
\Delta p^{CC}_{i,t}
=
100\left[
\log(P^{close}_{i,t})
-
\log(P^{close}_{i,t-1})
\right].
\]
This same-window comparison allows us to test whether earnings-announcement days are associated with larger bank equity-price movements than other trading days without mechanically changing the length of the return window.

For the construction of bank equity surprises, however, we use a timing-adjusted event window around each earnings announcement. This is the relevant object for identification because it isolates the stock-price movement over the window in which investors first incorporate the earnings news. For pre-open announcements, we compute the change from the previous trading day's close to the announcement-day open. For intraday announcements, we compute the change from the announcement-day open to the announcement-day close. For post-close announcements, we compute the change from the announcement-day close to the next trading day's open. We denote this timing-adjusted event-window stock-price reaction by
\[
\Delta p^{ER}_{i,t}.
\]

We use the Bloomberg Canadian corporate option-adjusted spread series to measure aggregate corporate credit conditions. Let \(OAS_t\) denote the Canadian corporate option-adjusted spread. Since the OAS series is observed at the daily frequency, OAS changes cannot be measured over the same intraday windows as stock prices. We therefore construct OAS changes using the closest daily window implied by the earnings-release timing. For pre-open and intraday announcements, we measure the OAS change from the previous trading day to the announcement day. For post-close announcements, we measure the OAS change from the announcement day to the next trading day.

%%%%%%%%%%%%%%%%%%%%%%%%%%%%%%%%%%%%%%%%%%%%%%%%%%%%%%%%%%%%%%%%%%%%%%
\paragraph{Announcement days as information events.}

We next verify that earnings-announcement days are information events for Canadian bank equity prices. Because the empirical strategy below studies the co-movement between bank equity reactions and corporate spread changes, we restrict the analysis to a common sample in which both the stock-price change and the OAS change are observed. Table \ref{tab:announcement_price_summary_same_window} compares stock-price movements on announcement and non-announcement days using same-window close-to-close stock returns. The table reports the raw close-to-close price change, its absolute value, and its square. The raw return captures the direction of the stock-price movement, while the absolute and squared returns measure the magnitude of the price movement regardless of sign.

%%%%%%%%%%%%%%%%%%%%%%%%%%%%%%%%%%%%%%%%%%%%%%%%%%%%%%%%%%%%%%%%%%%%%%
\begin{table}[H]
\centering
\caption{Bank Stock-Price Movements Around Canadian Bank Earnings Announcements}
\label{tab:announcement_price_summary_same_window}
\begin{threeparttable}
\scriptsize
\resizebox{\textwidth}{!}{
\begin{tabular}{lccccccccc}
\toprule
& \multicolumn{3}{c}{Non-Announcement Days} 
& \multicolumn{3}{c}{Announcement Days}
& \multicolumn{3}{c}{Tests -- $p$-values} \\
\cmidrule(lr){2-4} \cmidrule(lr){5-7} \cmidrule(lr){8-10}
Variable 
& $N$ & Mean & Median 
& $N$ & Mean & Median
& $t$-test 
& Variance-test 
& Median-test \\
\midrule

Close-to-close price change
& 33,513 & 0.031 & 0.069
& 543 & 0.012 & -0.099
& 0.750 & 0.000 & 0.181 \\

Absolute close-to-close price change
& 33,513 & 0.839 & 0.563
& 543 & 2.079 & 1.784
& 0.000 & 0.000 & 0.000 \\

Squared close-to-close price change
& 33,513 & 1.747 & 0.317
& 543 & 7.009 & 3.184
& 0.000 & 0.000 & 0.000 \\

\bottomrule
\end{tabular}
}
\footnotesize
\floatfoot{
\textit{Notes:} This table reports pooled summary statistics for bank stock-price movements on Canadian bank earnings-announcement and non-announcement days. For this descriptive comparison, stock-price changes are measured using the same close-to-close window for both groups, $\Delta p^{CC}_{i,t}=100[\log(P^{close}_{i,t})-\log(P^{close}_{i,t-1})]$. The sample is restricted to observations for which both the stock-price change and the Canadian corporate OAS change are observed. The final three columns report $p$-values from tests comparing announcement and non-announcement days: a two-sample $t$-test for equality of means, a variance test, and a Wilcoxon rank-sum test for differences in medians. The high-frequency shocks used in the empirical analysis are not constructed from close-to-close returns; they are constructed using timing-adjusted event-window stock-price reactions around earnings announcements.}
\end{threeparttable}
\end{table}
%%%%%%%%%%%%%%%%%%%%%%%%%%%%%%%%%%%%%%%%%%%%%%%%%%%%%%%%%%%%%%%%%%%%%%

Table \ref{tab:announcement_price_summary_same_window} shows that earnings announcements are associated with substantially larger movements in bank equity prices. The mean and median of raw close-to-close returns are not statistically different between announcement and non-announcement days, as expected, because earnings news can be either positive or negative. However, the variance test strongly rejects the null of equal variances for raw close-to-close returns.

The magnitude of price movements is much larger on announcement days. The average absolute close-to-close price change rises from 0.839 percentage points on non-announcement days to 2.079 percentage points on announcement days, while the median rises from 0.563 to 1.784 percentage points. Squared price changes display the same pattern: the average squared return rises from 1.747 to 7.009. For both absolute and squared price changes, the equality-of-means, equality-of-variances, and equality-of-medians tests all reject at conventional levels. These results confirm that earnings announcements release information that is incorporated into Canadian bank equity prices.

This table is used only to validate that earnings announcements are information events. The shocks used in the empirical analysis are constructed from the timing-adjusted event-window stock-price reaction \(\Delta p^{ER}_{i,t}\), not from close-to-close returns. The distinction is important: close-to-close returns provide a comparable benchmark across announcement and non-announcement days, while timing-adjusted event-window returns isolate the high-frequency market reaction to the earnings release.

%%%%%%%%%%%%%%%%%%%%%%%%%%%%%%%%%%%%%%%%%%%%%%%%%%%%%%%%%%%%%%%%%%%%%%
\paragraph{Validating stock-price reactions with earnings surprises.}

Before aggregating bank-level surprises, we verify that the timing-adjusted event-window stock-price changes capture the information released in banks' earnings announcements. Bloomberg reports an earnings surprise for each announcement, constructed as the difference between realized earnings and the pre-announcement analyst consensus estimate. If the event-window stock-price reaction captures news about bank profitability, then positive Bloomberg earnings surprises should be associated with positive stock-price reactions.

Figure \ref{fig:scatter_price_surprise_pooled} shows this relationship in the pooled sample. Larger positive Bloomberg earnings surprises are associated with larger positive announcement-window stock-price reactions. This positive relationship supports the interpretation of \(\Delta p^{ER}_{i,t}\) as a bank-level equity-market response to earnings news rather than as an arbitrary daily price movement. This positive relationship supports the interpretation of \(\Delta p^{ER}_{i,t}\) as an equity-market response to earnings news rather than as an arbitrary daily price movement. Separating the component of this earnings news that corresponds to credit-supply variation is the purpose of the sign-restriction step below.

%%%%%%%%%%%%%%%%%%%%%%%%%%%%%%%%%%%%%%%%%%%%%%%%%%%%%%%%%%%%%%%%%%%%%%
\begin{figure}[H]
\centering
\includegraphics[width=0.90\textwidth]{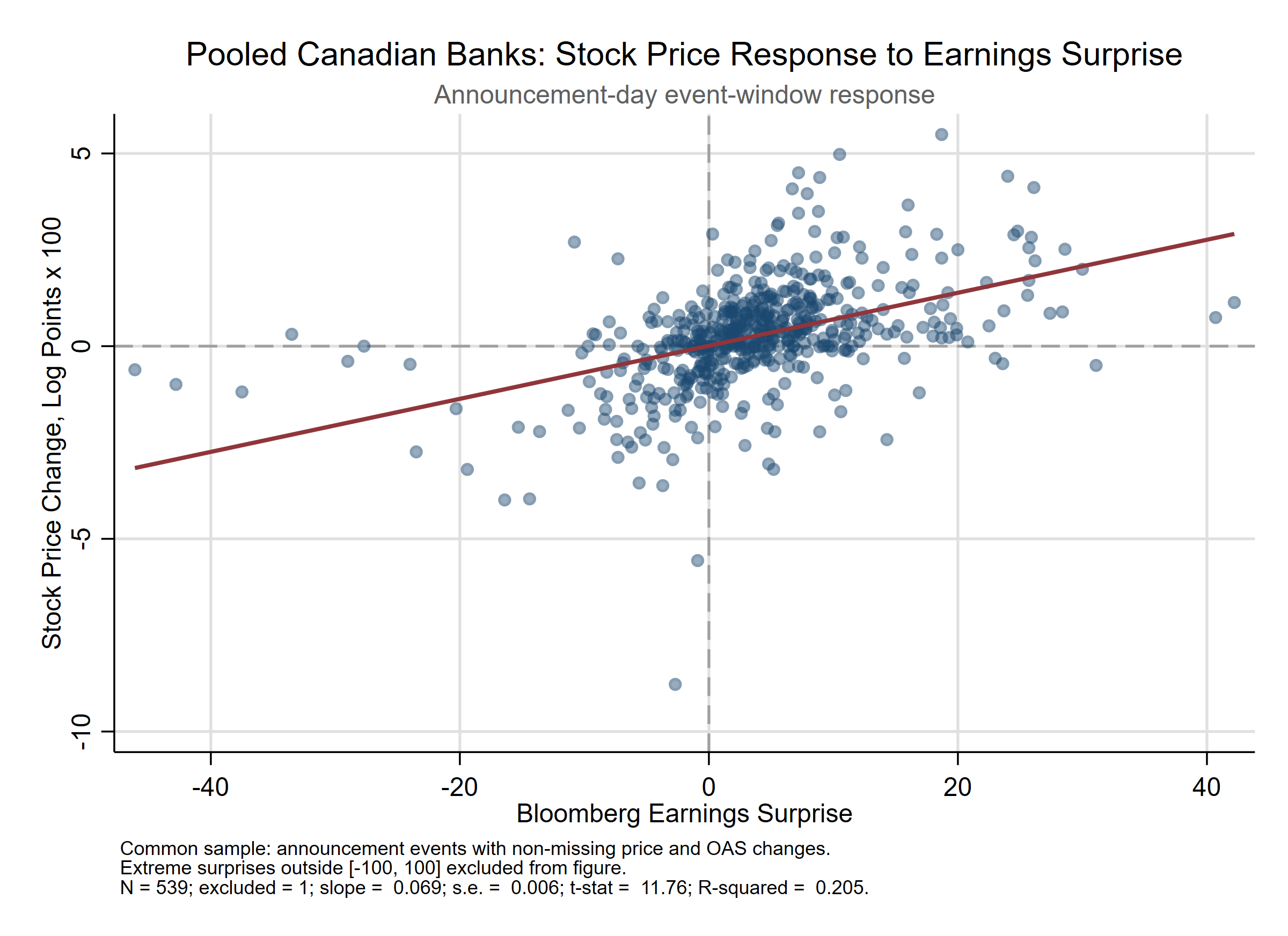}
\caption{Stock-Price Responses to Bloomberg Earnings Surprises: Pooled Canadian Banks}
\label{fig:scatter_price_surprise_pooled}
\floatfoot{\textbf{Notes:} This figure plots timing-adjusted announcement-window stock-price changes against Bloomberg earnings surprises for the pooled sample of Canadian bank earnings announcements. Stock-price changes are log price changes multiplied by 100. Bloomberg earnings surprises are constructed as the difference between realized earnings and the pre-announcement analyst consensus estimate. The sample is restricted to announcement events with non-missing stock-price and OAS changes. Extreme surprises outside $[-100,100]$ are excluded from the figure. The fitted line corresponds to an OLS regression of the event-window stock-price change on the Bloomberg earnings surprise.}
\end{figure}
%%%%%%%%%%%%%%%%%%%%%%%%%%%%%%%%%%%%%%%%%%%%%%%%%%%%%%%%%%%%%%%%%%%%%%

The pooled relationship in Figure \ref{fig:scatter_price_surprise_pooled} summarizes the main validation exercise. Appendix \ref{appendix:additional_details_surprises} presents the corresponding bank-by-bank scatter plots and shows that the positive relationship between Bloomberg earnings surprises and announcement-window stock-price reactions is not driven by a single institution. This evidence motivates using timing-adjusted announcement-window stock-price changes as the bank-level surprises from which we construct the aggregate Canadian bank equity shock.

%%%%%%%%%%%%%%%%%%%%%%%%%%%%%%%%%%%%%%%%%%%%%%%%%%%%%%%%%%%%%%%%%%%%%%
\paragraph{Aggregating bank-level surprises.}

Having established that timing-adjusted announcement-window stock-price reactions capture earnings news, we aggregate bank-level surprises into a Canadian bank equity surprise. The aggregation uses lagged market-capitalization shares, motivated by the concentration of the Canadian banking sector documented in Figure \ref{fig:bank_market_cap_shares}. An earnings surprise for a larger bank represents news about a larger share of the Canadian intermediation sector and should therefore receive more weight in the aggregate shock.

Let \(MC_{i,q}\) denote the market capitalization of bank \(i\) in quarter \(q\), and define total market capitalization across the six banks as
\[
MC_q
=
\sum_{j=1}^{N}MC_{j,q}.
\]
Bank \(i\)'s market-capitalization share is
\[
\omega_{i,q}
=
\frac{MC_{i,q}}{\sum_{j=1}^{N}MC_{j,q}}.
\]
For an earnings announcement occurring on date \(t\), let \(q(t)\) denote the corresponding quarter. The raw aggregate Canadian bank equity surprise is
\[
v_t
=
\sum_{i=1}^{N}
\omega_{i,q(t)-1}
\Delta p^{ER}_{i,t}.
\]
The weights are lagged by one quarter, so they are predetermined relative to the earnings announcement and are not mechanically affected by the announcement-window stock-price reaction. Positive values of \(v_t\) correspond to favorable news about the market value of the Canadian banking sector, while negative values correspond to adverse bank equity news.

Before using \(v_t\) for identification, we also examine whether the underlying bank-level announcement-window stock-price reactions are predictable from previous earnings-announcement reactions. Appendix \ref{appendix:autocorrelation_surprises} reports bank-level and pooled serial-correlation tests. The evidence provides little indication of systematic persistence: event-time autocorrelations are generally small and centered near zero. This supports the interpretation of the announcement-window stock-price reactions as high-frequency surprises rather than predictable components of bank equity returns.

%%%%%%%%%%%%%%%%%%%%%%%%%%%%%%%%%%%%%%%%%%%%%%%%%%%%%%%%%%%%%%%%%%%%%%
\paragraph{Purging credit-supply bank net-worth shocks.}

The raw bank equity surprise \(v_t\) is informative, but it is not
automatically a clean credit-supply shock. Earnings announcements may
reveal information about banks' own balance-sheet strength and lending
capacity, but they may also contain other information about expected loan
demand, borrower fundamentals, aggregate macroeconomic conditions, or risk
premia. The model in Section \ref{sec:illustrative_model} provides the sign
pattern that motivates the empirical decomposition: favorable intermediary
net-worth news should raise bank equity valuations and lower corporate
borrowing spreads. We use this implication as an identifying restriction
rather than as a mechanical classification rule.

We first examine this sign pattern directly in the data. Figure
\ref{fig:scatter_price_oas_pooled} plots event-window changes in the
Canadian corporate OAS against timing-adjusted event-window bank stock-price
reactions for the pooled sample of Canadian bank earnings announcements. The
fitted relationship is negative: higher bank equity-price reactions are
associated with lower corporate spreads. Quantitatively, the pooled slope is
\(-0.123\), with a standard error of \(0.077\), implying a \(t\)-statistic of
\(-1.59\) and a \(p\)-value of \(0.112\). The correlation is also negative,
equal to \(-0.069\). Thus, the sign of the relationship is consistent with
the intermediary balance-sheet channel, although the contemporaneous
event-window relationship is imprecisely estimated.

%%%%%%%%%%%%%%%%%%%%%%%%%%%%%%%%%%%%%%%%%%%%%%%%%%%%%%%%%%%%%%%%%%%%%%
\begin{figure}[H]
\centering
\includegraphics[width=0.90\textwidth]{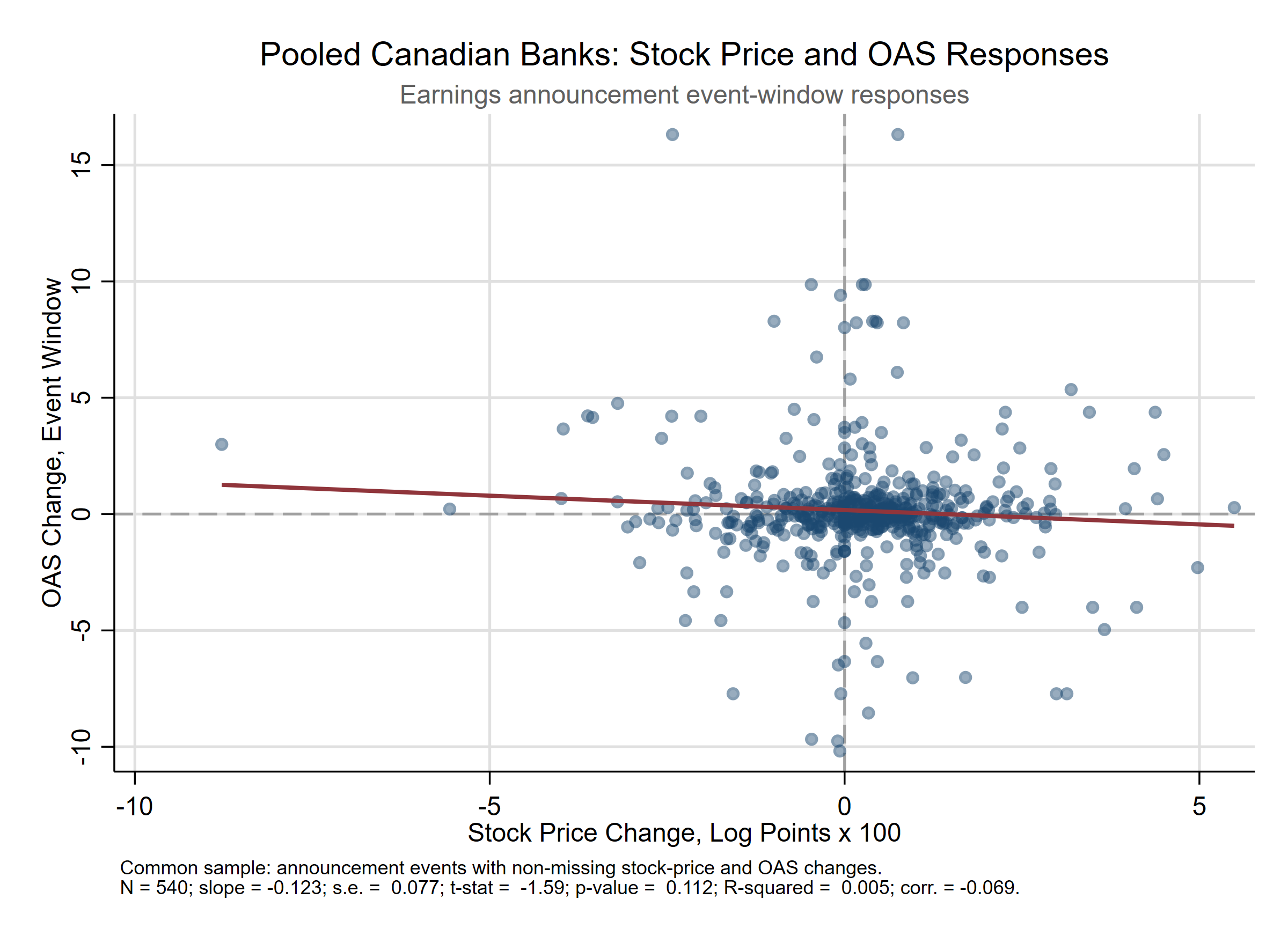}
\caption{Bank Stock-Price Reactions and Canadian Corporate OAS Changes}
\label{fig:scatter_price_oas_pooled}
\floatfoot{\textbf{Notes:} This figure plots event-window changes in Canadian corporate OAS against timing-adjusted event-window bank stock-price changes for the pooled sample of Canadian bank earnings announcements. Stock-price changes are log price changes multiplied by 100. OAS changes are measured in basis points. The sample is restricted to announcement events with non-missing stock-price and OAS changes. The fitted line corresponds to an OLS regression of the OAS change on the event-window stock-price change.}
\end{figure}
%%%%%%%%%%%%%%%%%%%%%%%%%%%%%%%%%%%%%%%%%%%%%%%%%%%%%%%%%%%%%%%%%%%%%%

We interpret Figure \ref{fig:scatter_price_oas_pooled} as sign-validating
evidence rather than as a stand-alone test of the credit-supply channel.
Bank stock prices are bank-specific and respond immediately to earnings
news, while the Canadian corporate OAS is an aggregate credit-market object
that may adjust more slowly and is affected by broader financial conditions.
Appendix \ref{appendix:additional_details_surprises} reports the
corresponding bank-by-bank scatter plots. The estimated slopes are negative
for each bank, although the individual-bank relationships are imprecisely
estimated.

Guided by this sign pattern, we decompose the raw bank equity surprise into
two orthogonal components:
\[
v_t
=
v^{CS}_t
+
v^{C}_t,
\]
where \(v^{CS}_t\) denotes the purged credit-supply component and \(v^{C}_t\)
denotes the residual component. The object of interest is \(v^{CS}_t\), the
component of bank equity news aligned with the model's prediction that
favorable intermediary net-worth news raises bank valuations and lowers
corporate spreads. The residual component \(v^{C}_t\) absorbs the remaining
variation in bank equity surprises, including news about expected loan
demand, borrower fundamentals, aggregate conditions, or risk premia.

The decomposition is sign-identified rather than point-identified. To make
the rotation explicit, define the bivariate event-window vector
\[
z_t
=
\begin{pmatrix}
v_t \\
\Delta OAS_t
\end{pmatrix},
\]
where \(v_t\) is the raw market-capitalization-weighted bank equity surprise
and \(\Delta OAS_t\) is the corresponding event-window change in the Canadian
corporate OAS. We orthogonalize the innovations in \(z_t\) and consider the
set of two-dimensional rotations that generate two orthogonal components,
\(v^{CS}_t\) and \(v^C_t\), satisfying
\[
v_t
=
v^{CS}_t+v^{C}_t,
\qquad
\operatorname{cov}\left(v^{CS}_t,v^{C}_t\right)=0.
\]
Given these two conditions, variance additivity follows mechanically:
\[
\operatorname{var}(v_t)
=
\operatorname{var}(v^{CS}_t)
+
\operatorname{var}(v^{C}_t).
\]
The identifying sign restrictions are
\[
\operatorname{cov}\left(v^{CS}_t,\Delta OAS_t\right)<0,
\qquad
\operatorname{cov}\left(v^{C}_t,\Delta OAS_t\right)>0.
\]
The first restriction assigns to the credit-supply component the part of bank
equity news that moves bank valuations and corporate spreads in opposite
directions. The second restriction assigns to the residual component the
variation that moves bank valuations and corporate spreads in the same
direction. These restrictions do not rule out all alternative interpretations
of negative equity--spread co-movement, such as favorable risk-premium news.
They should therefore be interpreted as identifying assumptions that isolate
the component of earnings-window bank equity news most closely aligned with
the intermediary net-worth channel in the model.

Our baseline implementation uses the median admissible rotational angle
proposed by \cite{jarocinski2022central}. We consider the set of rotations
that satisfy the sign restrictions above and select the median rotation among
the admissible rotations. In our sample, the admissible set is nonempty and
ranges from \(0\) to \(86.8\) degrees. The median rotation is \(43.4\)
degrees. At this rotation, the credit-supply component has correlation
\(0.73\) with the raw bank equity surprise, correlation \(-0.73\) with the
event-window OAS change, and correlation \(0.70\) with the credit-supply shock
obtained from the poor man's sign restriction. Thus, the benchmark shock uses
both pieces of high-frequency information: it is neither simply the raw bank
equity surprise nor simply the negative of the OAS innovation.

As an alternative implementation, we also construct a simple ``poor man's''
sign-restriction shock. In that approach, announcement observations in which
bank stock prices and OAS changes move in opposite directions are classified
as credit-supply news, while observations in which they move in the same
direction are assigned to the residual component. The empirical analysis below
focuses on the median-rotation credit-supply component \(v^{CS}_t\), and
Section \ref{subsec:additional_evidence} compares the benchmark responses with
those obtained from the raw bank equity surprise, the poor man's
sign-restriction shock, and alternative admissible rotations.

%%%%%%%%%%%%%%%%%%%%%%%%%%%%%%%%%%%%%%%%%%%%%%%%%%%%%%%%%%%%%%%%%%%%%%
\paragraph{Time series of purged shocks.}

Figure \ref{fig:bank_supply_shocks_timeseries} presents the monthly time series of the raw bank equity surprise and the two purged credit-supply shock series. The black line reports the raw market-capitalization-weighted bank equity surprise \(v_t\). The light-blue dashed line reports the credit-supply shock obtained from the poor man's sign restriction, while the dark-blue dotted line reports the credit-supply shock obtained from the median rotational-angle procedure.\footnote{The three series are closely related but need not coincide at the monthly frequency. The raw surprise sums all bank-level earnings-announcement surprises within a month, while the poor man's supply measure keeps only the event-level surprises whose stock-price and OAS changes move in opposite directions. Thus, in months with both supply-type and non-supply-type announcement events, the raw monthly surprise and the poor man's supply shock differ. The median rotational-angle series differs more generally because it is obtained from a rotation of the joint system of bank equity surprises and OAS changes.}

The shocks are centered around zero and display sharp movements around specific earnings-announcement months. This pattern is consistent with the high-frequency nature of the identification strategy: the series captures discrete news about bank profitability and balance-sheet strength rather than slow-moving changes in financial conditions. Several large movements occur during periods of elevated financial uncertainty, including the global financial crisis, the early 2010s, the COVID period, and the most recent part of the sample. In the empirical analysis that follows, we use the median rotational-angle purged credit-supply shock as the benchmark measure of Canadian bank net-worth shocks.

%%%%%%%%%%%%%%%%%%%%%%%%%%%%%%%%%%%%%%%%%%%%%%%%%%%%%%%%%%%%%%%%%%%%%%
\begin{figure}[H]
\centering
\includegraphics[width=0.95\textwidth]{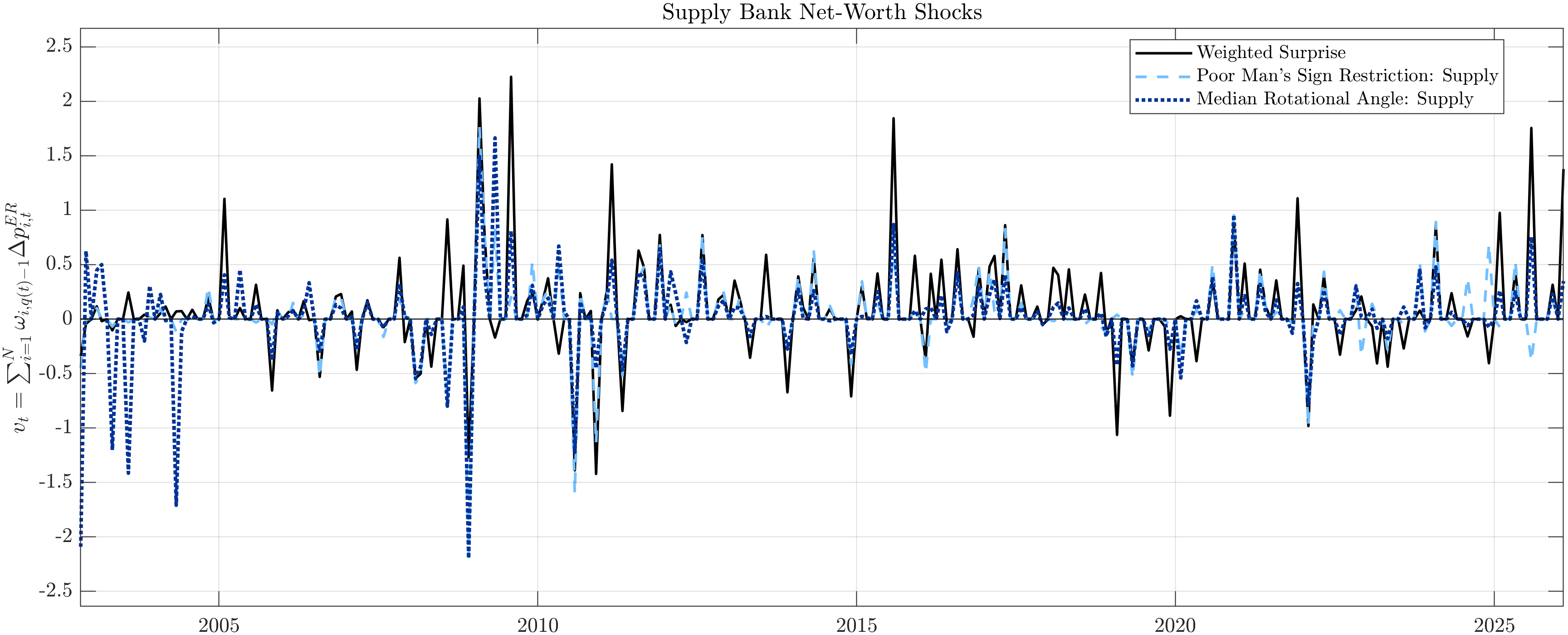}
\caption{Raw and Purged Canadian Bank Net-Worth Supply Shocks}
\label{fig:bank_supply_shocks_timeseries}
\floatfoot{\textbf{Notes:} This figure reports monthly Canadian bank net-worth shock series constructed from earnings-announcement stock-price reactions. The black line is the raw market-capitalization-weighted bank equity surprise. The light-blue dashed line reports the purged credit-supply component identified using the poor man's sign restriction. The dark-blue dotted line reports the purged credit-supply component identified using the median rotational-angle procedure. Monthly series are constructed by summing event-level shocks within each month. The empirical analysis focuses on the median rotational-angle purged credit-supply component, which is the object directly motivated by the theoretical framework in Section \ref{sec:illustrative_model}.}
\end{figure}

\paragraph{Correlation with \cite{ottonello2025financial} shocks.}

As an additional validation exercise, we test whether the Canadian bank net-worth shocks are predictable from U.S. bank surprises observed before the Canadian earnings announcements, constructed by \cite{ottonello2025financial}. This exercise matters because U.S. bank earnings announcements may reveal information about global banking-sector conditions, intermediary balance sheets, or aggregate financial risk. If the Canadian shocks were systematically predictable from previously observed U.S. bank surprises, then the Canadian announcement-window variation could partly reflect information already revealed by U.S. bank earnings news rather than new information contained in Canadian bank announcements.

For each Canadian bank-shock date, we attach the last three U.S. bank surprises observed strictly before the Canadian announcement date and regress the Canadian shock on these predetermined U.S. shocks. Appendix \ref{appendix:predictability_us_bank_shocks} reports the results. The raw total Canadian bank-equity surprise displays some mild predictability from the most recent raw U.S. bank surprise: the joint test for the last three raw U.S. surprises has a \(p\)-value of \(0.082\). This is consistent with the idea that the raw Canadian bank-equity surprise may contain a common banking-sector or global financial component. However, the relationship disappears for the preferred median-rotation credit-supply shock. When the Canadian credit-supply shock is regressed on the last three raw U.S. bank surprises, the joint \(p\)-value is \(0.292\); when it is regressed on the last three purged U.S. bank surprises, the joint \(p\)-value is \(0.818\). Thus, the preferred Canadian credit-supply bank net-worth shock is not predictable from U.S. bank earnings-news shocks observed before the Canadian earnings announcement. This supports the interpretation that the benchmark shock captures Canadian bank-specific earnings news rather than pre-existing U.S. banking-sector information.

%%%%%%%%%%%%%%%%%%%%%%%%%%%%%%%%%%%%%%%%%%%%%%%%%%%%%%%%%%%%%%%%%%%%%%

% -------------------------------------------------
% Macroeconomic Propagation
% -------------------------------------------------
\section{Macroeconomic Propagation} \label{sec:Macroeconomic_Propagation}

This section studies the macroeconomic propagation of the purged credit-supply bank net-worth shocks constructed in Section \ref{sec:shock_construction}. The previous section shows that Canadian bank earnings announcements contain information about expected profitability and lending capacity, and uses the co-movement between bank equity reactions and corporate spreads to purge raw bank equity surprises from contaminating information. We now ask whether the resulting shock affects aggregate Canadian financial conditions and macroeconomic outcomes. The theoretical framework in Section \ref{sec:illustrative_model} predicts that favorable news about intermediary net worth should raise bank valuations, lower corporate borrowing spreads, expand credit supply, and stimulate real activity. We test these predictions using monthly local projections. We first describe the macro-financial dataset and empirical specification. We then present benchmark impulse responses to the median rotational-angle purged credit-supply shock. Finally, we provide additional evidence using the raw bank equity surprise, an alternative poor man's sign-restriction shock, and an IV exercise that isolates bank-induced movements in Canadian corporate spreads.

%%%%%%%%%%%%%%%%%%%%%%%%%%%%%%%%%%%%%%%%%%%%%%%%%%%%%%%%%%%%%%%%%%%%%%
\subsection{Data and Empirical Specification}

\paragraph{Dataset.}

We estimate the dynamic effects of Canadian bank net-worth shocks using a monthly macro-financial dataset. Most Canadian macroeconomic variables are sourced from the Large Canadian Database for Macroeconomic Analysis compiled by \cite{fortin2022large}.\footnote{In particular, we use data sourced from the March 2026 vintage of the LCDMA dataset.} We merge these data with the monthly purged credit-supply bank net-worth shock constructed in Section \ref{sec:shock_construction}, the Bloomberg Canadian corporate option-adjusted spread, and aggregate market capitalization for the six large Canadian banks.

The benchmark shock is the purged credit-supply component of the Canadian bank equity surprise identified using the median rotational-angle procedure, denoted by \(v_t^{CS}\). We aggregate the event-level shock to the monthly frequency and normalize it by its sample standard deviation. The impulse responses below are therefore interpreted as responses to a one-standard-deviation favorable credit-supply bank net-worth shock.

The benchmark outcomes capture Canadian financial conditions, asset prices, monetary policy, real activity, and prices. We use the Bloomberg Canadian corporate option-adjusted spread, denoted by \(OAS_t\), as a measure of corporate borrowing conditions. The daily OAS series is converted to monthly frequency using the last available observation in each month. We use the aggregate market capitalization of the six large Canadian banks as a measure of bank equity valuations. To construct this series, we first compute each bank's monthly market capitalization as the last of its daily market capitalizations within the month. We then sum the monthly market capitalizations of RBC, TD, Scotiabank, BMO, CIBC, and National Bank to obtain aggregate Canadian bank market capitalization, denoted by \(MC_t\).

The remaining benchmark outcomes are the Canadian equity index, the nominal exchange rate vis-\`a-vis the U.S. dollar, the Bank of Canada policy rate, real activity, and consumer prices. The equity index, nominal exchange rate, real activity, consumer price index, and aggregate bank market capitalization are transformed into log variables multiplied by 100.\footnote{In particular, we use the following variables from the LCDMA dataset: (i) GDP, code \texttt{GDP\textunderscore NEW}; (ii) the USDCAD exchange rate, code \texttt{USCAD\textunderscore NEW}; (iii) the CPI index, code \texttt{CPI\textunderscore ALL\textunderscore CAN}; (iv) the stock market index, code \texttt{TSX\textunderscore CLO}; and (v) the Bank of Canada policy rate, code \texttt{BANK\textunderscore RATE\textunderscore L}.}

%%%%%%%%%%%%%%%%%%%%%%%%%%%%%%%%%%%%%%%%%%%%%%%%%%%%%%%%%%%%%%%%%%%%%%
\paragraph{Empirical specification.}

For each outcome variable \(y_t\) and horizon \(h=0,\dots,36\), we estimate monthly local projections of the form:
\begin{equation}
Y_{t+h}^{(h)}
=
\alpha_h
+
\beta_h v_t^{CS}
+
\Gamma_h' X_{t-1}
+
u_{t+h}^{(h)},
\label{eq:lp_benchmark}
\end{equation}
where \(v_t^{CS}\) is the normalized purged credit-supply bank net-worth shock and \(X_{t-1}\) is a vector of predetermined Canadian macro-financial controls. The coefficient of interest is \(\beta_h\), which traces the dynamic response of the outcome variable at horizon \(h\) to a favorable bank net-worth shock. The benchmark specification includes only the contemporaneous shock, not its lags.

For variables expressed in logs, the dependent variable is the cumulative change between horizon \(t+h\) and the month before the shock:
\begin{equation}
Y_{t+h}^{(h)}
=
y_{t+h}-y_{t-1}.
\end{equation}
Since log variables are multiplied by 100, these responses are measured in percentage points. This transformation is used for real activity, the nominal exchange rate, consumer prices, the Canadian equity index, and real aggregate bank market capitalization. For the Bank of Canada policy rate and the Canadian corporate OAS, we estimate responses in forward levels:
\begin{equation}
Y_{t+h}^{(h)}
=
y_{t+h}.
\end{equation}
The responses of the policy rate and the corporate spread are therefore interpreted in the original units of their respective series.

The benchmark control vector includes six monthly lags of Canadian macro-financial variables. The benchmark specification also includes a COVID-period dummy equal to one from March 2020 through June 2021. Standard errors are computed using Newey--West corrections with 36 lags, the same as the horizon for the local projections. We report 68 percent and 90 percent confidence intervals.

%%%%%%%%%%%%%%%%%%%%%%%%%%%%%%%%%%%%%%%%%%%%%%%%%%%%%%%%%%%%%%%%%%%%%%%%%%%%%%%%%%%%%%%%
\subsection{Benchmark Results} \label{subsec:benchmark_results}

Figure \ref{fig:benchmark_macro_lp_supply_shock} reports the dynamic responses to a one-standard-deviation favorable purged credit-supply bank net-worth shock, identified using the median rotational-angle procedure. The responses are consistent with the intermediary balance-sheet channel described in Section \ref{sec:illustrative_model}: favorable bank net-worth news eases corporate borrowing conditions, raises bank valuations, improves broader asset prices, and is followed by a gradual expansion in real activity.

%%%%%%%%%%%%%%%%%%%%%%%%%%%%%%%%%%%%%%%%%%%%%%%%%%%%%%%%%%%%%%%%%%%%%%
\begin{figure}[H]
\centering
\caption{Macroeconomic Propagation of Credit-Supply Bank Net-Worth Shocks}
\includegraphics[width=0.95\textwidth]{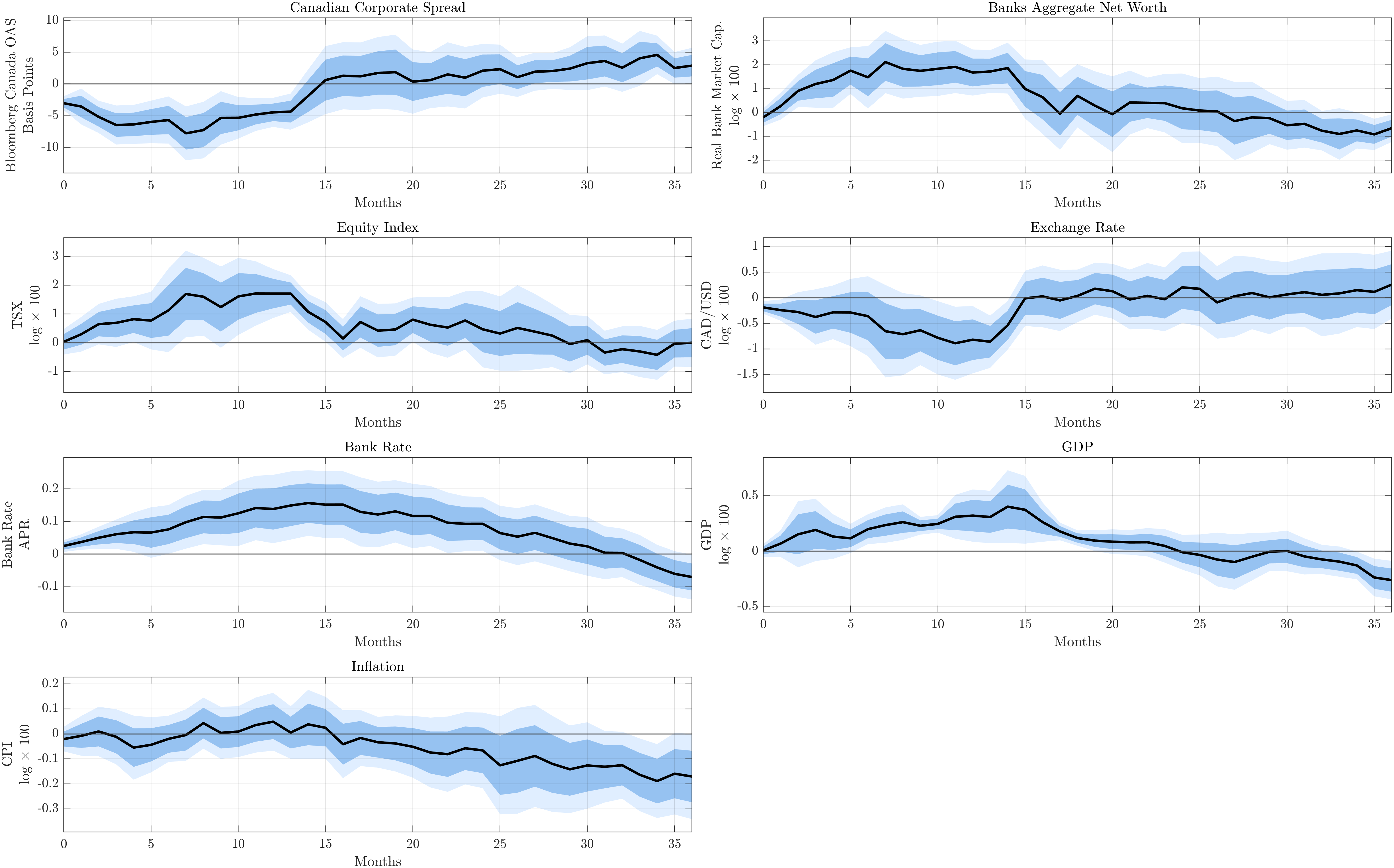}
\label{fig:benchmark_macro_lp_supply_shock}
\floatfoot{\textbf{Notes:} This figure reports local projection impulse responses to a one-standard-deviation favorable purged credit-supply bank net-worth shock identified using the median rotational-angle procedure. The shock is constructed from timing-adjusted earnings-announcement stock-price reactions and aggregated using lagged market-capitalization shares. Responses are estimated using Equation \eqref{eq:lp_benchmark}. For variables in logs, responses are cumulative changes relative to the month before the shock and are measured in percentage points. The Bank of Canada policy rate and the Canadian corporate OAS are reported in forward levels. Dark and light shaded areas denote 68 and 90 percent confidence intervals, respectively, computed using Newey--West standard errors.}
\end{figure}
%%%%%%%%%%%%%%%%%%%%%%%%%%%%%%%%%%%%%%%%%%%%%%%%%%%%%%%%%%%%%%%%%%%%%%

The response of corporate spreads provides the most direct test of the credit-supply interpretation. Following a favorable bank net-worth shock, the Canadian corporate OAS declines on impact and remains below its pre-shock level for roughly one year. The decline reaches approximately 5 to 8 basis points during the first several months. This response is consistent with the model's prediction that an increase in intermediary net worth relaxes balance-sheet constraints, expands credit supply, and compresses lending spreads.

Bank valuations and asset prices also move in the direction predicted by the model. Real aggregate market capitalization of the six large Canadian banks increases immediately and remains elevated for more than a year, peaking at around 2 percent. This persistent increase suggests that earnings-announcement stock-price reactions capture news about expected bank profitability and lending capacity rather than purely transitory price movements. The broader Canadian equity index also increases, with the response peaking at roughly 1.5 to 2 percent after about one year. Thus, the shock appears to ease financial conditions beyond the banking sector itself.

The exchange-rate and policy-rate responses are also consistent with an improvement in domestic financial conditions. The Canadian dollar appreciates against the U.S. dollar over the first year after the shock before the response gradually reverses. The Bank of Canada policy rate increases gradually, peaking after roughly one year. This pattern is consistent with an endogenous policy-rate response to improved financial conditions and stronger domestic activity.

The real-side responses indicate that bank net-worth shocks propagate to aggregate activity. GDP rises following the shock, with the response peaking after about one year. The increase is economically modest but persistent over the first part of the horizon, consistent with the view that easier credit conditions stimulate real activity gradually. Consumer prices respond little on impact and decline later in the horizon, suggesting that the shock operates first through financial conditions and real activity, while price dynamics adjust more slowly.

\paragraph{Credit-market responses.}
The credit-market responses provide additional evidence on the transmission
mechanism. Figure \ref{fig:financial_prices_supply_shock} reports responses
for several interest-rate spreads: government bond yields at different
maturities relative to the Bank of Canada policy rate, and one- and five-year
mortgage rates relative to the policy rate. These spreads are not used in the
construction of the shock. A favorable bank net-worth shock lowers medium- and
longer-maturity government bond spreads and reduces mortgage spreads over the
first part of the horizon. This pattern is consistent with an easing of broader
financial conditions beyond the corporate OAS used in the shock construction.

%%%%%%%%%%%%%%%%%%%%%%%%%%%%%%%%%%%%%%%%%%%%%%%%%%%%%%%%%%%%%%%%%%%%%%
\begin{figure}[H]
\centering
\includegraphics[width=0.95\textwidth]{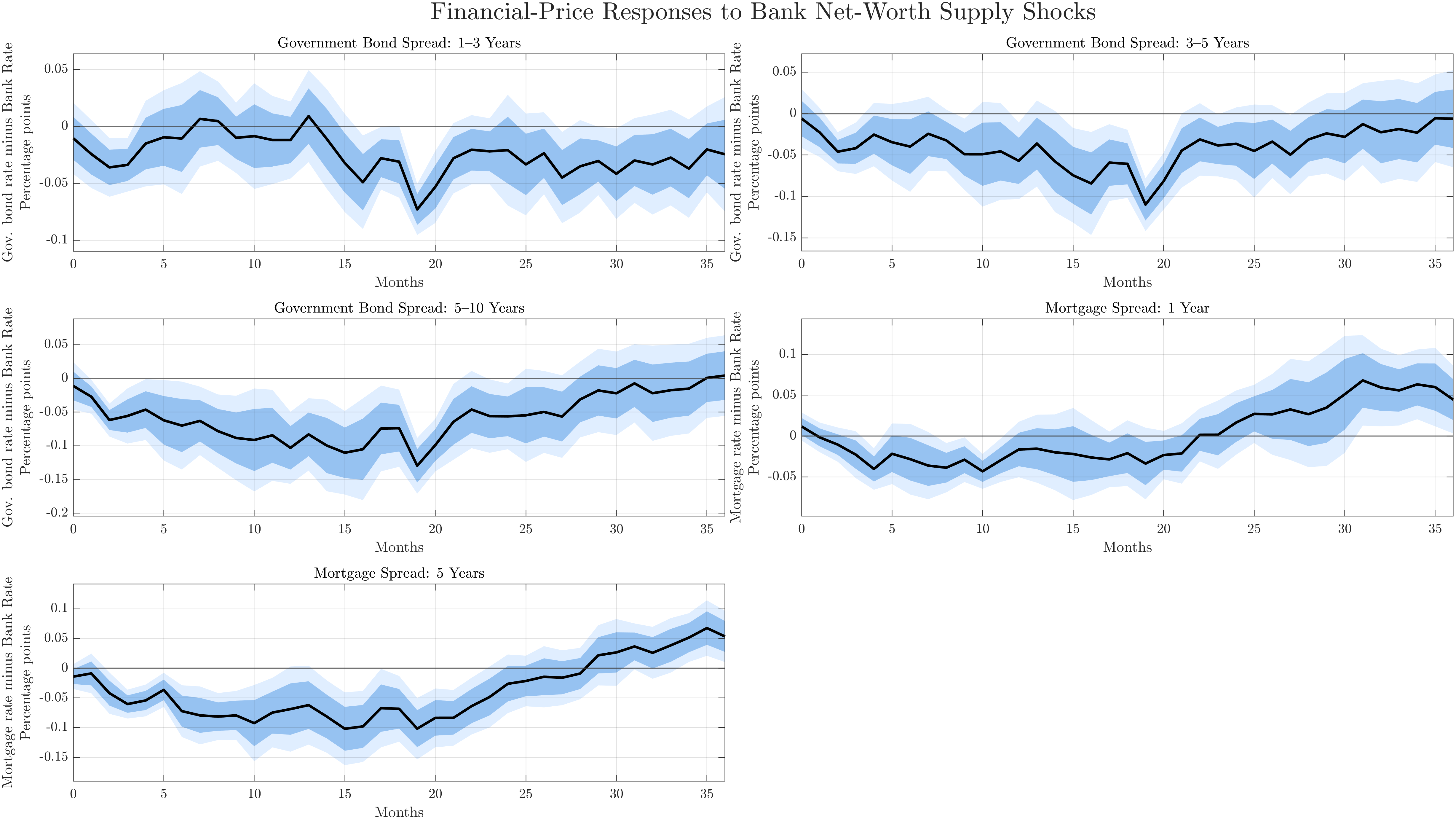}
\caption{Financial-Price Responses to Bank Net-Worth Supply Shocks}
\label{fig:financial_prices_supply_shock}
\floatfoot{\textbf{Notes:} This figure reports local projection impulse responses of additional financial-price variables to a one-standard-deviation favorable credit-supply bank net-worth shock identified using the median rotational-angle procedure. The outcomes are government bond spreads at one- to three-year, three- to five-year, and five- to ten-year maturities, measured relative to the Bank of Canada policy rate, and one- and five-year mortgage spreads, also measured relative to the policy rate. The specification includes the benchmark Canadian macro-financial controls, a COVID-period dummy, and six lags of the corresponding outcome. Responses are measured in percentage points. Dark and light shaded areas denote 68 and 90 percent confidence intervals, respectively, computed using Newey--West standard errors.}
\end{figure}
%%%%%%%%%%%%%%%%%%%%%%%%%%%%%%%%%%%%%%%%%%%%%%%%%%%%%%%%%%%%%%%%%%%%%%

Figure \ref{fig:financial_quantities_supply_shock} reports the corresponding
responses of credit quantities. This exercise is particularly useful for
interpreting the shock because credit quantities are not used in the sign
restriction. Business lending rises gradually after a favorable bank
net-worth shock and remains elevated over the medium run. Total loans also
increase during the first year, while residential mortgages respond positively
in the short run before reversing later in the horizon. Personal loans respond
more weakly and imprecisely. The stronger response of business credit is
consistent with the mechanism in the model: favorable news about bank balance
sheets eases credit supply and supports lending to the real economy.

%%%%%%%%%%%%%%%%%%%%%%%%%%%%%%%%%%%%%%%%%%%%%%%%%%%%%%%%%%%%%%%%%%%%%%
\begin{figure}[H]
\centering
\includegraphics[width=0.95\textwidth]{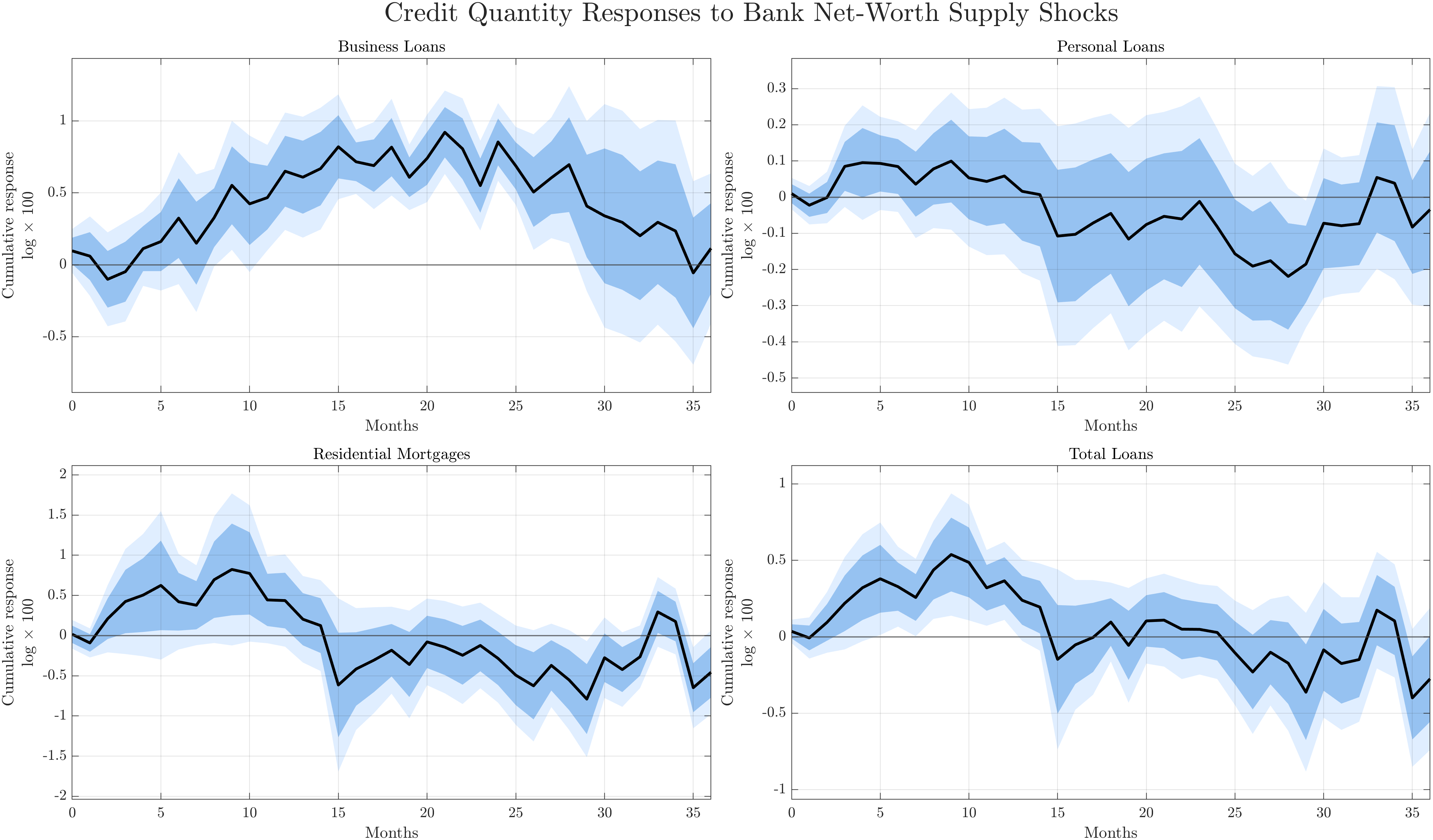}
\caption{Credit Quantity Responses to Bank Net-Worth Supply Shocks}
\label{fig:financial_quantities_supply_shock}
\floatfoot{\textbf{Notes:} This figure reports local projection impulse responses of credit quantities to a one-standard-deviation favorable credit-supply bank net-worth shock identified using the median rotational-angle procedure. The outcomes are business loans, personal loans, residential mortgages, and total loans. Credit quantities are transformed as $100$ times log levels, and responses are cumulative changes relative to the month before the shock. The specification includes the benchmark Canadian macro-financial controls, a COVID-period dummy, and six lags of the corresponding credit quantity. Dark and light shaded areas denote 68 and 90 percent confidence intervals, respectively, computed using Newey--West standard errors.}
\end{figure}
%%%%%%%%%%%%%%%%%%%%%%%%%%%%%%%%%%%%%%%%%%%%%%%%%%%%%%%%%%%%%%%%%%%%%%

Taken together, the benchmark responses support the paper's central mechanism.
Favorable purged credit-supply news about Canadian banks lowers corporate
borrowing spreads, reduces broader credit spreads, raises bank and equity
valuations, appreciates the Canadian dollar, and is followed by an increase in
real activity. The accompanying rise in business credit and total credit
provides additional evidence that the shock operates through credit supply
rather than only through asset prices. These results suggest that bank
net-worth shocks are not merely bank-level valuation shocks: in a concentrated
banking system, they propagate to aggregate financial conditions, credit
quantities, and macroeconomic outcomes.

%%%%%%%%%%%%%%%%%%%%%%%%%%%%%%%%%%%%%%%%%%%%%%%%%%%%%%%%%%%%%%%%%%%%%%%%%%%%%%%%%%%%%%%%
\subsection{Additional Evidence} \label{subsec:additional_evidence}

The benchmark results use the purged credit-supply component of bank equity news identified through the median rotational-angle procedure. In this subsection, we provide three complementary exercises. First, we estimate the same local projections using the raw market-capitalization-weighted bank equity surprise, without purging contaminating information using corporate spread co-movement. This exercise asks whether bank equity news before the sign-restriction step contains similar macroeconomic information. Second, we estimate responses using the credit-supply shock obtained from the simpler poor man's sign-restriction approach. This exercise assesses whether the benchmark results are driven by the particular implementation of the median rotational-angle procedure. Third, we use the benchmark purged credit-supply shock as an external instrument for Canadian corporate spreads to study the macroeconomic effects of bank-induced movements in borrowing costs.

%%%%%%%%%%%%%%%%%%%%%%%%%%%%%%%%%%%%%%%%%%%%%%%%%%%%%%%%%%%%%%%%%%%%%%
\paragraph{Raw bank equity surprises.}

Figure \ref{fig:macro_lp_total_bank_shock} reports the responses to a
one-standard-deviation increase in the raw market-capitalization-weighted bank
equity surprise, \(v_t\), without purging the shock using the sign-restriction
procedure. The raw surprise contains useful information about Canadian banks:
real aggregate bank market capitalization rises after the shock, indicating
that the event-window stock-price reactions capture news about bank valuations.
However, the broader macro-financial responses are less clearly aligned with
the credit-supply mechanism than in the benchmark specification.

In particular, the response of Canadian corporate spreads is weaker and less
systematic than in the purged credit-supply shock. The OAS declines only with a
lag and the response is not as persistent or precisely estimated as in the
benchmark. Broader equity prices also do not move in the same direction as bank
valuations: the TSX falls over much of the horizon, even though aggregate bank
market capitalization increases. GDP rises in the medium run, but the
overall pattern is less tightly connected to an easing of broad financial
conditions than in the benchmark results.

This comparison illustrates why the co-movement between bank equity reactions
and corporate spreads is informative. Raw bank equity news mixes several types
of information released during earnings announcements. Some of this news may
reflect improvements in bank balance-sheet strength and lending capacity, but
other components may reflect borrower fundamentals, expected loan demand,
aggregate conditions, or risk premia. These components can raise bank
valuations without generating the broader easing in credit conditions predicted
by the model. The sign-restriction step isolates the component of bank equity
news that moves bank valuations and corporate borrowing spreads in opposite
directions. The sharper responses in the benchmark specification therefore
support the interpretation that this co-movement helps separate credit-supply
news from the other information contained in raw bank equity surprises.

%%%%%%%%%%%%%%%%%%%%%%%%%%%%%%%%%%%%%%%%%%%%%%%%%%%%%%%%%%%%%%%%%%%%%%
\begin{figure}[H]
\centering
\includegraphics[width=0.95\textwidth]{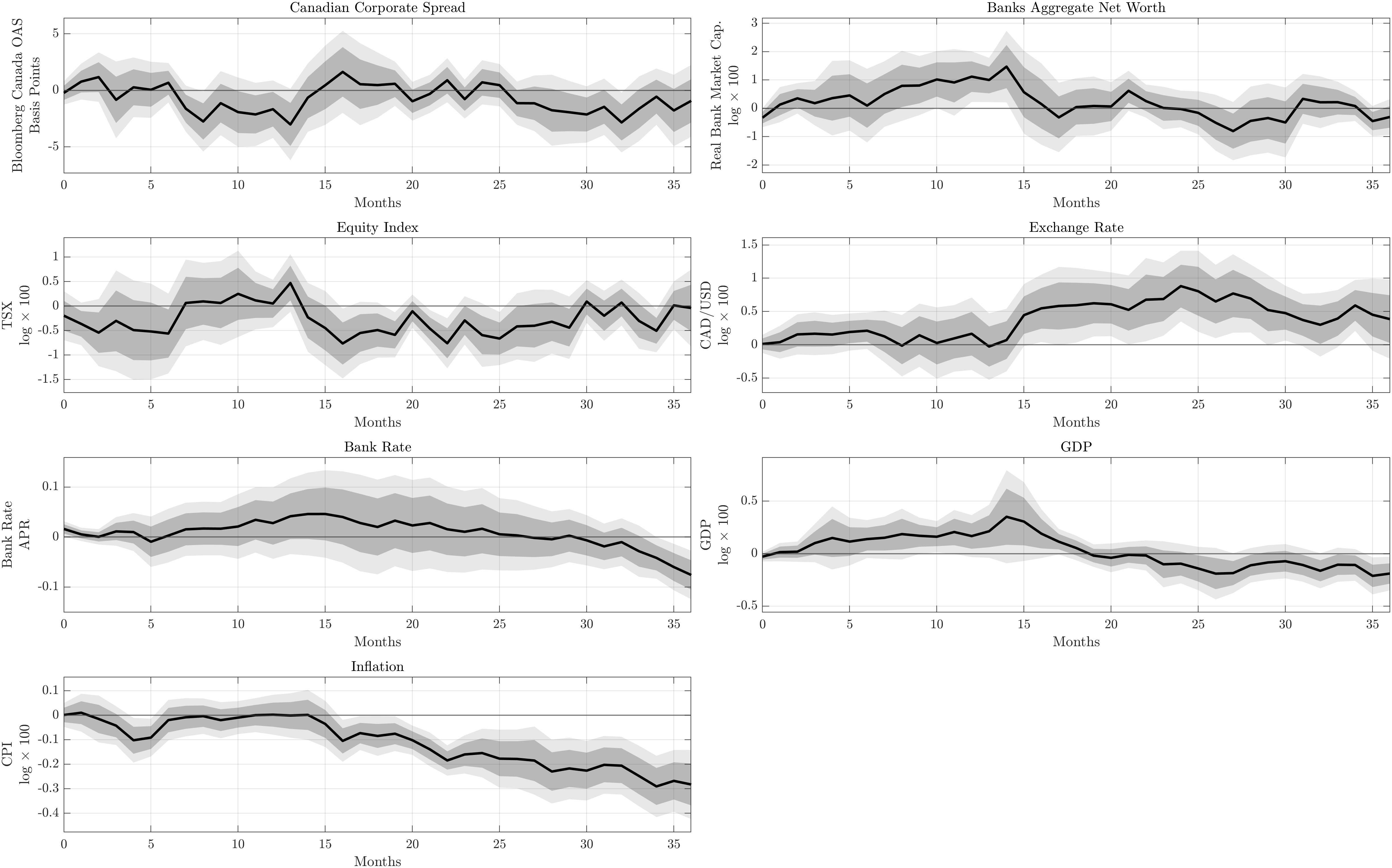}
\caption{Macroeconomic Propagation of Raw Bank Equity Surprises}
\label{fig:macro_lp_total_bank_shock}
\floatfoot{\textbf{Notes:} This figure reports local projection impulse responses to a one-standard-deviation increase in the raw market-capitalization-weighted bank equity surprise, without purging the surprise using the sign-restriction procedure. The raw surprise is constructed from timing-adjusted earnings-announcement stock-price reactions and aggregated using lagged market-capitalization shares. Responses are estimated using Equation \eqref{eq:lp_benchmark}. For variables in logs, responses are cumulative changes relative to the month before the shock and are measured in percentage points. The Bank of Canada policy rate and the Canadian corporate OAS are reported in forward levels. Dark and light shaded areas denote 68 and 90 percent confidence intervals, respectively, computed using Newey--West standard errors.}
\end{figure}
%%%%%%%%%%%%%%%%%%%%%%%%%%%%%%%%%%%%%%%%%%%%%%%%%%%%%%%%%%%%%%%%%%%%%%

%%%%%%%%%%%%%%%%%%%%%%%%%%%%%%%%%%%%%%%%%%%%%%%%%%%%%%%%%%%%%%%%%%%%%%
\paragraph{Alternative sign-restriction implementation.}

Figure \ref{fig:macro_lp_pm_supply_shock} reports the responses obtained using the poor man's sign-restriction credit-supply shock. This shock classifies event-window observations as credit-supply news when bank stock prices and corporate spreads move in opposite directions. Thus, favorable credit-supply news corresponds to an increase in bank equity prices together with a decline in the Canadian corporate OAS. This procedure is less structured than the median rotational-angle decomposition, but it imposes the same economic sign restriction motivated by the model.

The responses are close to the benchmark results. A favorable poor man's credit-supply shock lowers corporate spreads in the first year after the shock, raises real aggregate bank market capitalization, increases equity prices, and is followed by an expansion in GDP. The Bank of Canada policy rate rises gradually, consistent with an endogenous policy-rate response to improved financial conditions and stronger activity. The exchange-rate and inflation responses are also broadly similar to those obtained with the median rotational-angle shock.

These results show that the benchmark findings are not driven by the particular rotational implementation of the sign restrictions. Instead, the main propagation patterns are present when credit-supply news is identified using a simpler classification rule based directly on the sign of the co-movement between bank equity prices and corporate spreads. Together with the raw-shock results, this evidence supports the main interpretation of the paper: earnings-announcement news about Canadian banks contains macroeconomically relevant information, and the component associated with improvements in bank credit supply propagates to broader financial conditions and real activity.

%%%%%%%%%%%%%%%%%%%%%%%%%%%%%%%%%%%%%%%%%%%%%%%%%%%%%%%%%%%%%%%%%%%%%%
\begin{figure}[H]
\centering
\includegraphics[width=0.95\textwidth]{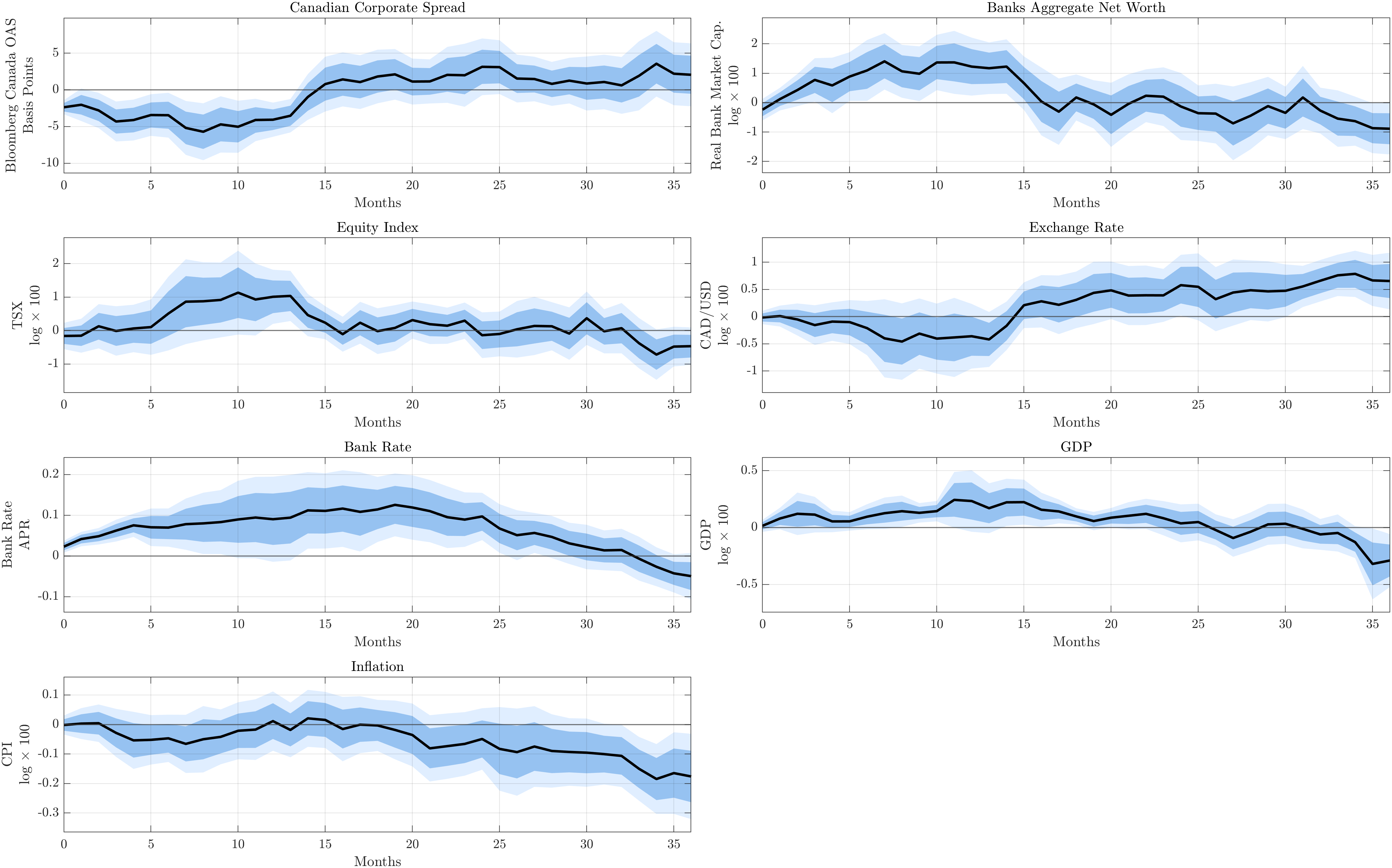}
\caption{Macroeconomic Propagation of Poor Man's Credit-Supply Bank Net-Worth Shocks}
\label{fig:macro_lp_pm_supply_shock}
\floatfoot{\textbf{Notes:} This figure reports local projection impulse responses to a one-standard-deviation favorable credit-supply bank net-worth shock identified using the poor man's sign-restriction approach. The shock is constructed from timing-adjusted earnings-announcement stock-price reactions and aggregated using lagged market-capitalization shares. Credit-supply observations are classified using the sign of the co-movement between bank equity-price reactions and event-window changes in the Canadian corporate OAS. Responses are estimated using Equation \eqref{eq:lp_benchmark}. For variables in logs, responses are cumulative changes relative to the month before the shock and are measured in percentage points. The Bank of Canada policy rate and the Canadian corporate OAS are reported in forward levels. Dark and light shaded areas denote 68 and 90 percent confidence intervals, respectively, computed using Newey--West standard errors.}
\end{figure}

%%%%%%%%%%%%%%%%%%%%%%%%%%%%%%%%%%%%%%%%%%%%%%%%%%%%%%%%%%%%%%%%%%%%%%
\paragraph{Instrumenting Canadian corporate spreads.}

The previous exercises estimate the reduced-form dynamic effects of bank equity surprises. We now provide additional evidence on the credit-spread channel by using the identified purged credit-supply bank net-worth shock as an external instrument for Canadian corporate spreads. This exercise asks how Canadian macro-financial variables respond to movements in corporate borrowing spreads that are induced by bank net-worth news.

For each horizon \(h=0,\dots,36\), we estimate the local projection IV specification:
\begin{equation}
Y_{t+h}^{(h)}
=
\alpha_h
+
\theta_h \widehat{OAS}_t
+
\Gamma_h'X_{t-1}
+
u_{t+h}^{(h)},
\label{eq:lp_iv_oas}
\end{equation}
where \(OAS_t\) is the Canadian corporate option-adjusted spread and is instrumented using the normalized median rotational-angle purged credit-supply bank net-worth shock, \(v_t^{CS}\). The first stage is:
\begin{equation}
OAS_t
=
\pi v_t^{CS}
+
\Lambda'X_{t-1}
+
\eta_t.
\label{eq:first_stage_oas}
\end{equation}
The control vector \(X_{t-1}\) is the same as in the benchmark local projections and includes six monthly lags of Canadian macro-financial variables and the COVID-period dummy. Standard errors are computed using Newey--West corrections with 36 lags.

The first stage has the expected sign and is statistically relevant across horizons. Favorable credit-supply bank net-worth shocks predict lower Canadian corporate spreads, consistent with the sign restriction used to construct the shock. The horizon-specific first-stage \(F\)-statistics range from 10.18 to 13.15, with a median of 12.20. Thus, the instrument generates relevant variation in Canadian corporate spreads, although the IV results should be interpreted as complementary evidence rather than as the paper's primary identification exercise.

Because favorable bank net-worth shocks reduce corporate spreads, the IV responses in Figure \ref{fig:macro_lp_iv_supply_shock} are normalized as responses to a bank-induced tightening in Canadian corporate credit conditions. Equivalently, the figure reports the effects of an increase in the Canadian corporate OAS induced by adverse bank net-worth news.

%%%%%%%%%%%%%%%%%%%%%%%%%%%%%%%%%%%%%%%%%%%%%%%%%%%%%%%%%%%%%%%%%%%%%%
\begin{figure}[H]
\centering
\includegraphics[width=0.95\textwidth]{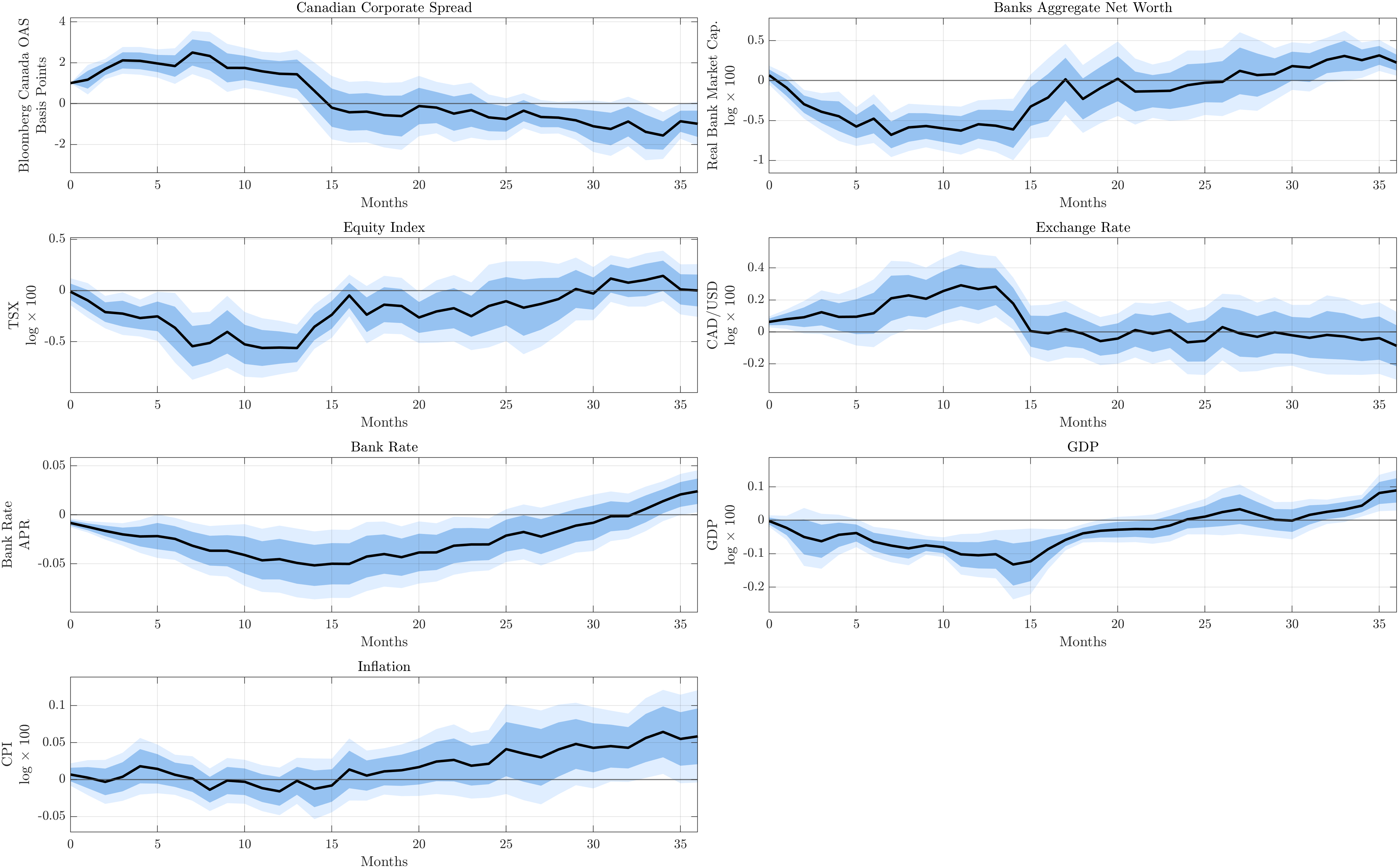}
\caption{Macroeconomic Effects of Bank-Induced Corporate Spread Movements}
\label{fig:macro_lp_iv_supply_shock}
\floatfoot{\textbf{Notes:} This figure reports local projection IV impulse responses to an increase in the Canadian corporate OAS instrumented by the median rotational-angle purged credit-supply bank net-worth shock. The first stage instruments the current Canadian corporate OAS with the normalized purged credit-supply bank net-worth shock constructed from timing-adjusted earnings-announcement stock-price reactions and aggregated using lagged market-capitalization shares. Responses are estimated using Equation \eqref{eq:lp_iv_oas}. For variables in logs, responses are cumulative changes relative to the month before the shock and are measured in percentage points. The Bank of Canada policy rate and the Canadian corporate OAS are reported in forward levels. Dark and light shaded areas denote 68 and 90 percent confidence intervals, respectively, computed using Newey--West standard errors.}
\end{figure}
%%%%%%%%%%%%%%%%%%%%%%%%%%%%%%%%%%%%%%%%%%%%%%%%%%%%%%%%%%%%%%%%%%%%%%

The IV responses are consistent with the credit-supply interpretation. A bank-induced increase in corporate spreads is followed by a decline in real aggregate bank market capitalization, a fall in the broader equity index, and a contraction in real activity. Bank market capitalization falls during the first year after the shock, indicating that the spread tightening is associated with a deterioration in bank valuations. Equity prices also decline, suggesting that tighter bank-induced credit conditions spill over to broader asset markets.

The real-side responses point to contractionary macroeconomic effects. GDP falls after the bank-induced spread increase, with the decline concentrated over the first year. The Bank of Canada policy rate falls gradually, consistent with an endogenous policy-rate response to weaker financial conditions and lower real activity. The exchange-rate response implies a depreciation of the Canadian dollar against the U.S. dollar during the first year, while inflation increases gradually later in the horizon.

Taken together, the IV evidence provides complementary support for the interpretation of the benchmark results. The reduced-form responses show that favorable credit-supply bank net-worth shocks lower corporate spreads and stimulate activity. The IV exercise shows the same mechanism from the opposite direction: movements in Canadian corporate spreads induced by bank net-worth news have contractionary effects on bank valuations, equity prices, and real activity. This supports the view that corporate borrowing spreads are an important transmission channel through which bank net-worth shocks propagate to the Canadian macroeconomy.

%%%%%%%%%%%%%%%%%%%%%%%%%%%%%%%%%%%%%%%%%%%%%%%%%%%%%%%%%%%%%%%%%%
%%%%%%%%%%%%%%%%%%%%%%%%%%%%%%%%%%%%%%%%%%%%%%%%%%%%%%%%%%%%%%%%%%
%%%%%%%%%%%%%%%%%%%%%%%%%%%%%%%%%%%%%%%%%%%%%%%%%%%%%%%%%%%%%%%%%%
\section{Additional Results \& Robustness Checks } \label{sec:robustness}

This section summarizes additional exercises that assess the robustness of the benchmark propagation results and explore further dimensions of the transmission mechanism. We first examine whether the main impulse responses are sensitive to the choice of admissible rotation used to construct the sign-restricted credit-supply shock. We then study alternative empirical specifications, including a more parsimonious lag structure and the inclusion of a deterministic time trend. Next, we examine sample robustness by estimating the responses before 2020 and in the post-2010 period. We also examine whether the results are driven by the largest banks in the sample by reconstructing the shock after excluding RBC, after excluding TD, and after excluding both RBC and TD. Finally, we extend the analysis to additional outcomes, including labor-market variables, goods- and services-producing sectors, and more disaggregated measures of sectoral activity. To keep the main text focused on the core findings, all corresponding impulse-response figures are reported in Appendix \ref{appendix:figures_robustness_additional}.

%%%%%%%%%%%%%%%%%%%%%%%%%%%%%%%%%%%%%%%%%%%%%%%%%%%%%%%%%%%%%%%%%%%%%%
\subsection{Sensitivity to the Admissible Rotation}
\label{subsec:robustness_rotation}

We first examine whether the benchmark responses are sensitive to the choice
of the median admissible rotation used to construct the purged credit-supply
bank net-worth shock. The sign restrictions in Section \ref{sec:shock_construction}
do not point-identify a unique decomposition of the raw bank equity surprise.
Instead, they define a set of admissible rotations. Our benchmark shock uses
the median admissible rotation, \(w=0.50\). To assess whether the main
propagation patterns depend on this particular choice, we re-estimate the
benchmark local projections using two alternative rotations from the central
part of the admissible set: \(w=0.25\) and \(w=0.75\). These alternatives
preserve the same sign restrictions, sample, outcomes, controls, horizon,
COVID-period dummy, and Newey--West inference procedure as the benchmark.

Appendix Figures \ref{fig:appendix_rotation_p25} and
\ref{fig:appendix_rotation_p75} report the corresponding impulse responses.
The results are broadly consistent with the benchmark median-rotation
responses. For both alternative rotations, favorable credit-supply bank
net-worth shocks lower Canadian corporate spreads during the first year after
the shock, raise real aggregate bank market capitalization, increase broader
equity prices, and are followed by a positive response of GDP over the medium
run. The magnitudes differ across rotations, as expected given that the
rotations place different weights on the raw bank equity surprise and the OAS
innovation, but the central financial and real-side propagation patterns
remain intact. The \(w=0.25\) rotation is closer to the raw bank equity
surprise, while the \(w=0.75\) rotation places relatively more weight on the
spread component. The similarity of the main responses across these
interquartile rotations suggests that the benchmark results are not driven by
the particular choice of the median admissible angle.

%%%%%%%%%%%%%%%%%%%%%%%%%%%%%%%%%%%%%%%%%%%%%%%%%%%%%%%%%%%%%%%%%%%%%% 
\subsection{Specification Robustness} \label{subsec:robustness_specification}

We first examine whether the benchmark propagation results are sensitive to the specification of the local projections. The benchmark specification includes six monthly lags of Canadian macro-financial controls. While this lag structure is intended to flexibly absorb predictable dynamics in Canadian financial and macroeconomic variables, it is useful to verify that the main findings are not driven by the particular number of lags included in the control vector. We therefore re-estimate the benchmark local projections using only two monthly lags of the Canadian controls. This more parsimonious specification preserves the same shock, outcomes, sample, horizon, COVID-period dummy, and Newey--West inference procedure as the benchmark.

The results are reported in Appendix Figure \ref{fig:appendix_robust_lags2}. They are broadly consistent with the benchmark responses. A favorable credit-supply bank net-worth shock lowers Canadian corporate spreads during the first year after the shock, raises real aggregate bank market capitalization, increases broader equity prices, and is followed by a positive response of real activity. The timing and magnitudes differ somewhat from the benchmark specification, but the central pattern remains unchanged: favorable bank balance-sheet news eases financial conditions and propagates to aggregate activity. This suggests that the benchmark findings are not an artifact of including six lags of the control variables.

As a second specification check, we augment the benchmark local projections with a linear deterministic time trend. This exercise addresses the possibility that the benchmark responses reflect low-frequency trends in Canadian financial markets, bank valuations, corporate spreads, or macroeconomic activity. The trend specification keeps the six monthly lags of Canadian macro-financial controls and adds a deterministic linear trend to the control vector.

Appendix Figure \ref{fig:appendix_robust_trend} shows that the main results are also robust to including the linear trend. Corporate spreads decline after a favorable credit-supply bank net-worth shock, real aggregate bank market capitalization rises, and the broader equity market responds positively. GDP increases over the medium run, while the policy rate rises gradually, consistent with an endogenous response to stronger financial conditions and real activity. Overall, the trend specification delivers propagation patterns that are close to the benchmark results. This supports the interpretation that the estimated responses are driven by high-frequency bank net-worth news rather than by spurious low-frequency co-movement.

%%%%%%%%%%%%%%%%%%%%%%%%%%%%%%%%%%%%%%%%%%%%%%%%%%%%%%%%%%%%%%%%%%%%%%
\subsection{Sample Robustness}
\label{subsec:robustness_sample}

We next examine whether the benchmark results are driven by particular parts of the sample. The baseline sample begins in late 2002, when the Bloomberg Canadian corporate OAS series becomes available, and extends through early 2026. This period includes several distinct macro-financial environments, including the global financial crisis, the post-crisis low-rate period, the COVID-19 shock, and the subsequent tightening cycle. We therefore estimate two additional sample restrictions. First, we re-estimate the benchmark specification using the pre-2020 sample, which excludes the COVID period and the most recent post-pandemic observations. Second, we estimate the responses using the post-2010 sample, which focuses on the period after the global financial crisis.

The pre-2020 results are reported in Appendix Figure \ref{fig:appendix_robust_pre2020}. The responses are close to the benchmark estimates. A favorable credit-supply bank net-worth shock lowers Canadian corporate spreads during the first year after the shock, raises real aggregate bank market capitalization, increases broader equity prices, and is followed by a positive response of GDP. The exchange-rate response is somewhat larger than in the full sample, while the policy-rate response remains positive over the medium run. Overall, the pre-2020 estimates indicate that the benchmark propagation patterns are not driven by the COVID episode or by the unusual macro-financial dynamics of the post-pandemic period.

Appendix Figure \ref{fig:appendix_robust_post2010} reports the responses for the post-2010 sample. This exercise asks whether the main results survive when we focus on the post-global-financial-crisis period, a sample characterized by a different regulatory environment, lower interest rates, and a more stable Canadian banking system. The financial responses remain consistent with the credit-supply interpretation: corporate spreads fall after a favorable bank net-worth shock, while real aggregate bank market capitalization rises. The response of GDP is also positive over the medium run, although the equity-price response is less pronounced than in the full sample. The policy-rate response is positive initially and then turns negative later in the horizon, suggesting some differences in monetary-policy dynamics in the post-2010 period.

Taken together, these sample restrictions support the robustness of the main results. The core pattern---favorable bank net-worth news lowers corporate borrowing spreads, raises bank valuations, and is followed by stronger real activity---is present both before 2020 and after 2010. The magnitudes and persistence vary across samples, as expected given the different macro-financial environments, but the central propagation mechanism remains intact.

%%%%%%%%%%%%%%%%%%%%%%%%%%%%%%%%%%%%%%%%%%%%%%%%%%%%%%%%%%%%%%%%%%%%%%
\subsection{Robustness to Excluding the Largest Banks}
\label{subsec:robustness_large_banks}

We also examine whether the benchmark results are driven by the largest banks in the sample. This is an important concern because the Canadian banking sector is highly concentrated, and the market-capitalization-weighted shock places larger weight on the largest institutions. If the benchmark responses were mechanically driven by one or two dominant banks, the interpretation of the shock as aggregate Canadian bank net-worth news would be weaker.

To address this concern, we reconstruct the shock series after excluding RBC, after excluding TD, and after excluding both RBC and TD. In each case, we repeat the shock-construction procedure using the remaining banks and re-estimate the benchmark local projections with the same outcomes, controls, normalization, horizon, COVID-period dummy, and Newey--West inference procedure as in the baseline specification. Appendix Figures \ref{fig:appendix_no_rbc}, \ref{fig:appendix_no_td}, and \ref{fig:appendix_no_rbc_td} report the corresponding impulse responses.

The results are close to the benchmark estimates. In all three exercises, a favorable credit-supply bank net-worth shock lowers Canadian corporate spreads during the first year after the shock, raises real aggregate bank market capitalization, increases the broader equity index, and is followed by a positive response of GDP over the medium run. The policy-rate response is also similar to the benchmark, rising gradually after the shock before declining later in the horizon. The exchange-rate and inflation responses display some variation across exclusions, but the central financial and real-activity patterns remain intact.

These results indicate that the benchmark propagation mechanism is not driven only by RBC, only by TD, or by the joint influence of the two largest banks. Instead, the responses survive when the shock is reconstructed from the remaining institutions. This supports the interpretation that the identified shock captures broader news about Canadian bank balance sheets rather than idiosyncratic variation from a single dominant bank.

%%%%%%%%%%%%%%%%%%%%%%%%%%%%%%%%%%%%%%%%%%%%%%%%%%%%%%%%%%%%%%%%%%%%%%
\subsection{Sectoral Propagation}
\label{subsec:robustness_sectoral}

We next study whether the real effects of bank net-worth shocks are concentrated in particular sectors of the Canadian economy. The benchmark results show that a favorable credit-supply bank net-worth shock is followed by an increase in aggregate real activity. This subsection decomposes that response by estimating local projections for sectoral measures of GDP. The goal is to assess whether the real-side propagation is broad-based or instead concentrated in sectors that are more sensitive to credit conditions, investment, inventories, or external demand.

We conduct the exercise in two steps. First, we distinguish between durable goods, non-durable goods, and services. Second, we study a more disaggregated set of sectors: industrial production, oil, mining, petroleum and gas, construction, retail trade, wholesale trade, and finance and insurance. In each case, the sectoral outcome is transformed as $100$ times the log of the corresponding activity index, and the dependent variable is the cumulative change from the month before the shock to horizon $h$. The regressions use the same normalized median rotational-angle credit-supply bank net-worth shock as in the benchmark specification. They also include the benchmark Canadian macro-financial controls, the COVID-period dummy, and six lags of the corresponding sectoral outcome. Thus, each sectoral response is estimated conditional on both aggregate Canadian macro-financial dynamics and the sector's own lagged dynamics.

The goods-versus-services responses are reported in Appendix Figure \ref{fig:appendix_goods_services}. The results suggest that the aggregate GDP response is driven more strongly by goods-producing sectors than by services. Durable goods display the largest and most cyclical response: output rises during the first year after a favorable bank net-worth shock, peaks around the medium-run horizon, and then reverses later in the horizon. Non-durable goods also increase initially, although the response is smaller. Services respond positively but more modestly and with less persistence. This pattern is consistent with the interpretation that credit-supply shocks operate more strongly through sectors whose production is more sensitive to financing conditions, inventory accumulation, and durable expenditure.

Appendix Figure \ref{fig:appendix_sectoral_gdp} provides additional sectoral detail. Industrial production increases after the shock and peaks around one year later, closely matching the timing of the aggregate GDP response. Oil, mining, petroleum and gas display a large positive response, suggesting that resource-related activity contributes importantly to the aggregate real-side propagation. Construction, retail trade, and wholesale trade also increase over the first year, consistent with stronger domestic demand and easier credit conditions. Finance and insurance responds positively but more moderately, with wider confidence bands later in the horizon.

Overall, the sectoral results reinforce the interpretation of the benchmark macroeconomic responses. Favorable credit-supply news about Canadian banks does not only affect financial prices; it also propagates to real activity in economically meaningful sectors. The response is strongest in goods-producing and resource-related sectors, while services respond more modestly. These patterns are consistent with a transmission mechanism in which bank net-worth shocks ease financial conditions, support credit-sensitive expenditure, and stimulate real activity across several parts of the Canadian economy.

%%%%%%%%%%%%%%%%%%%%%%%%%%%%%%%%%%%%%%%%%%%%%%%%%%%%%%%%%%%%%%%%%%%%%%
\subsection{Labor-Market Responses}
\label{subsec:robustness_labor_market}

Finally, we examine whether the real effects of bank net-worth shocks extend to labor-market outcomes. The benchmark results show that favorable credit-supply bank net-worth shocks raise real activity over the medium run. If this increase in activity reflects a meaningful expansion in production rather than only movements in asset prices or measured GDP, we should also observe improvements in labor-market conditions. We therefore estimate additional local projections for total employment, the unemployment rate, and total hours worked. The results are reported in Appendix Figure \ref{fig:appendix_labor_market}.

The specification follows the structure of the sectoral exercises. Employment and total hours worked are transformed as $100$ times log levels, and their impulse responses are measured as cumulative log changes relative to the month before the shock. The unemployment rate is measured in levels, and the dependent variable is the change in the unemployment rate relative to the month before the shock. Each regression includes the benchmark Canadian macro-financial controls, the COVID-period dummy, and six lags of the corresponding labor-market variable. The shock is the same normalized median rotational-angle credit-supply bank net-worth shock used in the benchmark specification.

Appendix Figure \ref{fig:appendix_labor_market} shows that favorable credit-supply bank net-worth shocks are followed by an improvement in labor-market conditions. Employment and total hours worked rise during the first year after the shock, with both responses peaking around one year after the event. The unemployment rate falls over the same horizon, reaching its largest decline when employment and hours are near their peak. These responses are consistent with the real-side propagation documented above: easier credit conditions are followed not only by higher GDP, but also by higher labor input and lower unemployment.

The effects are not permanent. Employment and hours gradually return toward zero and become negative later in the horizon, while the unemployment rate eventually rises. This reversal is consistent with the hump-shaped GDP and sectoral responses documented in the benchmark and additional-results exercises. Overall, the labor-market evidence reinforces the interpretation that bank net-worth shocks propagate beyond financial markets: favorable credit-supply news about Canadian banks is followed by stronger real activity, higher labor input, and lower unemployment in the medium run.

% -------------------------------------------------
% Conclusions
% -------------------------------------------------
\section{Conclusions} \label{sec:conclusions}

This paper studies whether news about banks' balance sheets propagates to aggregate financial conditions and macroeconomic activity. We construct a high-frequency measure of Canadian bank net-worth shocks using stock-price reactions around the earnings announcements of the six large Canadian banks. The empirical strategy is motivated by a simple model in which higher intermediary net worth raises bank equity valuations, expands credit supply, and lowers corporate borrowing spreads. This prediction guides our use of the co-movement between bank equity prices and Canadian corporate spreads to purge raw bank equity surprises from contaminating information and isolate the credit-supply component of bank net-worth news.

The Canadian setting is useful because the banking sector is highly concentrated. Earnings announcements by a small number of large institutions therefore reveal information about a quantitatively important part of the domestic intermediation sector. We show that these announcements are meaningful information events: announcement days feature substantially larger stock-price movements than non-announcement days, and timing-adjusted stock-price reactions are strongly related to Bloomberg earnings surprises. We aggregate bank-level surprises using lagged market-capitalization shares and identify the purged credit-supply shock using sign restrictions motivated by the model.

The main results show that favorable credit-supply bank net-worth shocks have economically meaningful aggregate effects. A one-standard-deviation favorable shock lowers Canadian corporate spreads, raises the real aggregate market capitalization of Canadian banks, increases broader equity prices, appreciates the Canadian dollar, and is followed by a gradual expansion in real activity. These responses are consistent with the intermediary balance-sheet channel: positive news about bank net worth eases corporate financial conditions and propagates to the real economy.

Several additional exercises support this interpretation. Raw bank equity surprises produce similar but less sharply interpretable responses, while shocks identified using a simpler poor man's sign-restriction procedure generate results close to the benchmark median-rotation shocks. In an IV local projection exercise, we use the purged credit-supply shock as an instrument for Canadian corporate spreads and find that bank-induced spread increases reduce bank valuations, lower equity prices, and contract real activity. The main results are also robust to alternative lag structures, deterministic trends, and sample restrictions. Additional outcomes show that the effects extend to labor markets and sectoral activity, with stronger responses in goods-producing and resource-related sectors.

Overall, the evidence suggests that bank earnings announcements contain macroeconomically relevant information. In a concentrated banking system, news about the profitability and capitalization of large banks is not merely bank-level valuation news: it affects corporate borrowing conditions, asset prices, and real activity. The findings highlight intermediary net worth as a source of aggregate fluctuations and show that high-frequency bank earnings news can be used to identify credit-supply shocks in settings where banks play a central role in financial intermediation.

\bigskip

\bigskip

\bigskip

\bigskip

\bigskip

\bigskip

\paragraph{Declaration of AI Use.}
The authors used OpenAI's ChatGPT as an auxiliary tool during manuscript preparation. Its use was limited to language editing, improving exposition and organization, and assisting with the preparation and revision of \LaTeX, Stata, and MATLAB code. The authors did not use AI tools to generate original data, conduct independent empirical analysis, create research results, or determine the paper's conclusions. All AI-assisted material was reviewed, edited, and verified by the authors, who take full responsibility for the content of the manuscript.

% -------------------------------------------------
% References
% -------------------------------------------------
\newpage
\bibliography{references.bib}

%%%%%%%%%%%%%%%%%%%%%%%%%%%%%%%%%%%%%%%%%%%%%%%%%%%%%%%%%%%%%%%%%%%%%%
% -------------------------------------------------
% Appendix
% -------------------------------------------------

\newpage
\appendix

%%%%%%%%%%%%%%%%%%%%%%%%%%%%%%%%%%%%%%%%%%%%%%%%%%%%%%%%%%%%%%%%%%%%%%
\section{Equilibrium Conditions and Proofs}
\label{appendix:proofs}

This appendix provides the derivations behind the equilibrium conditions and
comparative statics in Section \ref{sec:illustrative_model}. The main text
uses the model as an illustrative framework for the empirical strategy. Here,
we provide the formal details. We first characterize the maintained
equilibrium region in which banks supply positive credit, leverage constraints
bind, banks issue positive external equity, and equilibrium lending spreads
are positive. We then derive the symmetric Cournot equilibrium and prove the
comparative-static results used in the main text.

%%%%%%%%%%%%%%%%%%%%%%%%%%%%%%%%%%%%%%%%%%%%%%%%%%%%%%%%%%%%%%%%%%%%%%
\subsection{Maintained Equilibrium Region}
\label{appendix_constraints}

Bank \(i\)'s problem before imposing the binding leverage constraint is:
\[
\max_{\ell_i,e_i,d_i}
\quad
\Pi_i
=
\left[
R(L)-R_D
\right]\ell_i
+
R_D(n_i+e_i)
-
\frac{\kappa}{2}e_i^2,
\]
subject to the balance-sheet identity
\[
\ell_i=d_i+n_i+e_i,
\]
and the leverage constraint
\[
\ell_i\leq \lambda(n_i+e_i).
\]
The inverse demand for credit is
\[
R(L)=\bar A(1-L),
\]
where \(L=\sum_{j=1}^{N}\ell_j\). The deposit rate is pinned down by
households' Euler equation:
\[
R_D=\frac{1}{\beta}.
\]

Let \(\mu_i\geq0\) denote the Lagrange multiplier on the leverage constraint.
The Lagrangian is:
\[
\mathcal L_i
=
\left[
R(L)-R_D
\right]\ell_i
+
R_D(n_i+e_i)
-
\frac{\kappa}{2}e_i^2
+
\mu_i
\left[
\lambda(n_i+e_i)-\ell_i
\right].
\]
The first-order condition with respect to lending is:
\[
\frac{\partial \mathcal L_i}{\partial \ell_i}
=
R(L)
+
R'(L)\ell_i
-
R_D
-
\mu_i
=
0.
\]
The first-order condition with respect to external equity issuance is:
\[
\frac{\partial \mathcal L_i}{\partial e_i}
=
R_D
-
\kappa e_i
+
\lambda\mu_i
=
0.
\]
Therefore, when the leverage constraint binds, the multiplier satisfies:
\[
\mu_i
=
\frac{\kappa e_i-R_D}{\lambda}.
\]
A strictly positive shadow value of the leverage constraint requires
\[
\mu_i>0
\qquad\Longleftrightarrow\qquad
\kappa e_i>R_D.
\]
Thus, the leverage constraint has a positive shadow value when the marginal
cost of external equity issuance is sufficiently high relative to the deposit
funding cost.

Under a binding leverage constraint,
\[
\ell_i=\lambda(n_i+e_i),
\]
so external equity issuance is:
\[
e_i
=
\frac{\ell_i}{\lambda}-n_i.
\]
Substituting this expression for \(e_i\), together with the expression for
\(\mu_i\), into the lending first-order condition gives the same first-order
condition obtained below after substituting the binding leverage constraint
directly into profits. Thus, the Lagrangian formulation and the reduced
loan-supply problem are equivalent in the maintained binding region.

In the symmetric equilibrium characterized below,
\[
\ell_i=\ell^*,
\qquad
n_i=n,
\qquad
L^*=N\ell^*.
\]
Equilibrium lending per bank is:
\[
\ell^*
=
\frac{
\bar A
-
R_D\left(1-\frac{1}{\lambda}\right)
+
\frac{\kappa}{\lambda}n
}{
\bar A(N+1)
+
\frac{\kappa}{\lambda^2}
}.
\]
It is useful to define
\[
B_0
\equiv
\bar A
-
R_D\left(1-\frac{1}{\lambda}\right),
\qquad
D
\equiv
\bar A(N+1)
+
\frac{\kappa}{\lambda^2}.
\]
Then equilibrium lending can be written as
\[
\ell^*
=
\frac{
B_0+\frac{\kappa}{\lambda}n
}{D}.
\]
Since \(D>0\), positive lending requires
\[
B_0+\frac{\kappa}{\lambda}n>0.
\]

Positive external equity issuance requires
\[
e^*>0.
\]
Using
\[
e^*
=
\frac{\ell^*}{\lambda}-n,
\]
this condition is equivalent to
\[
\ell^*>\lambda n.
\]
Substituting the expression for \(\ell^*\) gives
\[
\frac{
B_0+\frac{\kappa}{\lambda}n
}{D}
>
\lambda n.
\]
Rearranging yields
\[
n
<
\bar n_e
\equiv
\frac{
B_0
}{
\lambda \bar A(N+1)
}.
\]
Thus, banks issue positive external equity when initial net worth is
sufficiently low relative to profitable lending opportunities. If \(n\geq0\),
a nonempty region with positive lending and positive external equity requires
\(B_0>0\). In that case, positive lending holds for all \(n\geq0\), while
positive external equity imposes the upper bound \(n<\bar n_e\).

For the leverage constraint to bind with a strictly positive multiplier, we
require
\[
\mu^*>0.
\]
Using
\[
\mu^*
=
\frac{\kappa e^*-R_D}{\lambda},
\]
this is equivalent to
\[
e^*>\frac{R_D}{\kappa}.
\]
Substituting \(e^*=\ell^*/\lambda-n\), this is equivalent to
\[
\frac{\ell^*}{\lambda}-n
>
\frac{R_D}{\kappa}.
\]
This condition imposes a stricter upper bound on \(n\) than the requirement
\(e^*>0\). Finally, equilibrium lending spreads are
\[
S^*
=
R^*-R_D,
\]
where
\[
R^*
=
\bar A(1-L^*).
\]
Positive equilibrium lending spreads require
\[
S^*
=
\bar A(1-L^*)-R_D
>0.
\]

Throughout the model analysis, we restrict attention to parameterizations for
which these conditions hold jointly: positive lending, positive external
equity issuance, a positive shadow value of the leverage constraint, and
positive equilibrium lending spreads. Equivalently, we work in an interior
region in which the leverage constraint binds, banks use both internal net
worth and costly external equity, and the loan market clears at a positive
spread over the deposit rate.

%%%%%%%%%%%%%%%%%%%%%%%%%%%%%%%%%%%%%%%%%%%%%%%%%%%%%%%%%%%%%%%%%%%%%%
\subsection{Cournot Equilibrium}

This subsection derives the symmetric Cournot equilibrium used in the main
text. Let
\[
L=\ell_i+L_{-i},
\qquad
L_{-i}\equiv\sum_{j\neq i}\ell_j.
\]
Under a binding leverage constraint, bank \(i\)'s profits can be written as
\[
\Pi_i
=
\left[
R(\ell_i+L_{-i})-R_D
\right]\ell_i
+
R_D
\left(
\frac{\ell_i}{\lambda}
\right)
-
\frac{\kappa}{2}
\left(
\frac{\ell_i}{\lambda}-n_i
\right)^2.
\]
Using \(R(L)=\bar A(1-L)\), this becomes
\[
\Pi_i
=
\left[
\bar A(1-\ell_i-L_{-i})-R_D
\right]\ell_i
+
R_D
\left(
\frac{\ell_i}{\lambda}
\right)
-
\frac{\kappa}{2}
\left(
\frac{\ell_i}{\lambda}-n_i
\right)^2.
\]
Bank \(i\) chooses \(\ell_i\) taking \(L_{-i}\) as given. The first-order
condition is:
\[
\frac{\partial \Pi_i}{\partial \ell_i}
=
\bar A(1-\ell_i-L_{-i})
-
\bar A\ell_i
-
R_D
+
\frac{R_D}{\lambda}
-
\frac{\kappa}{\lambda}
\left(
\frac{\ell_i}{\lambda}-n_i
\right)
=0.
\]
Equivalently,
\[
\bar A
-
\bar A L_{-i}
-
2\bar A\ell_i
-
R_D\left(1-\frac{1}{\lambda}\right)
-
\frac{\kappa}{\lambda}
\left(
\frac{\ell_i}{\lambda}-n_i
\right)
=0.
\]

In a symmetric equilibrium,
\[
\ell_i=\ell,
\qquad
n_i=n,
\qquad
L_{-i}=(N-1)\ell.
\]
Substituting these conditions into the first-order condition gives
\[
\bar A
-
\bar A(N+1)\ell
-
R_D
\left(
1-\frac{1}{\lambda}
\right)
-
\frac{\kappa}{\lambda}
\left(
\frac{\ell}{\lambda}-n
\right)
=0.
\]
Solving for \(\ell\) yields
\[
\ell^*
=
\frac{
\bar A
-
R_D\left(1-\frac{1}{\lambda}\right)
+
\frac{\kappa}{\lambda}n
}{
\bar A(N+1)
+
\frac{\kappa}{\lambda^2}
}.
\]
Aggregate lending, the equilibrium lending rate, and the equilibrium lending
spread are therefore
\[
L^*=N\ell^*,
\qquad
R^*=\bar A(1-L^*),
\qquad
S^*=R^*-R_D.
\]

%%%%%%%%%%%%%%%%%%%%%%%%%%%%%%%%%%%%%%%%%%%%%%%%%%%%%%%%%%%%%%%%%%%%%%
\subsection{Comparative Static: Common Net Worth and Credit Supply}

\begin{proposition}
In the maintained equilibrium region, a common increase in bank net worth
increases equilibrium lending per bank and aggregate credit supply.
\end{proposition}

\begin{proof}
Equilibrium lending per bank is
\[
\ell^*
=
\frac{
\bar A
-
R_D\left(1-\frac{1}{\lambda}\right)
+
\frac{\kappa}{\lambda}n
}{
\bar A(N+1)
+
\frac{\kappa}{\lambda^2}
}.
\]
Differentiating with respect to the common net-worth level \(n\) gives
\[
\frac{\partial \ell^*}{\partial n}
=
\frac{
\kappa/\lambda
}{
\bar A(N+1)+\kappa/\lambda^2
}.
\]
Since
\[
\kappa>0,
\qquad
\lambda>0,
\qquad
\bar A(N+1)+\frac{\kappa}{\lambda^2}>0,
\]
we have
\[
\frac{\partial \ell^*}{\partial n}>0.
\]
Aggregate lending is
\[
L^*=N\ell^*.
\]
Therefore
\[
\frac{\partial L^*}{\partial n}
=
N
\frac{\partial \ell^*}{\partial n}
>0.
\]
Thus, higher common bank net worth increases equilibrium lending per bank and
aggregate credit supply.
\end{proof}

%%%%%%%%%%%%%%%%%%%%%%%%%%%%%%%%%%%%%%%%%%%%%%%%%%%%%%%%%%%%%%%%%%%%%%
\subsection{Comparative Static: Common Net Worth and Lending Spreads}

\begin{proposition}
In the maintained equilibrium region, a common increase in bank net worth
lowers the equilibrium corporate borrowing spread.
\end{proposition}

\begin{proof}
The equilibrium corporate borrowing spread is
\[
S^*
=
R^*-R_D,
\]
where
\[
R^*
=
\bar A(1-L^*).
\]
Differentiating the equilibrium lending rate with respect to common bank net
worth gives
\[
\frac{\partial R^*}{\partial n}
=
-\bar A
\frac{\partial L^*}{\partial n}.
\]
From the previous proposition,
\[
\frac{\partial L^*}{\partial n}>0.
\]
Since \(\bar A>0\), it follows that
\[
\frac{\partial R^*}{\partial n}<0.
\]
The deposit rate is
\[
R_D=\frac{1}{\beta},
\]
and is therefore independent of bank net worth. Hence,
\[
\frac{\partial S^*}{\partial n}
=
\frac{\partial R^*}{\partial n}
<0.
\]
Thus, an increase in common bank net worth compresses equilibrium corporate
borrowing spreads.
\end{proof}

%%%%%%%%%%%%%%%%%%%%%%%%%%%%%%%%%%%%%%%%%%%%%%%%%%%%%%%%%%%%%%%%%%%%%%
\subsection{Comparative Static: Bank-Specific Net Worth and Bank Equity Values}

\begin{proposition}
In the maintained equilibrium region, a bank-specific increase in net worth
raises the equity valuation of the bank receiving the net-worth shock.
\end{proposition}

\begin{proof}
The empirical design uses bank-level earnings announcements. The relevant
valuation comparative static is therefore a bank-specific perturbation in
\(n_i\), evaluated around the symmetric equilibrium.

The Cournot first-order conditions can be written as a linear system in loan
quantities. For each bank \(k\),
\[
\left(
2\bar A+\frac{\kappa}{\lambda^2}
\right)\ell_k
+
\bar A
\sum_{j\neq k}\ell_j
=
\bar A
-
R_D\left(1-\frac{1}{\lambda}\right)
+
\frac{\kappa}{\lambda}n_k.
\]
Define
\[
a
\equiv
\bar A+\frac{\kappa}{\lambda^2},
\qquad
D
\equiv
\bar A(N+1)+\frac{\kappa}{\lambda^2}
=
a+N\bar A.
\]
The coefficient matrix of the Cournot system is
\[
a I_N+\bar A \mathbf 1\mathbf 1',
\]
whose inverse is
\[
\left(a I_N+\bar A \mathbf 1\mathbf 1'\right)^{-1}
=
\frac{1}{a}I_N
-
\frac{\bar A}{aD}\mathbf 1\mathbf 1'.
\]
A marginal increase in \(n_i\) therefore implies
\[
\frac{\partial \ell_i^*}{\partial n_i}
=
\frac{\kappa}{\lambda}
\frac{
a+(N-1)\bar A
}{
aD
}
>0,
\]
and, for each \(j\neq i\),
\[
\frac{\partial \ell_j^*}{\partial n_i}
=
-
\frac{\kappa}{\lambda}
\frac{
\bar A
}{
aD
}
<0.
\]
Thus, the bank receiving the net-worth shock expands lending, while rival
banks reduce lending because loan quantities are strategic substitutes. The
aggregate lending response is
\[
\frac{\partial L^*}{\partial n_i}
=
\frac{\partial \ell_i^*}{\partial n_i}
+
\sum_{j\neq i}
\frac{\partial \ell_j^*}{\partial n_i}
=
\frac{\kappa/\lambda}{D}
>0.
\]
Since \(R^*=\bar A(1-L^*)\), the equilibrium lending rate and spread fall:
\[
\frac{\partial S^*}{\partial n_i}
=
\frac{\partial R^*}{\partial n_i}
=
-\bar A
\frac{\partial L^*}{\partial n_i}
<0.
\]

It remains to show that the value of bank \(i\) rises. Under a binding
leverage constraint, bank \(i\)'s optimized profits are
\[
\Pi_i^*
=
\left[
R(L^*)-R_D
\right]\ell_i^*
+
R_D
\left(
\frac{\ell_i^*}{\lambda}
\right)
-
\frac{\kappa}{2}
\left(
\frac{\ell_i^*}{\lambda}-n_i
\right)^2.
\]
When differentiating optimized profits with respect to \(n_i\), the envelope
theorem eliminates the effect of bank \(i\)'s own optimal lending response.
However, competitors' lending also responds and therefore affects the loan
rate faced by bank \(i\). Hence,
\[
\frac{d\Pi_i^*}{dn_i}
=
\frac{\partial \Pi_i}{\partial n_i}
+
\frac{\partial \Pi_i}{\partial L_{-i}}
\frac{\partial L_{-i}^*}{\partial n_i}.
\]
The direct effect is
\[
\frac{\partial \Pi_i}{\partial n_i}
=
\kappa
\left(
\frac{\ell_i^*}{\lambda}-n_i
\right)
=
\kappa e_i^*.
\]
The effect through competitors' lending is
\[
\frac{\partial \Pi_i}{\partial L_{-i}}
=
R'(L^*)\ell_i^*
=
-\bar A\ell_i^*,
\]
and
\[
\frac{\partial L_{-i}^*}{\partial n_i}
=
\sum_{j\neq i}
\frac{\partial \ell_j^*}{\partial n_i}
=
-(N-1)
\frac{\kappa}{\lambda}
\frac{\bar A}{aD}
<0.
\]
Therefore,
\[
\frac{d\Pi_i^*}{dn_i}
=
\kappa e_i^*
+
\bar A\ell_i^*
(N-1)
\frac{\kappa}{\lambda}
\frac{\bar A}{aD}
>0,
\]
where the inequality uses \(e_i^*>0\) in the maintained equilibrium region.

Bank equity values reflect the present discounted value of expected bank
profits:
\[
V_i
=
\beta \mathbb E[\Pi_i^*].
\]
The derivative above is positive in every state in the maintained region.
Thus, differentiating under the expectation gives
\[
\frac{dV_i}{dn_i}
=
\beta
\mathbb E
\left[
\frac{d\Pi_i^*}{dn_i}
\right]
>0.
\]
A bank-specific increase in net worth therefore raises the equity valuation
of the bank receiving the shock.
\end{proof}

%%%%%%%%%%%%%%%%%%%%%%%%%%%%%%%%%%%%%%%%%%%%%%%%%%%%%%%%%%%%%%%%%%%%%%
\subsection{Summary of the Identification Implication}

The comparative statics above imply that favorable news about bank net worth
raises bank equity valuations and lowers corporate borrowing spreads. A common
increase in bank net worth expands aggregate lending and compresses spreads:
\[
n \uparrow
\quad\Rightarrow\quad
L^* \uparrow,
\qquad
S^* \downarrow.
\]
A bank-specific increase in net worth raises the equity valuation of the bank
receiving the shock, increases aggregate lending, and lowers spreads:
\[
n_i \uparrow
\quad\Rightarrow\quad
V_i \uparrow,
\qquad
L^* \uparrow,
\qquad
S^* \downarrow.
\]
This latter comparative static is the one most closely aligned with the
earnings-announcement design, since earnings releases reveal news about a
particular bank. It delivers the sign pattern used in the empirical strategy:
favorable credit-supply news about bank net worth raises bank equity
valuations and lowers corporate borrowing spreads. The raw stock-price
reaction around an earnings announcement may also contain information about
expected loan demand, aggregate fundamentals, borrower quality, or risk
premia. The empirical strategy therefore uses the co-movement between bank
equity prices and Canadian corporate spreads to isolate the component of bank
equity news most closely aligned with the credit-supply mechanism in the
model.

%%%%%%%%%%%%%%%%%%%%%%%%%%%%%%%%%%%%%%%%%%%%%%%%%%%%%%%%%%%%%%%%%%%%%%
\newpage
\section{Additional Details on Banks' Surprises}
\label{appendix:additional_details_surprises}

This appendix provides additional details on the construction and validation of the bank-level equity surprises used in the empirical analysis.

Table \ref{tab:appendix_bank_price_summary_tests} reports bank-level summary statistics comparing stock-price movements on earnings-announcement days and non-announcement days. For this descriptive comparison, we use the same close-to-close stock-return window for both groups, so that announcement and non-announcement days are compared on an equal basis. For each bank, we report the close-to-close price change, its absolute value, and its square, together with tests for differences in means, variances, and medians across announcement and non-announcement days. These statistics provide supporting evidence that earnings announcements are information events for Canadian bank equity prices. The high-frequency shocks used in the empirical analysis are constructed separately using timing-adjusted event-window stock-price reactions around earnings announcements.

%%%%%%%%%%%%%%%%%%%%%%%%%%%%%%%%%%%%%%%%%%%%%%%%%%%%%%%%%%%%%%%%%%%%%%
\begin{landscape}
\begin{table}[p]
\centering
\caption{Bank-Level Stock-Price Movements Around Earnings Announcements: Same-Window Comparison}
\label{tab:appendix_bank_price_summary_tests}
\begin{threeparttable}
\scriptsize
\setlength{\tabcolsep}{3pt}
\resizebox{\linewidth}{!}{
\begin{tabular}{llccccccccccc}
\toprule
& & \multicolumn{4}{c}{Non-Announcement Days} 
& \multicolumn{4}{c}{Announcement Days}
& \multicolumn{3}{c}{Tests -- $p$-values} \\
\cmidrule(lr){3-6} \cmidrule(lr){7-10} \cmidrule(lr){11-13}
Bank & Variable 
& $N$ & Mean & Median & SD
& $N$ & Mean & Median & SD
& $t$-test & Variance-test & Median-test \\
\midrule

RBC & Close-to-close price change
& 5,582 & 0.03 & 0.07 & 1.26
& 94 & -0.08 & -0.24 & 2.51
& 0.377 & 0.000 & 0.203 \\

RBC & Absolute close-to-close price change
& 5,582 & 0.81 & 0.56 & 0.97
& 94 & 1.96 & 1.52 & 1.56
& 0.000 & 0.000 & 0.000 \\

RBC & Squared close-to-close price change
& 5,582 & 1.59 & 0.31 & 7.49
& 94 & 6.26 & 2.30 & 8.73
& 0.000 & 0.026 & 0.000 \\

\midrule

Scotiabank & Close-to-close price change
& 5,596 & 0.03 & 0.06 & 1.29
& 80 & -0.40 & -0.52 & 2.32
& 0.004 & 0.000 & 0.004 \\

Scotiabank & Absolute close-to-close price change
& 5,596 & 0.83 & 0.56 & 0.99
& 80 & 1.77 & 1.41 & 1.54
& 0.000 & 0.000 & 0.000 \\

Scotiabank & Squared close-to-close price change
& 5,596 & 1.67 & 0.32 & 7.60
& 80 & 5.47 & 2.00 & 9.66
& 0.000 & 0.001 & 0.000 \\

\midrule

TD & Close-to-close price change
& 5,582 & 0.04 & 0.07 & 1.29
& 94 & -0.05 & 0.13 & 2.31
& 0.528 & 0.000 & 1.000 \\

TD & Absolute close-to-close price change
& 5,582 & 0.84 & 0.56 & 0.99
& 94 & 1.79 & 1.47 & 1.45
& 0.000 & 0.000 & 0.000 \\

TD & Squared close-to-close price change
& 5,582 & 1.68 & 0.32 & 7.71
& 94 & 5.29 & 2.17 & 8.73
& 0.000 & 0.073 & 0.000 \\

\midrule

BMO & Close-to-close price change
& 5,586 & 0.02 & 0.07 & 1.35
& 90 & 0.13 & 0.15 & 2.78
& 0.445 & 0.000 & 0.253 \\

BMO & Absolute close-to-close price change
& 5,586 & 0.85 & 0.57 & 1.05
& 90 & 2.18 & 2.15 & 1.71
& 0.000 & 0.000 & 0.000 \\

BMO & Squared close-to-close price change
& 5,586 & 1.82 & 0.32 & 8.91
& 90 & 7.65 & 4.64 & 12.12
& 0.000 & 0.000 & 0.000 \\

\midrule

National Bank & Close-to-close price change
& 5,585 & 0.04 & 0.08 & 1.37
& 91 & 0.26 & 0.04 & 2.79
& 0.127 & 0.000 & 0.980 \\

National Bank & Absolute close-to-close price change
& 5,585 & 0.84 & 0.56 & 1.08
& 91 & 2.21 & 1.89 & 1.71
& 0.000 & 0.000 & 0.000 \\

National Bank & Squared close-to-close price change
& 5,585 & 1.88 & 0.32 & 10.88
& 91 & 7.79 & 3.58 & 12.31
& 0.000 & 0.079 & 0.000 \\

\midrule

CIBC & Close-to-close price change
& 5,582 & 0.03 & 0.07 & 1.36
& 94 & 0.17 & -0.20 & 3.08
& 0.354 & 0.000 & 0.683 \\

CIBC & Absolute close-to-close price change
& 5,582 & 0.86 & 0.57 & 1.06
& 94 & 2.52 & 2.18 & 1.77
& 0.000 & 0.000 & 0.000 \\

CIBC & Squared close-to-close price change
& 5,582 & 1.85 & 0.32 & 8.68
& 94 & 9.42 & 4.74 & 11.97
& 0.000 & 0.000 & 0.000 \\

\midrule

Pooled & Close-to-close price change
& 33,513 & 0.03 & 0.07 & 1.32
& 543 & 0.01 & -0.10 & 2.65
& 0.750 & 0.000 & 0.181 \\

Pooled & Absolute close-to-close price change
& 33,513 & 0.84 & 0.56 & 1.02
& 543 & 2.08 & 1.78 & 1.64
& 0.000 & 0.000 & 0.000 \\

Pooled & Squared close-to-close price change
& 33,513 & 1.75 & 0.32 & 8.63
& 543 & 7.01 & 3.18 & 10.76
& 0.000 & 0.000 & 0.000 \\

\bottomrule
\end{tabular}
}
\floatfoot{\textit{Notes:} This table reports bank-level summary statistics for stock-price movements on Canadian bank earnings-announcement and non-announcement days. For this descriptive comparison, stock-price changes are measured using the same close-to-close window for both groups, $\Delta p^{CC}_{i,t}=100[\log(P^{close}_{i,t})-\log(P^{close}_{i,t-1})]$. The sample is restricted to observations for which both the stock-price change and the Canadian corporate OAS change are observed. The final three columns report $p$-values from tests comparing announcement and non-announcement days: a two-sample $t$-test for equality of means, a variance test, and a Wilcoxon rank-sum test for differences in medians. The high-frequency bank equity surprises used in the empirical analysis are not constructed from close-to-close returns; they are constructed using timing-adjusted event-window stock-price reactions around earnings announcements.}
\end{threeparttable}
\end{table}
\end{landscape}
%%%%%%%%%%%%%%%%%%%%%%%%%%%%%%%%%%%%%%%%%%%%%%%%%%%%%%%%%%%%%%%%%%%%%%

\bigskip

\bigskip

%%%%%%%%%%%%%%%%%%%%%%%%%%%%%%%%%%%%%%%%%%%%%%%%%%%%%%%%%%%%%%%%%%%%%%
\begin{figure}[p]
\centering

\begin{subfigure}{0.48\textwidth}
\centering
\includegraphics[width=\textwidth]{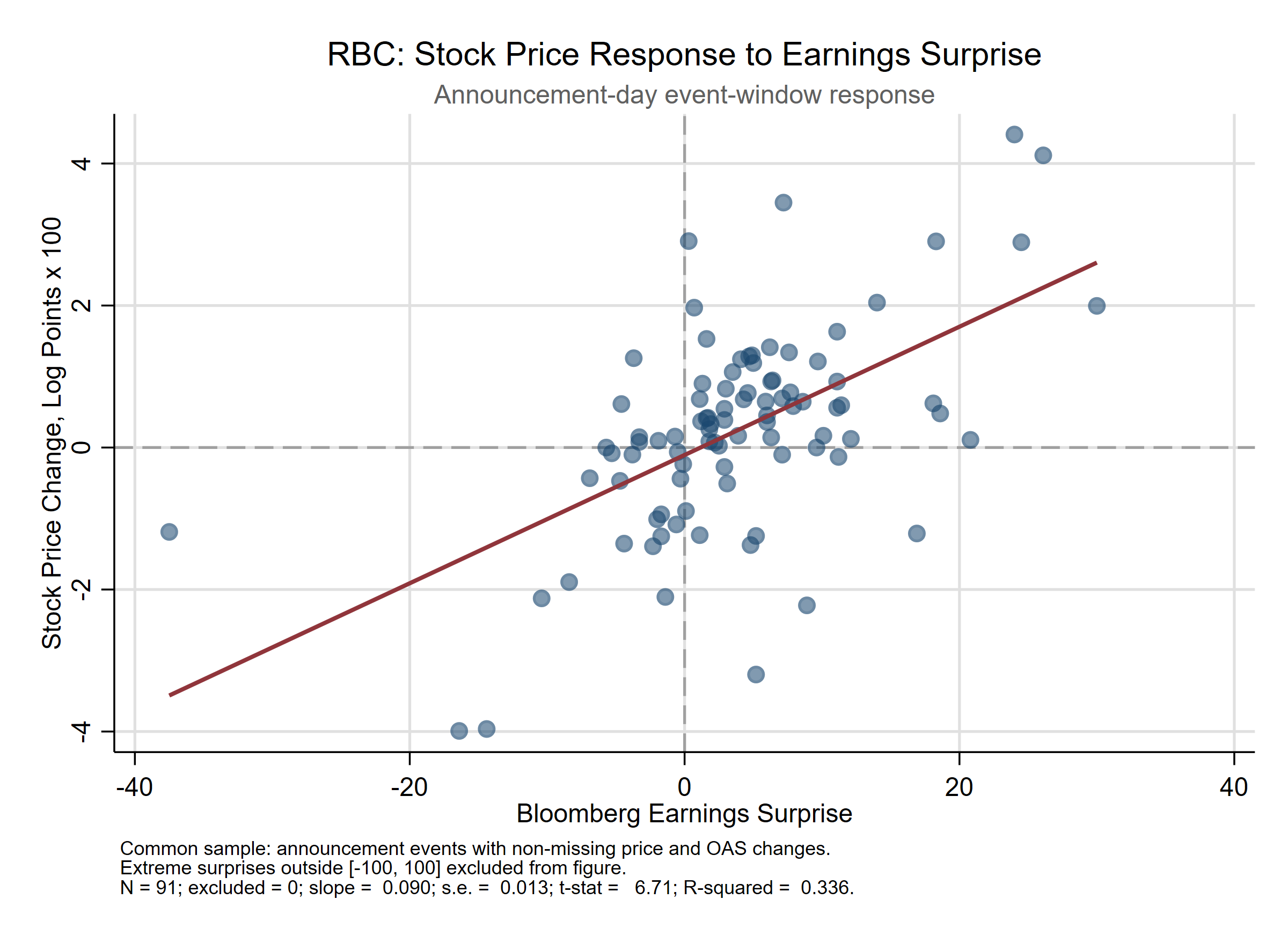}
\caption{RBC}
\end{subfigure}
\hfill
\begin{subfigure}{0.48\textwidth}
\centering
\includegraphics[width=\textwidth]{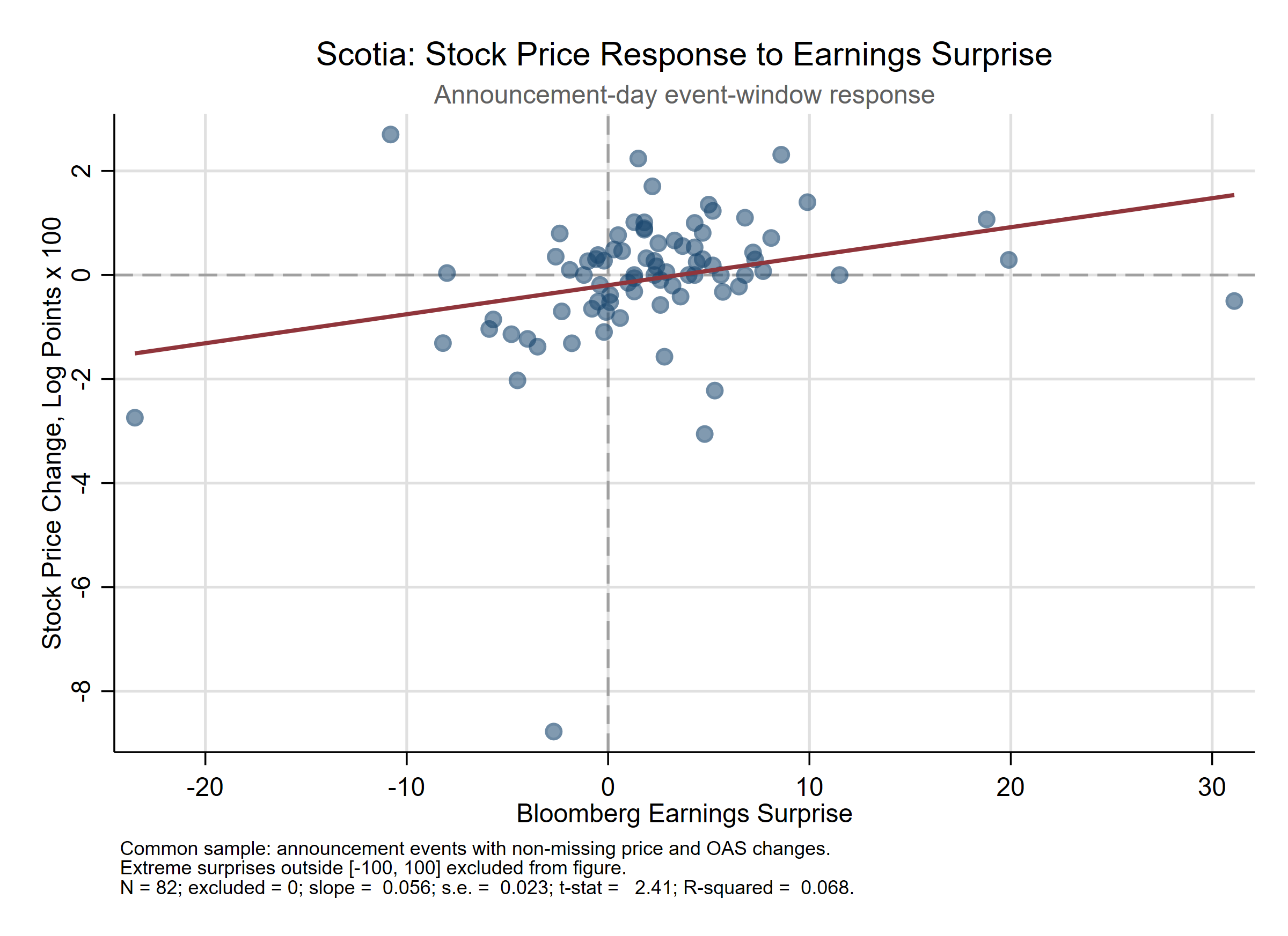}
\caption{Scotiabank}
\end{subfigure}

\vspace{0.35cm}

\begin{subfigure}{0.48\textwidth}
\centering
\includegraphics[width=\textwidth]{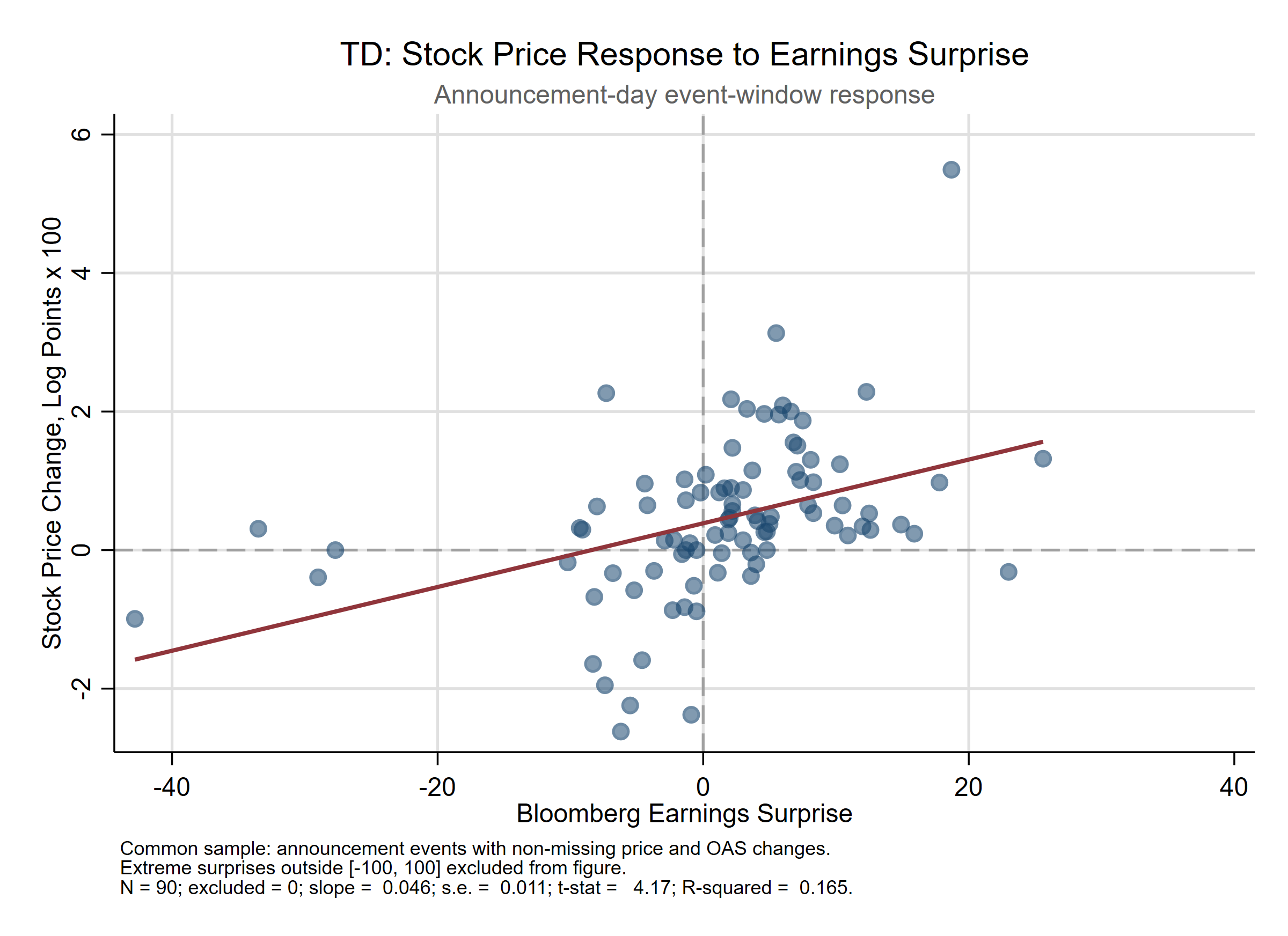}
\caption{TD}
\end{subfigure}
\hfill
\begin{subfigure}{0.48\textwidth}
\centering
\includegraphics[width=\textwidth]{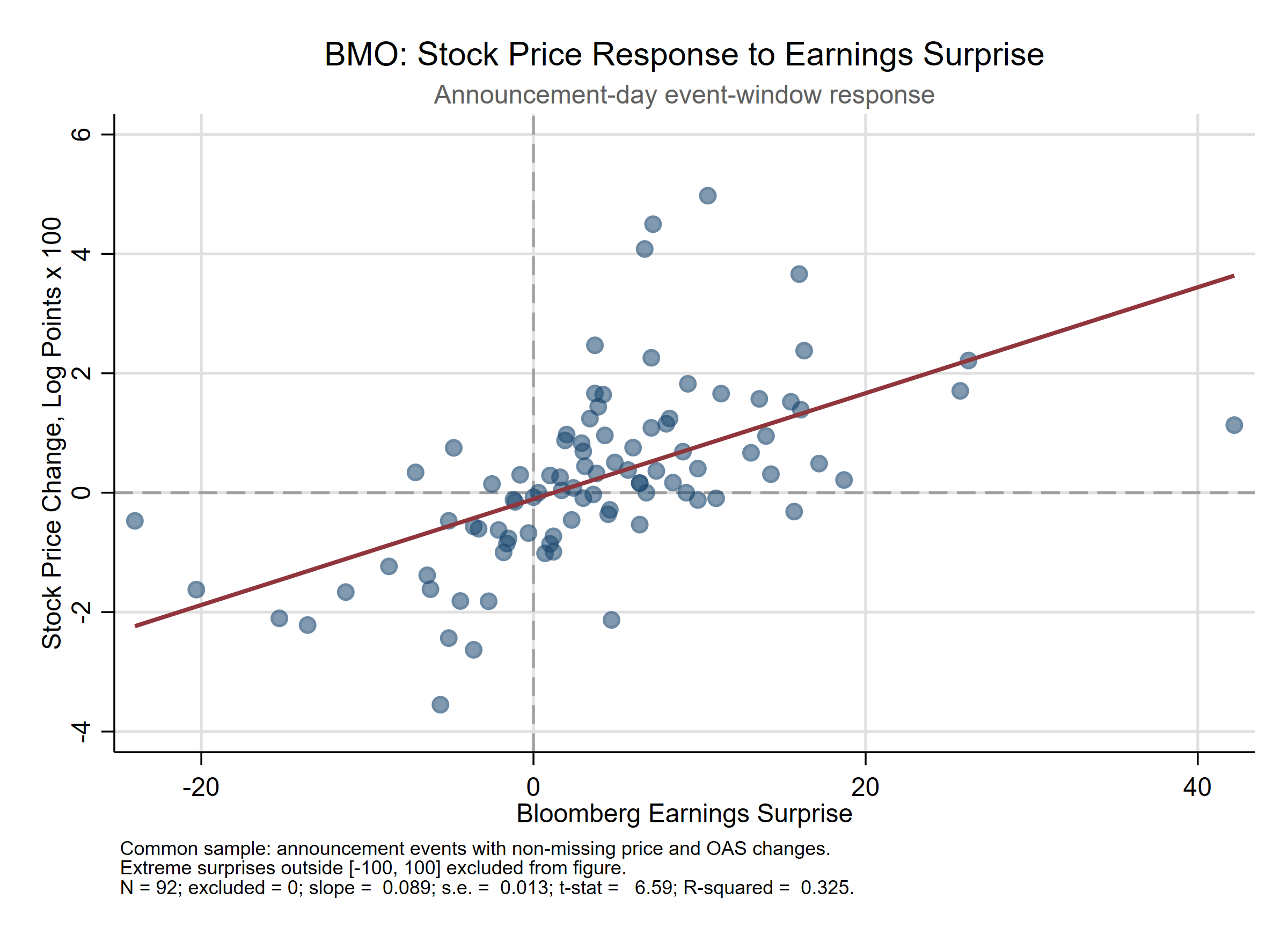}
\caption{BMO}
\end{subfigure}

\vspace{0.35cm}

\begin{subfigure}{0.48\textwidth}
\centering
\includegraphics[width=\textwidth]{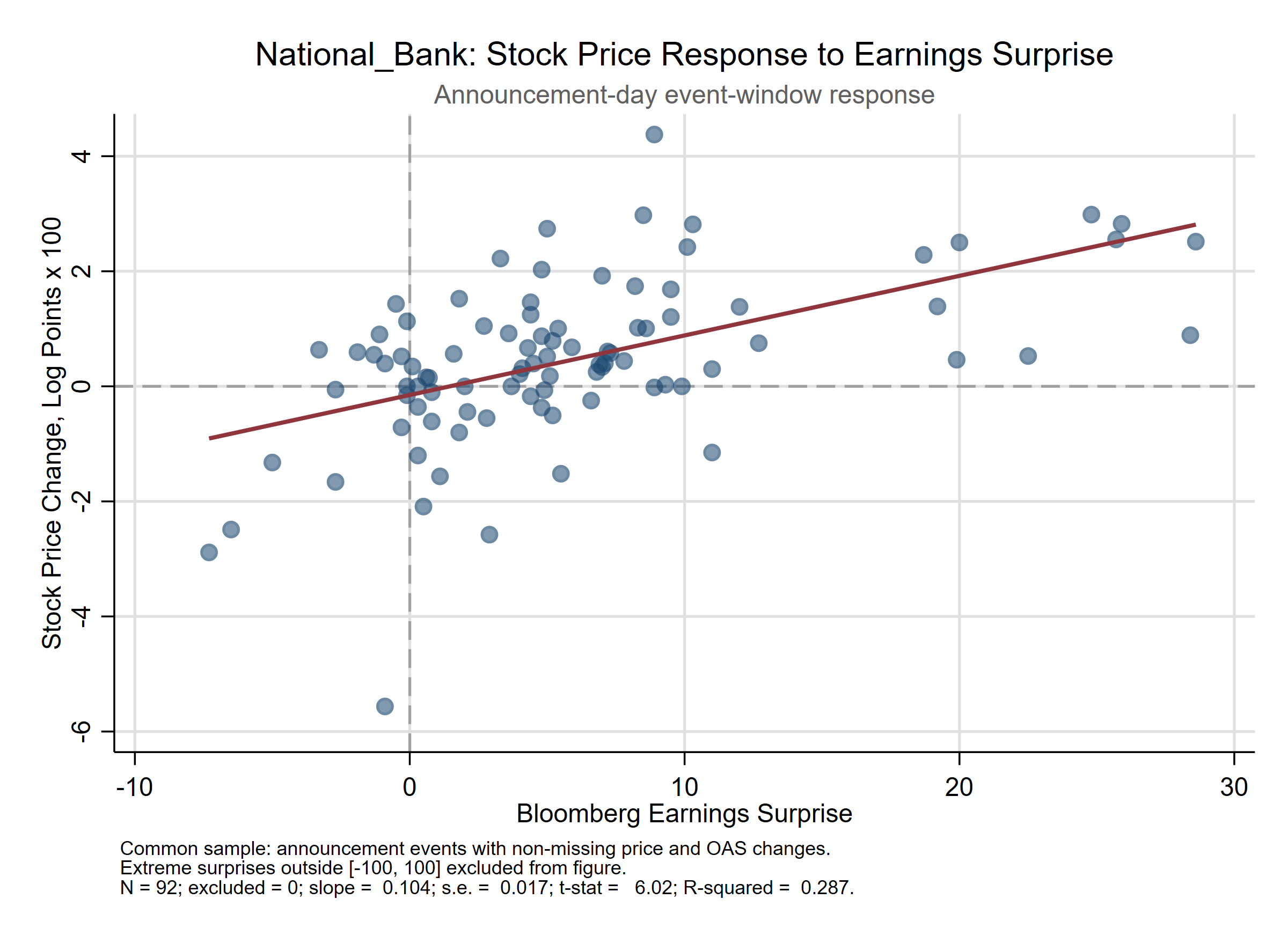}
\caption{National Bank}
\end{subfigure}
\hfill
\begin{subfigure}{0.48\textwidth}
\centering
\includegraphics[width=\textwidth]{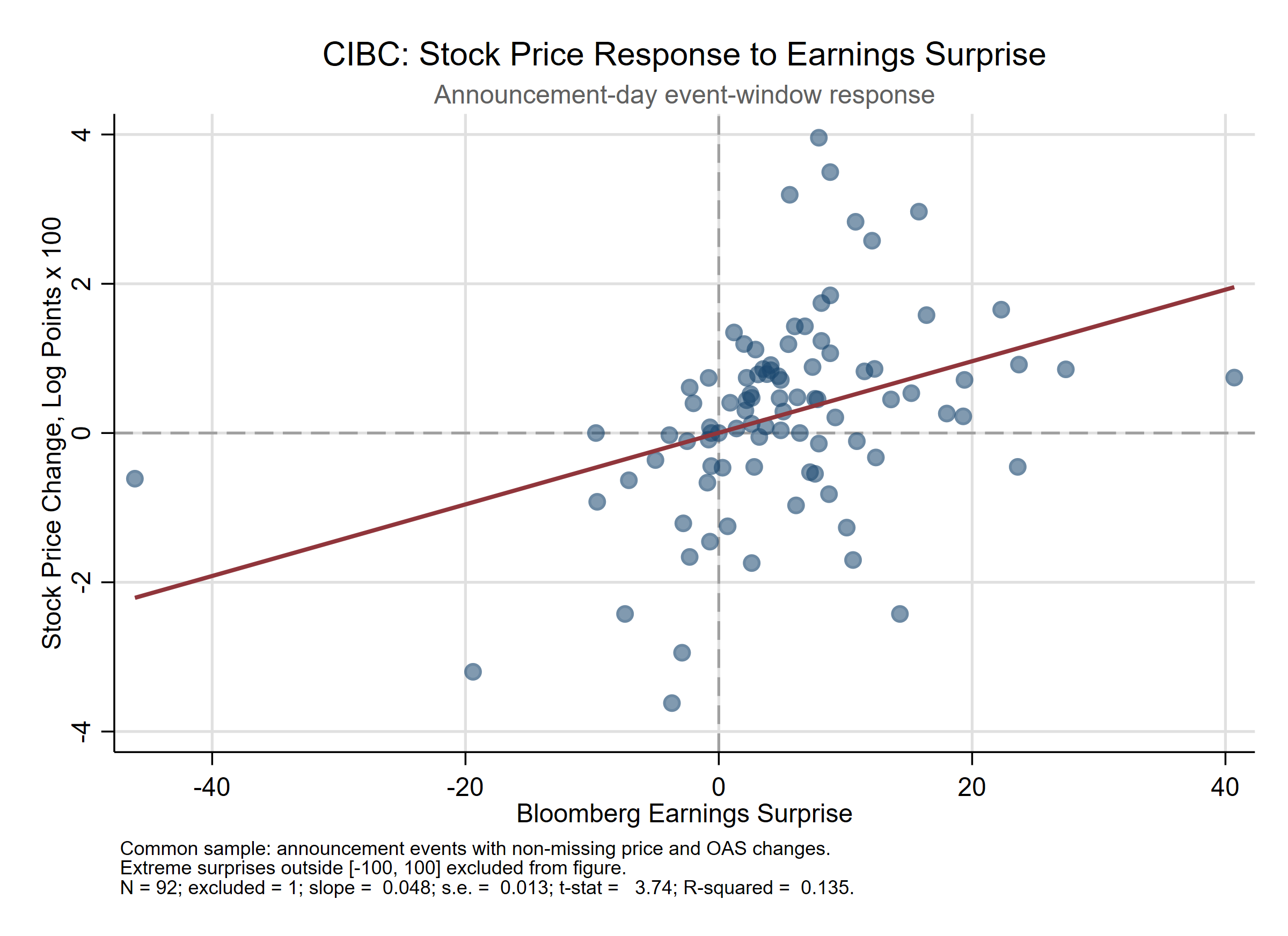}
\caption{CIBC}
\end{subfigure}

\caption{Bank-Level Stock-Price Responses to Bloomberg Earnings Surprises}
\label{fig:scatter_price_surprise_banks}
\floatfoot{\textbf{Notes:} This figure plots announcement-window stock-price changes against Bloomberg earnings surprises separately for each Canadian bank. Stock-price changes are event-window log price changes multiplied by 100. Bloomberg earnings surprises are constructed as the difference between realized earnings and the pre-announcement analyst consensus estimate. The sample is restricted to announcement events with non-missing stock-price and OAS changes. Extreme surprises outside $[-100,100]$ are excluded from the figure. Each panel reports the fitted OLS regression line and associated regression statistics.}
\end{figure}
%%%%%%%%%%%%%%%%%%%%%%%%%%%%%%%%%%%%%%%%%%%%%%%%%%%%%%%%%%%%%%%%%%%%%%

%%%%%%%%%%%%%%%%%%%%%%%%%%%%%%%%%%%%%%%%%%%%%%%%%%%%%%%%%%%%%%%%%%%%%%
\subsection{Serial Correlation of Bank Equity Surprises}
\label{appendix:autocorrelation_surprises}

This appendix examines whether the timing-adjusted announcement-window stock-price reactions display systematic serial correlation. This exercise is useful because the empirical strategy interprets these stock-price reactions as high-frequency surprises around bank earnings announcements. If these surprises were strongly predictable from previous earnings-announcement reactions, this would weaken their interpretation as innovations. We therefore test for serial correlation using event-time lags within each bank. Specifically, for each bank, we order earnings announcements chronologically and test whether the stock-price reaction around a given announcement predicts reactions around subsequent announcements. The sample is restricted to the common set of events with non-missing stock-price and OAS changes.

Table \ref{tab:appendix_autocorr_tests} reports three sets of tests. First, we estimate an AR(1) specification for each bank and test whether the first event-time lag predicts the current announcement-window stock-price reaction. Second, we jointly test whether lags one through four are statistically significant. Third, we report Ljung--Box tests for serial correlation up to four event-time lags. We also report pooled specifications with bank fixed effects, as well as a pooled-sequence robustness check that treats the full sample as a single ordered sequence of events.

%%%%%%%%%%%%%%%%%%%%%%%%%%%%%%%%%%%%%%%%%%%%%%%%%%%%%%%%%%%%%%%%%%%%%%
\begin{table}[H]
\centering
\caption{Autocorrelation Tests for Announcement-Window Stock-Price Reactions}
\label{tab:appendix_autocorr_tests}
\begin{threeparttable}
\scriptsize
\setlength{\tabcolsep}{5pt}
\resizebox{\textwidth}{!}{
\begin{tabular}{lllrcc}
\toprule
Sample & Bank & Test & $N$ & Statistic & $p$-value \\
\midrule

Bank-level & BMO & AR(1) coefficient & 91 & 1.986 & 0.050 \\
Bank-level & BMO & Joint AR(1)--AR(4) & 88 & 1.772 & 0.142 \\
Bank-level & BMO & Ljung--Box Q(4) &  & 7.955 & 0.093 \\

\midrule

Bank-level & CIBC & AR(1) coefficient & 92 & -0.708 & 0.481 \\
Bank-level & CIBC & Joint AR(1)--AR(4) & 89 & 1.054 & 0.385 \\
Bank-level & CIBC & Ljung--Box Q(4) &  & 4.363 & 0.359 \\

\midrule

Bank-level & National Bank & AR(1) coefficient & 91 & -0.913 & 0.364 \\
Bank-level & National Bank & Joint AR(1)--AR(4) & 88 & 1.299 & 0.277 \\
Bank-level & National Bank & Ljung--Box Q(4) &  & 4.788 & 0.310 \\

\midrule

Bank-level & RBC & AR(1) coefficient & 90 & 0.987 & 0.326 \\
Bank-level & RBC & Joint AR(1)--AR(4) & 87 & 0.407 & 0.803 \\
Bank-level & RBC & Ljung--Box Q(4) &  & 2.182 & 0.702 \\

\midrule

Bank-level & Scotiabank & AR(1) coefficient & 81 & -0.682 & 0.497 \\
Bank-level & Scotiabank & Joint AR(1)--AR(4) & 78 & 1.558 & 0.195 \\
Bank-level & Scotiabank & Ljung--Box Q(4) &  & 7.861 & 0.097 \\

\midrule

Bank-level & TD & AR(1) coefficient & 89 & 1.218 & 0.226 \\
Bank-level & TD & Joint AR(1)--AR(4) & 86 & 1.924 & 0.114 \\
Bank-level & TD & Ljung--Box Q(4) &  & 6.526 & 0.163 \\

\midrule

Pooled & All banks & AR(1), bank FE & 534 & 0.604 & 0.572 \\
Pooled & All banks & Joint AR(1)--AR(4), bank FE & 516 & 4.026 & 0.079 \\
Pooled sequence & All banks & Joint AR(1)--AR(4) & 536 & 2.964 & 0.019 \\

\bottomrule
\end{tabular}
}
\floatfoot{\textit{Notes:} This table reports serial-correlation tests for timing-adjusted announcement-window stock-price reactions. Event-time lags are defined within each bank by ordering earnings announcements chronologically. The AR(1) test reports the $t$-statistic and $p$-value for the first event-time lag. The joint AR(1)--AR(4) test reports the $F$-statistic and $p$-value from a joint test that the first four event-time lags are equal to zero. The Ljung--Box test reports the Q-statistic and $p$-value for serial correlation up to four event-time lags. The pooled specifications include bank fixed effects and cluster standard errors by bank. The pooled-sequence test treats the full sample as a single sequence ordered by date and bank and is included only as a descriptive robustness check.}
\end{threeparttable}
\end{table}
%%%%%%%%%%%%%%%%%%%%%%%%%%%%%%%%%%%%%%%%%%%%%%%%%%%%%%%%%%%%%%%%%%%%%%

The results provide little evidence of systematic autocorrelation in announcement-window stock-price reactions. At the bank level, the AR(1) coefficient is statistically insignificant for all banks except BMO, where the evidence is only marginal. The joint AR(1)--AR(4) tests fail to reject the null of no serial correlation for every individual bank, and the Ljung--Box tests also generally fail to reject. The pooled specification with bank fixed effects similarly shows no evidence of a significant AR(1) coefficient. The joint pooled AR(1)--AR(4) test is only marginally significant, while the pooled-sequence test rejects more strongly; however, the latter treats events from different banks as a single sequence and is therefore less directly aligned with the bank-level event-time structure used in the construction of the shocks. Overall, the evidence supports the interpretation of the announcement-window stock-price reactions as high-frequency surprises rather than predictable components of bank equity returns.

%%%%%%%%%%%%%%%%%%%%%%%%%%%%%%%%%%%%%%%%%%%%%%%%%%%%%%%%%%%%%%%%%%%%%%
\begin{figure}[H]
\centering
\includegraphics[width=0.80\textwidth]{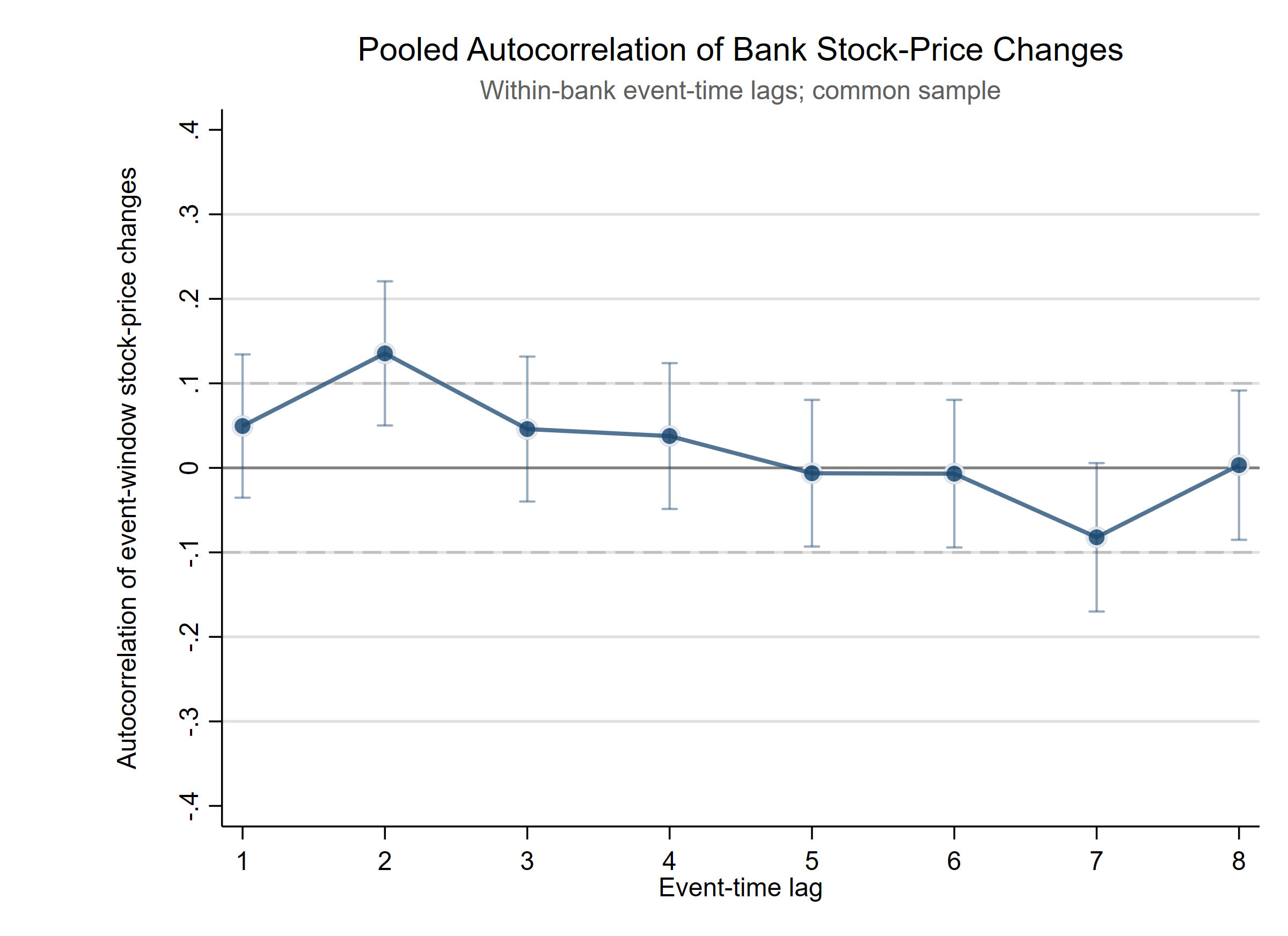}
\caption{Pooled Autocorrelation of Announcement-Window Stock-Price Reactions}
\label{fig:appendix_autocorr_pooled}
\floatfoot{\textbf{Notes:} This figure reports autocorrelations of timing-adjusted announcement-window stock-price reactions using within-bank event-time lags. The sample pools all banks and is restricted to earnings-announcement events with non-missing stock-price and OAS changes. Vertical bars report approximate 95 percent confidence intervals.}
\end{figure}

%%%%%%%%%%%%%%%%%%%%%%%%%%%%%%%%%%%%%%%%%%%%%%%%%%%%%%%%%%%%%%%%%%%%%%
%%%%%%%%%%%%%%%%%%%%%%%%%%%%%%%%%%%%%%%%%%%%%%%%%%%%%%%%%%%%%%%%%%%%%%
\begin{figure}[p]
\centering

\begin{subfigure}{0.5\textwidth}
\centering
\includegraphics[width=\textwidth]{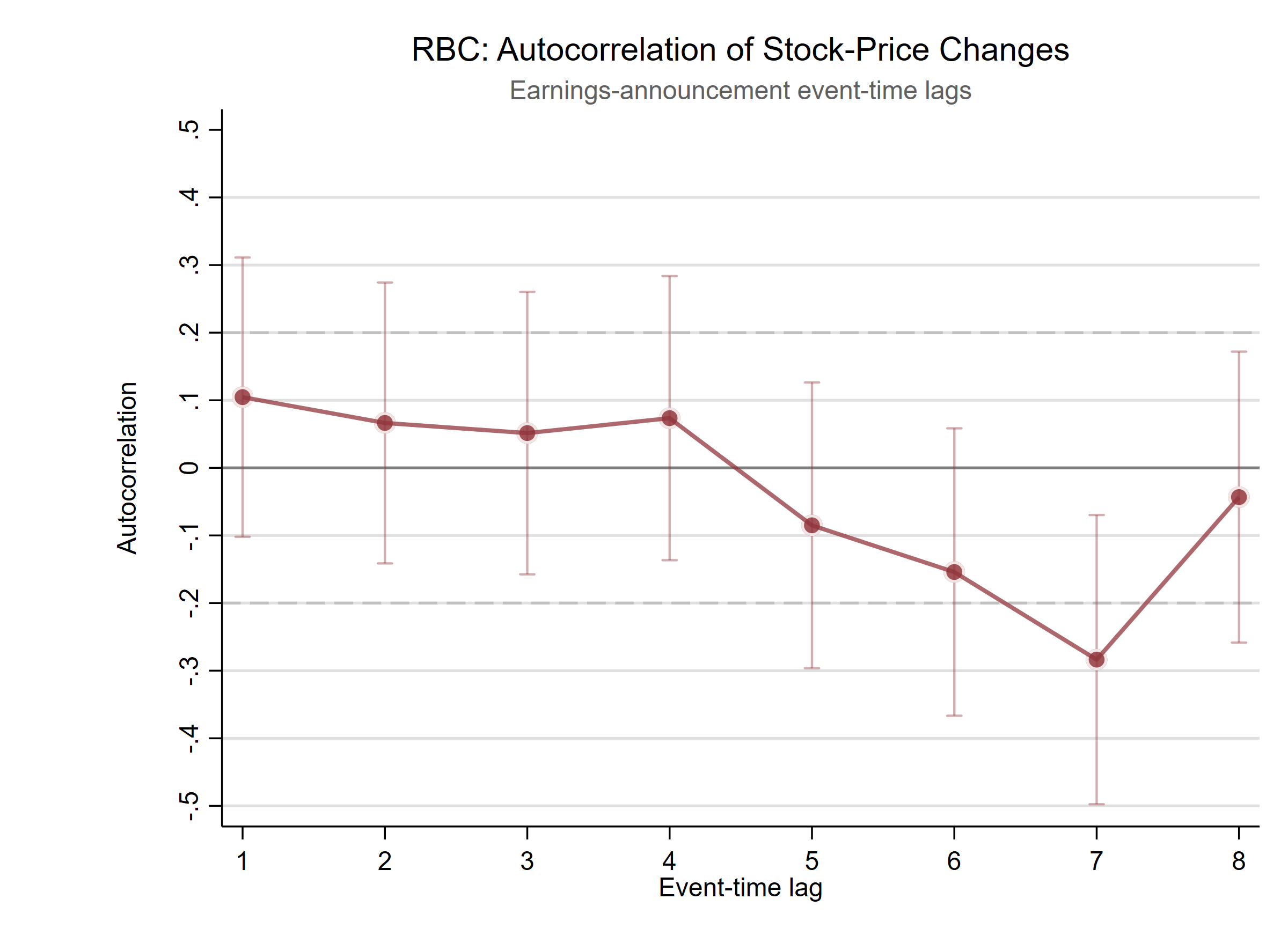}
\caption{RBC}
\end{subfigure}

\vspace{0.25cm}

\begin{subfigure}{0.5\textwidth}
\centering
\includegraphics[width=\textwidth]{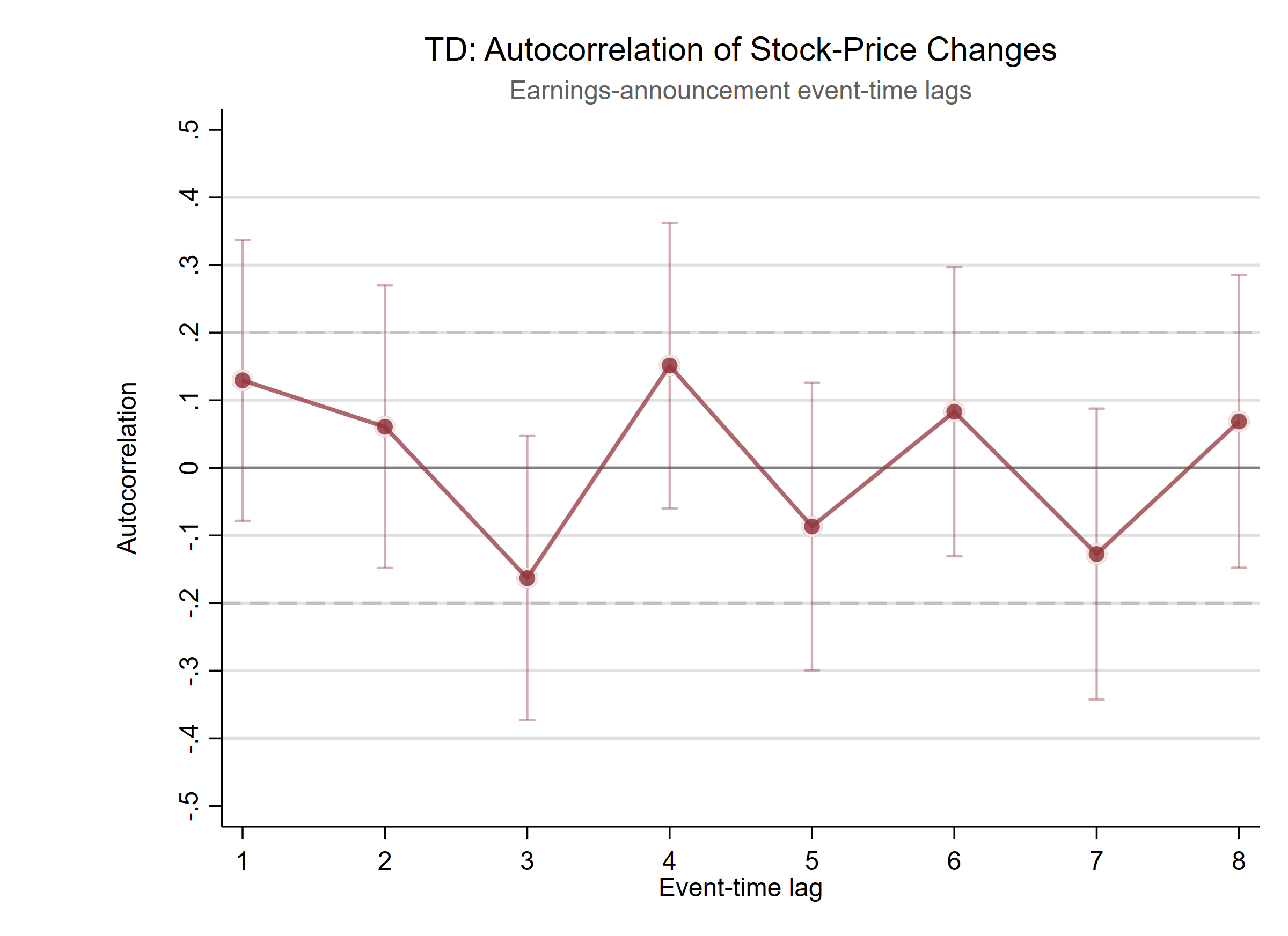}
\caption{TD}
\end{subfigure}

\vspace{0.25cm}

\begin{subfigure}{0.5\textwidth}
\centering
\includegraphics[width=\textwidth]{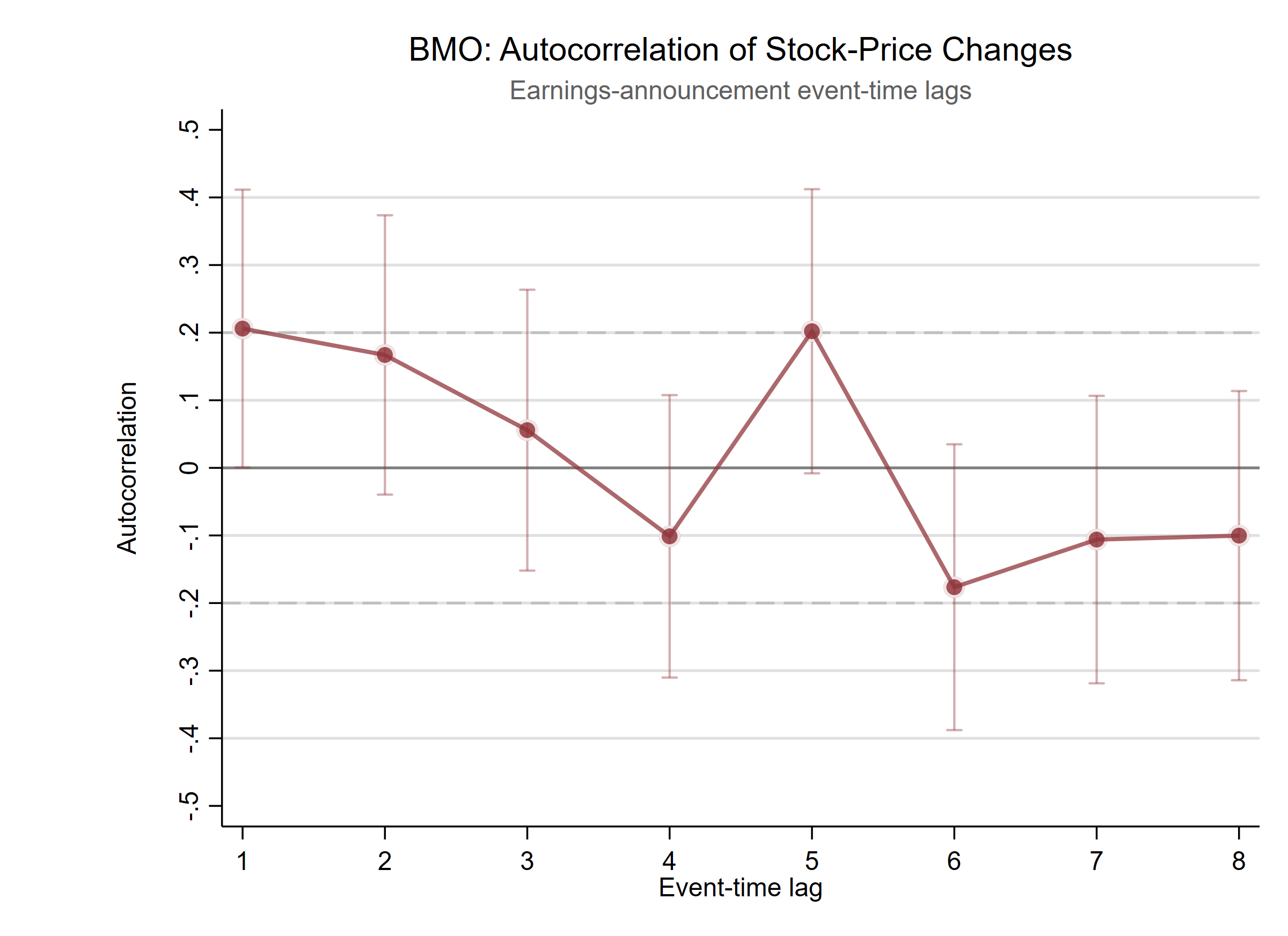}
\caption{BMO}
\end{subfigure}

\caption{Bank-Level Autocorrelation of Announcement-Window Stock-Price Reactions: Large Banks}
\label{fig:appendix_autocorr_bank_panel_1}
\floatfoot{\textbf{Notes:} This figure reports bank-level autocorrelations of timing-adjusted announcement-window stock-price reactions. Event-time lags are defined separately within each bank by ordering earnings announcements chronologically. The sample is restricted to earnings-announcement events with non-missing stock-price and OAS changes. Vertical bars report approximate 95 percent confidence intervals.}
\end{figure}
%%%%%%%%%%%%%%%%%%%%%%%%%%%%%%%%%%%%%%%%%%%%%%%%%%%%%%%%%%%%%%%%%%%%%%

%%%%%%%%%%%%%%%%%%%%%%%%%%%%%%%%%%%%%%%%%%%%%%%%%%%%%%%%%%%%%%%%%%%%%%
\begin{figure}[p]
\centering

\begin{subfigure}{0.5\textwidth}
\centering
\includegraphics[width=\textwidth]{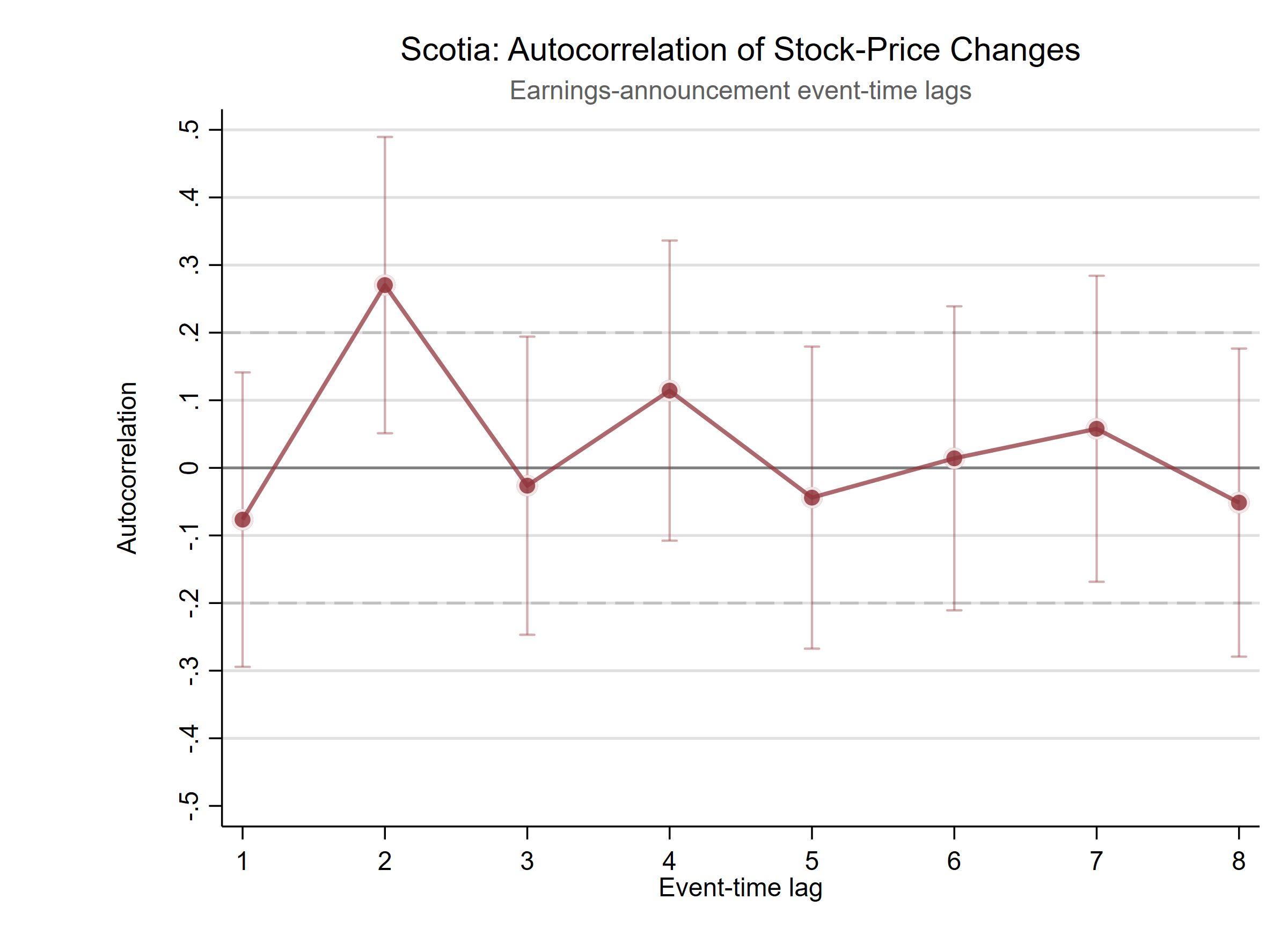}
\caption{Scotiabank}
\end{subfigure}

\vspace{0.25cm}

\begin{subfigure}{0.5\textwidth}
\centering
\includegraphics[width=\textwidth]{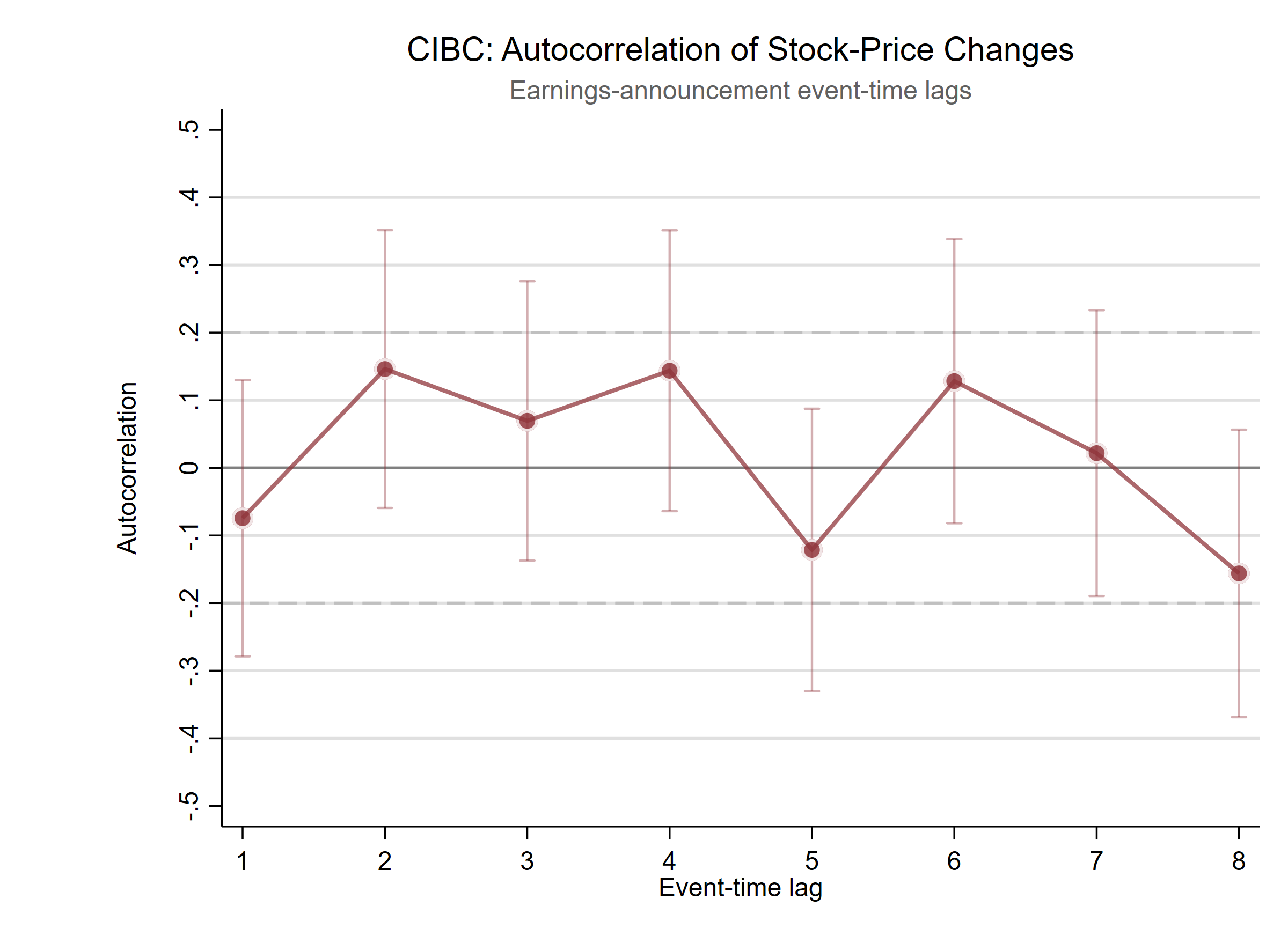}
\caption{CIBC}
\end{subfigure}

\vspace{0.25cm}

\begin{subfigure}{0.5\textwidth}
\centering
\includegraphics[width=\textwidth]{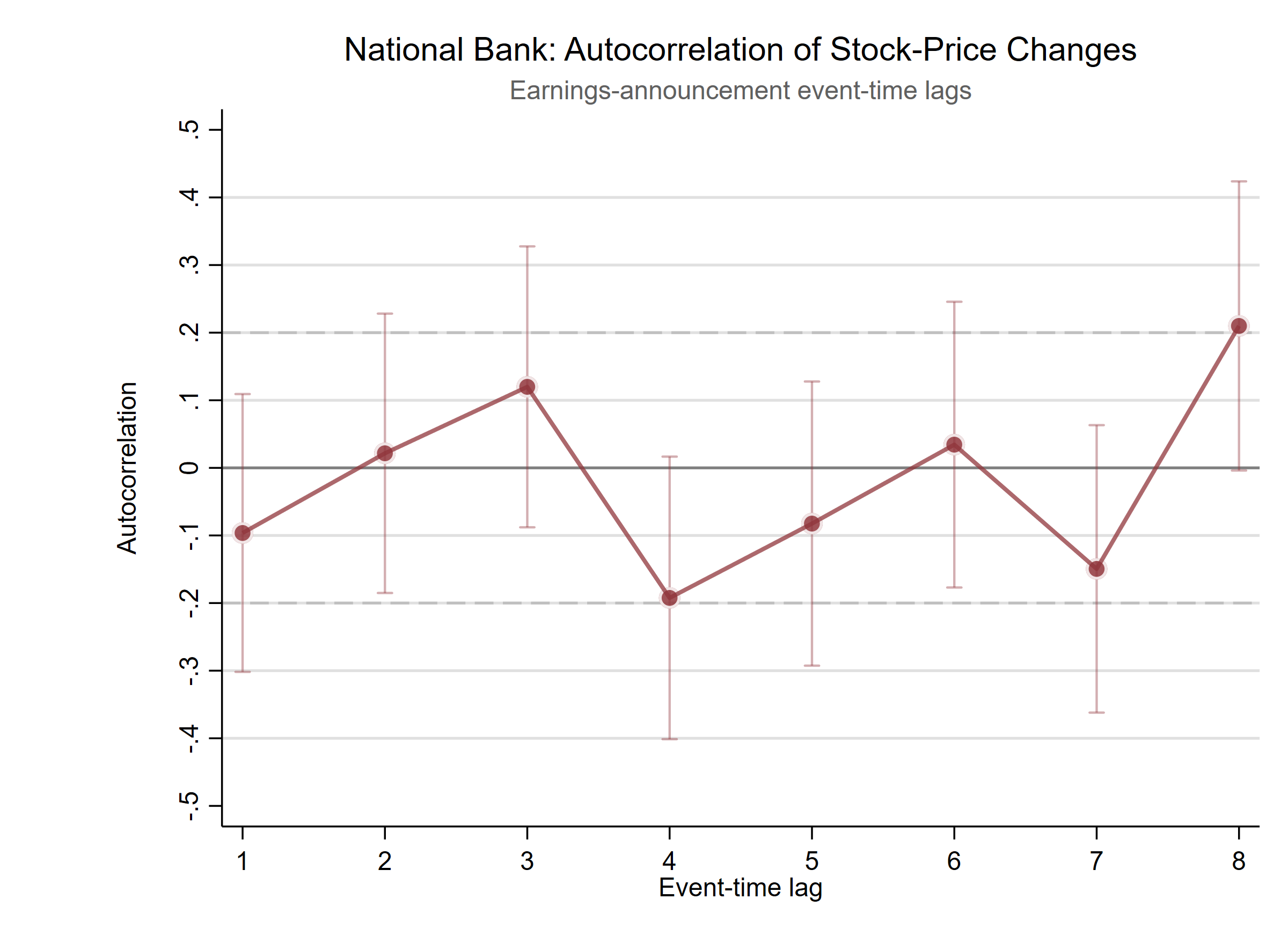}
\caption{National Bank}
\end{subfigure}

\caption{Bank-Level Autocorrelation of Announcement-Window Stock-Price Reactions: Additional Banks}
\label{fig:appendix_autocorr_bank_panel_2}
\floatfoot{\textbf{Notes:} This figure reports bank-level autocorrelations of timing-adjusted announcement-window stock-price reactions. Event-time lags are defined separately within each bank by ordering earnings announcements chronologically. The sample is restricted to earnings-announcement events with non-missing stock-price and OAS changes. Vertical bars report approximate 95 percent confidence intervals.}
\end{figure}
%%%%%%%%%%%%%%%%%%%%%%%%%%%%%%%%%%%%%%%%%%%%%%%%%%%%%%%%%%%%%%%%%%%%%%
%%%%%%%%%%%%%%%%%%%%%%%%%%%%%%%%%%%%%%%%%%%%%%%%%%%%%%%%%%%%%%%%%%%%%%

Figure \ref{fig:appendix_autocorr_pooled} presents the pooled autocorrelation function. The estimated autocorrelations are small and fluctuate around zero across event-time lags. Most estimates lie close to zero and within conventional confidence bands, reinforcing the evidence from Table \ref{tab:appendix_autocorr_tests} that the stock-price reactions do not display strong persistence.

%%%%%%%%%%%%%%%%%%%%%%%%%%%%%%%%%%%%%%%%%%%%%%%%%%%%%%%%%%%%%%%%%%%%%%

Figures \ref{fig:appendix_autocorr_bank_panel_1} and \ref{fig:appendix_autocorr_bank_panel_2} show the corresponding autocorrelation functions separately by bank. The bank-level estimates are generally centered around zero and do not exhibit a consistent pattern of persistence across institutions. Some individual autocorrelations are positive or negative at particular lags, but these movements are not systematic across banks. This pattern is consistent with the interpretation that earnings-announcement stock-price reactions primarily capture event-specific news rather than predictable serial dynamics.

%%%%%%%%%%%%%%%%%%%%%%%%%%%%%%%%%%%%%%%%%%%%%%%%%%%%%%%%%%%%%%%%%%%%%%
%%%%%%%%%%%%%%%%%%%%%%%%%%%%%%%%%%%%%%%%%%%%%%%%%%%%%%%%%%%%%%%%%%%%%%
%%%%%%%%%%%%%%%%%%%%%%%%%%%%%%%%%%%%%%%%%%%%%%%%%%%%%%%%%%%%%%%%%%%%%%
\subsection{Co-Movement of Corporate Spreads \& Bank Equity Surprises}
\label{appendix:comovement_spreads}

This subsection provides additional evidence on the co-movement between Canadian corporate credit spreads and bank equity surprises around earnings announcements. The main text presents the pooled relationship between timing-adjusted announcement-window stock-price reactions and event-window changes in the Canadian corporate OAS. Here, we report the corresponding bank-level scatter plots to assess whether the negative co-movement documented in the pooled sample is driven by a particular institution or is instead a broader feature of the data. These figures provide supporting evidence for the sign-restriction strategy used in the empirical analysis: favorable credit-supply news should raise bank equity prices and lower corporate spreads, while credit-demand news should move bank equity prices and spreads in the same direction.

%%%%%%%%%%%%%%%%%%%%%%%%%%%%%%%%%%%%%%%%%%%%%%%%%%%%%%%%%%%%%%%%%%%%%%
\begin{figure}[p]
\centering

\begin{subfigure}{0.48\textwidth}
\centering
\includegraphics[width=\textwidth]{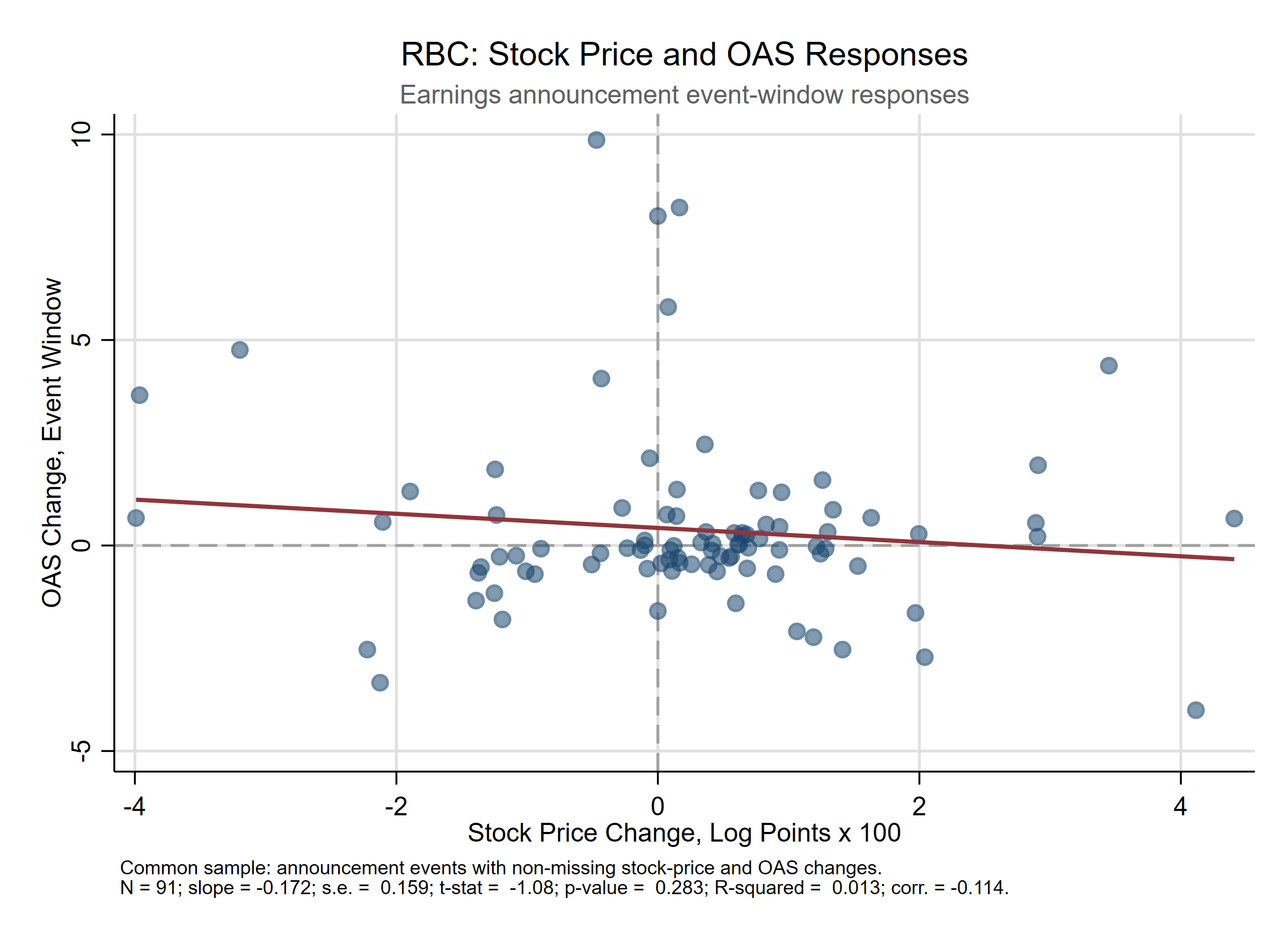}
\caption{RBC}
\end{subfigure}
\hfill
\begin{subfigure}{0.48\textwidth}
\centering
\includegraphics[width=\textwidth]{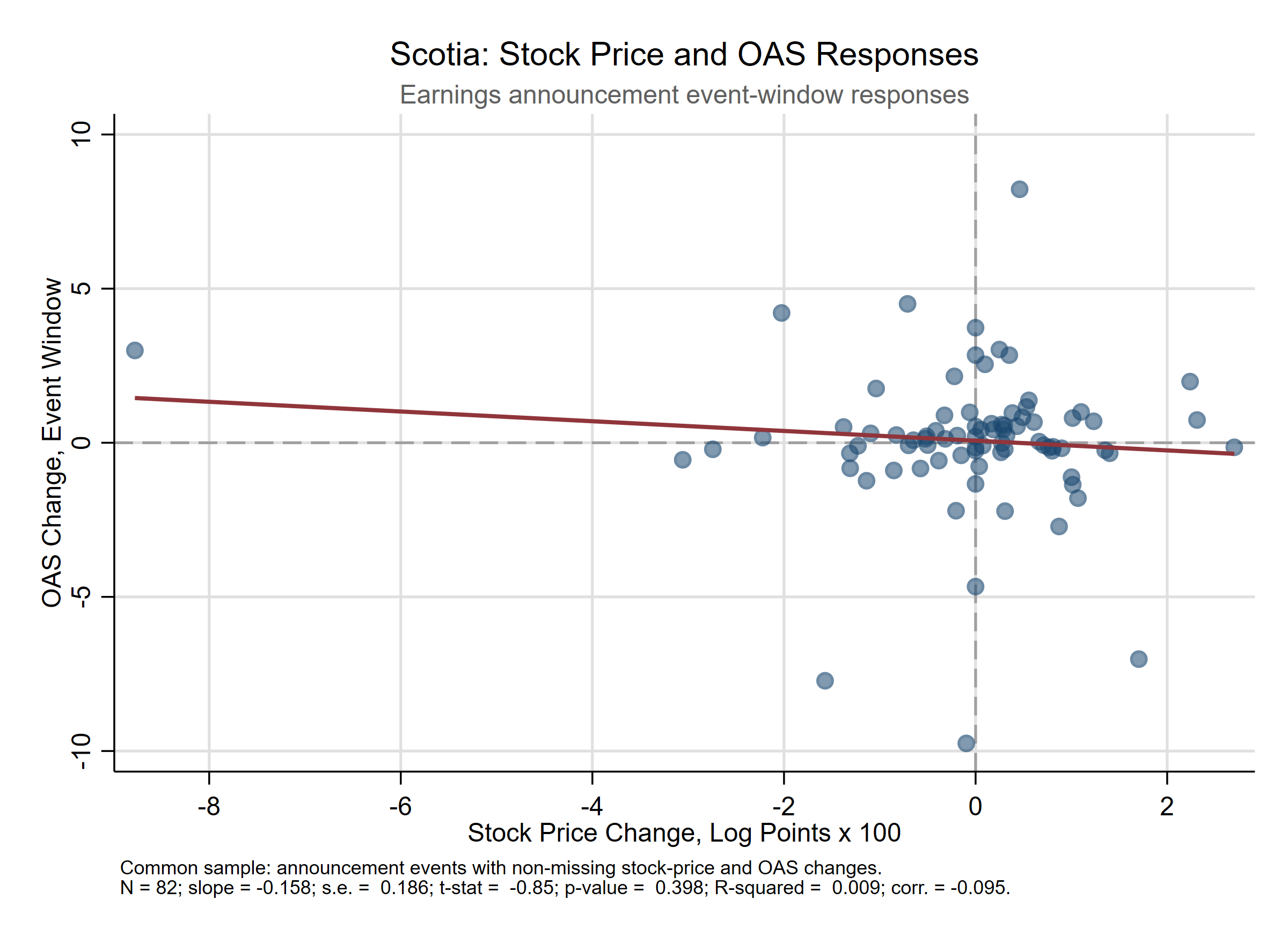}
\caption{Scotiabank}
\end{subfigure}

\vspace{0.35cm}

\begin{subfigure}{0.48\textwidth}
\centering
\includegraphics[width=\textwidth]{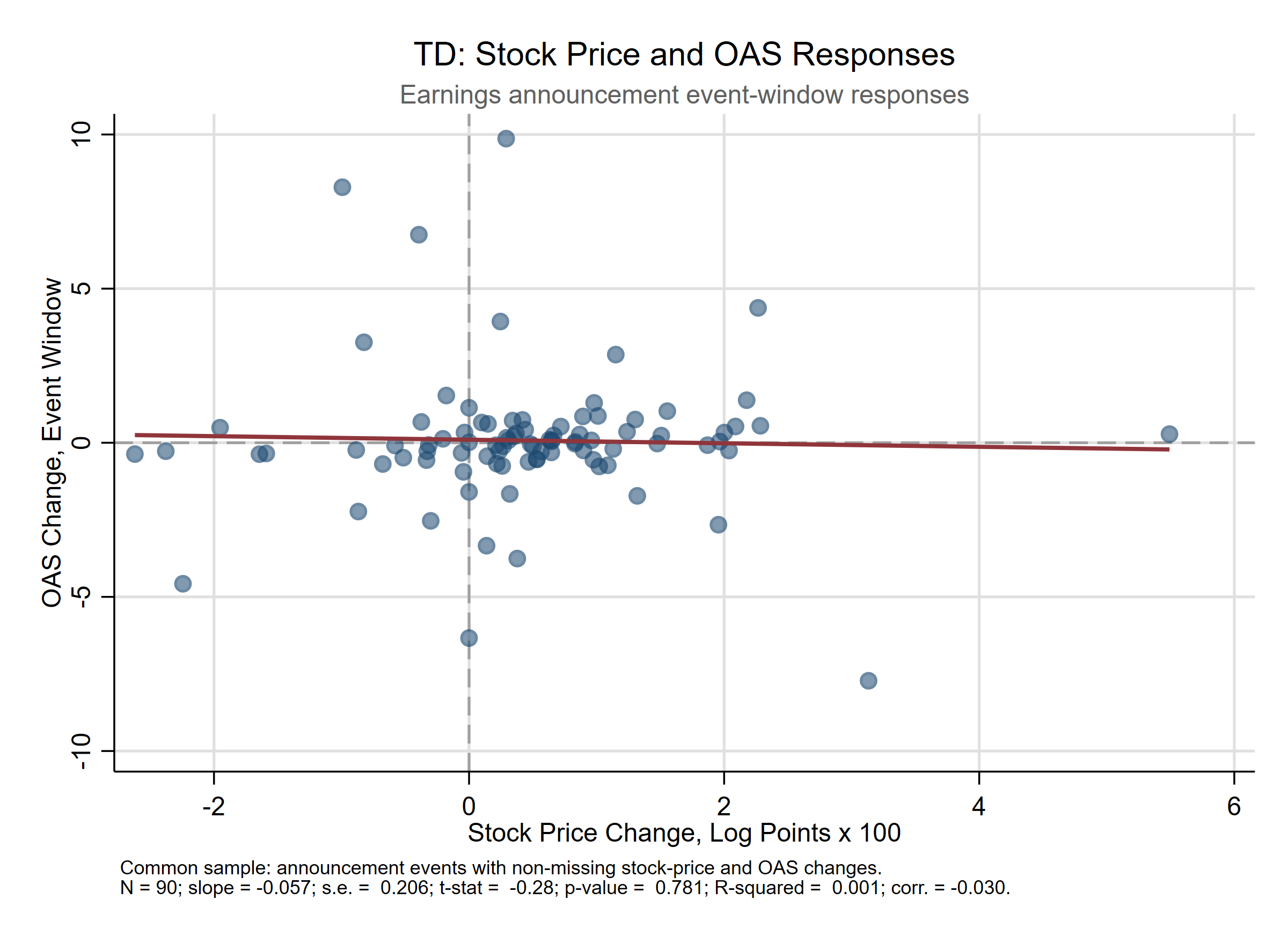}
\caption{TD}
\end{subfigure}
\hfill
\begin{subfigure}{0.48\textwidth}
\centering
\includegraphics[width=\textwidth]{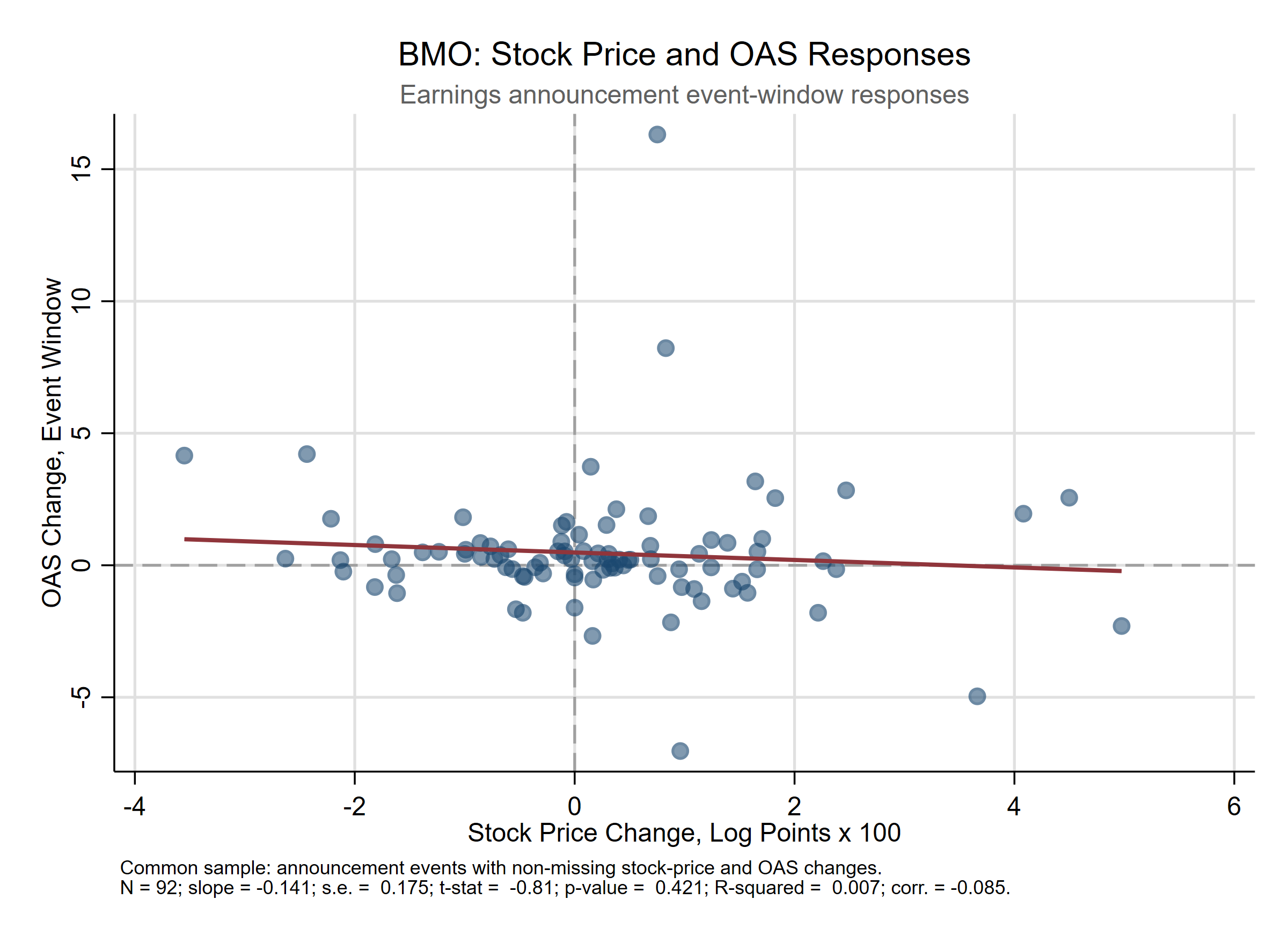}
\caption{BMO}
\end{subfigure}

\vspace{0.35cm}

\begin{subfigure}{0.48\textwidth}
\centering
\includegraphics[width=\textwidth]{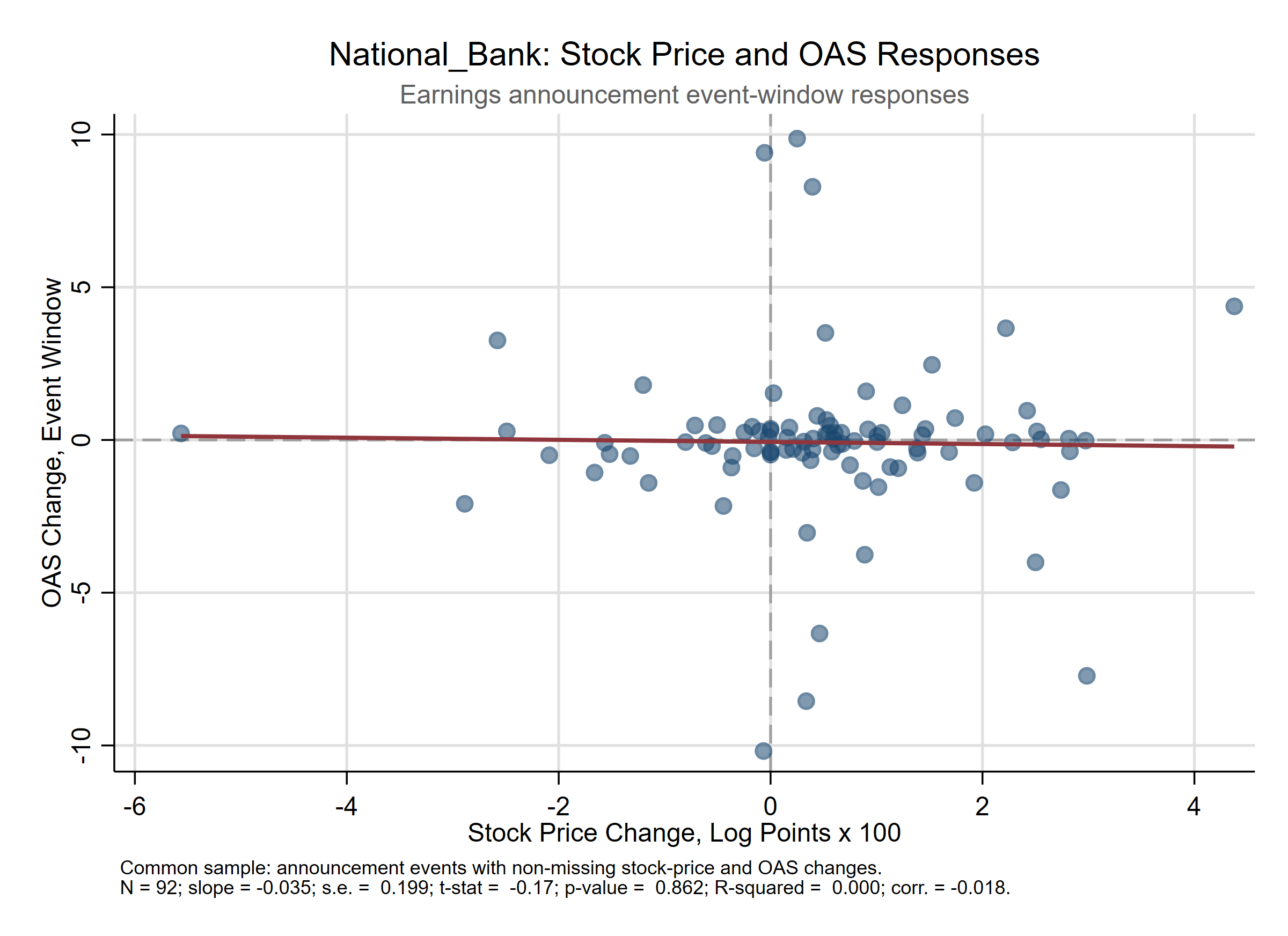}
\caption{National Bank}
\end{subfigure}
\hfill
\begin{subfigure}{0.48\textwidth}
\centering
\includegraphics[width=\textwidth]{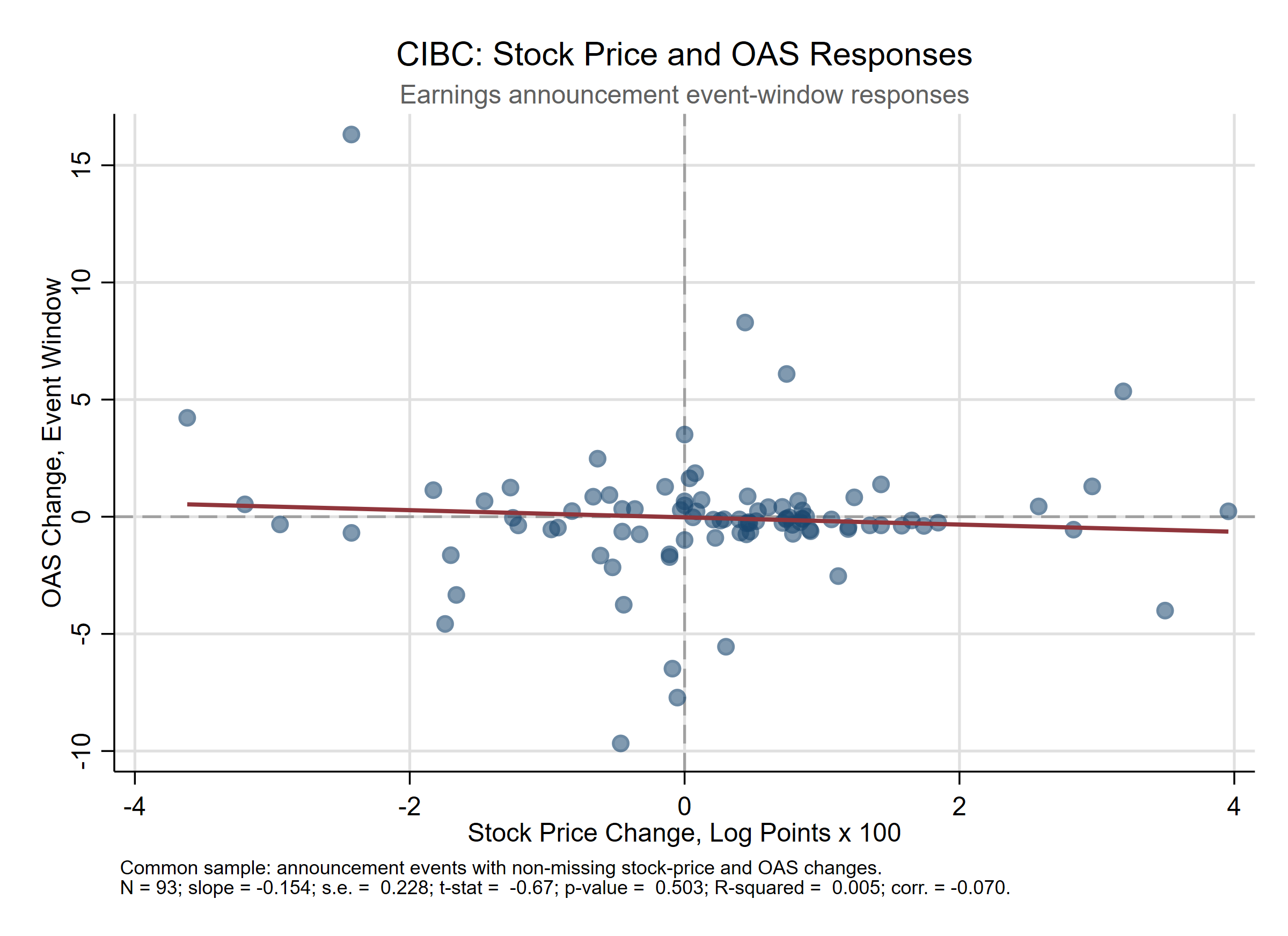}
\caption{CIBC}
\end{subfigure}

\caption{Bank-Level Stock-Price Reactions and Canadian Corporate OAS Changes}
\label{fig:scatter_price_oas_banks}
\floatfoot{\textbf{Notes:} This figure plots event-window changes in Canadian corporate OAS against event-window bank stock-price changes separately for each Canadian bank. Stock-price changes are event-window log price changes multiplied by 100. OAS changes are measured in basis points. The sample is restricted to announcement events with non-missing stock-price and OAS changes. Each panel reports the fitted OLS regression line and associated regression statistics.}
\end{figure}

%%%%%%%%%%%%%%%%%%%%%%%%%%%%%%%%%%%%%%%%%%%%%%%%%%%%%%%%%%%%%%%%%%%%%%
%%%%%%%%%%%%%%%%%%%%%%%%%%%%%%%%%%%%%%%%%%%%%%%%%%%%%%%%%%%%%%%%%%%%%%
\subsection{Admissible Rotations and Sign-Restriction Sensitivity}
\label{appendix:rotation_sensitivity}

This subsection provides additional diagnostics for the sign-restriction
decomposition used to construct the benchmark credit-supply bank net-worth
shock. The main text decomposes the raw market-capitalization-weighted bank
equity surprise, \(v_t\), into a credit-supply component, \(v_t^{CS}\), and a
residual component, \(v_t^C\). The decomposition is applied to the bivariate
event-window vector
\[
z_t
=
\begin{pmatrix}
v_t \\
\Delta OAS_t
\end{pmatrix},
\]
where \(\Delta OAS_t\) is the event-window change in the Canadian corporate
option-adjusted spread. We orthogonalize the innovations in \(z_t\) and
consider the set of two-dimensional rotations that satisfy
\[
\operatorname{cov}\left(v^{CS}_t,\Delta OAS_t\right)<0,
\qquad
\operatorname{cov}\left(v^{C}_t,\Delta OAS_t\right)>0.
\]
The first restriction assigns to the credit-supply component the part of bank
equity news that moves bank valuations and corporate spreads in opposite
directions. The second restriction assigns to the residual component the
variation that moves bank valuations and corporate spreads in the same
direction. Because these restrictions define a set of admissible rotations,
the benchmark shock uses the median admissible rotation.

Table \ref{tab:rotation_diagnostics} reports diagnostics for selected
rotations in the admissible set. The admissible interval is nonempty and
ranges from \(0\) to \(86.8\) degrees. Rotations close to the lower bound are
nearly identical to the raw bank equity surprise: for \(w=0.01\), the
credit-supply component has correlation \(1.00\) with \(v_t\) and correlation
\(-0.08\) with \(\Delta OAS_t\). Rotations close to the upper bound are nearly
identical to the negative OAS innovation: for \(w=0.99\), the credit-supply
component has correlation \(0.08\) with \(v_t\) and correlation \(-1.00\) with
\(\Delta OAS_t\). The median rotation balances the two pieces of
high-frequency information. At \(w=0.50\), the credit-supply component has
correlation \(0.73\) with the raw bank equity surprise, correlation \(-0.73\)
with the OAS change, and correlation \(0.70\) with the credit-supply shock
obtained from the poor man's sign restriction.

%%%%%%%%%%%%%%%%%%%%%%%%%%%%%%%%%%%%%%%%%%%%%%%%%%%%%%%%%%%%%%%%%%%%%%
\begin{table}[H]
\centering
\caption{Admissible Rotation Diagnostics}
\label{tab:rotation_diagnostics}
\begin{threeparttable}
\scriptsize
\setlength{\tabcolsep}{4pt}
\resizebox{\textwidth}{!}{
\begin{tabular}{lccccc}
\toprule
 & $w=0.01$ & $w=0.25$ & $w=0.50$ & $w=0.75$ & $w=0.99$ \\
\midrule
Rotation angle, degrees
& 0.87 & 21.71 & 43.41 & 65.12 & 85.96 \\

Effect on OAS, credit-supply component
& -0.67 & -4.32 & -9.55 & -21.09 & -135.49 \\

Effect on OAS, residual component
& 628.70 & 23.43 & 9.55 & 3.89 & 0.15 \\

Corr. with raw equity surprise
& 1.00 & 0.93 & 0.73 & 0.42 & 0.08 \\

Corr. with OAS change
& -0.08 & -0.43 & -0.73 & -0.93 & -1.00 \\

Corr. with poor man's supply shock
& 0.72 & 0.76 & 0.70 & 0.54 & 0.32 \\

\bottomrule
\end{tabular}
}
\floatfoot{\textit{Notes:} This table reports diagnostics for selected admissible rotations of the bivariate event-window system formed by the raw bank equity surprise and the Canadian corporate OAS change. The admissible interval ranges from \(0\) to \(86.8\) degrees. The benchmark specification uses the median admissible rotation, \(w=0.50\). The reported OAS effects correspond to the entries of the rotation-implied effects matrix \(C\). Correlations are computed at the event level.}
\end{threeparttable}
\end{table}
%%%%%%%%%%%%%%%%%%%%%%%%%%%%%%%%%%%%%%%%%%%%%%%%%%%%%%%%%%%%%%%%%%%%%%

The impulse-response robustness exercises use these diagnostics to choose two
sets of alternative rotations. The first set focuses on the economically
central part of the admissible interval, \(w\in\{0.25,0.50,0.75\}\). These
rotations all combine information from both bank equity surprises and OAS
changes. The second set uses the near-boundary rotations,
\(w\in\{0.01,0.50,0.99\}\), as stress tests. Since the near-boundary rotations
are close to the raw equity surprise and the negative OAS innovation,
respectively, we interpret them as useful diagnostics rather than as equally
plausible benchmark decompositions.

%%%%%%%%%%%%%%%%%%%%%%%%%%%%%%%%%%%%%%%%%%%%%%%%%%%%%%%%%%%%%%%%%%%%%%%%%%%%%%%%%%%%%%%%%%%%%%%%%%%%%%%%%%%%%%%%%%%%%%%%%%%%%%%%%%%
\subsection{Predictability from U.S. bank announcement shocks}
\label{appendix:predictability_us_bank_shocks}

This appendix reports an additional validation exercise that asks whether the
Canadian bank net-worth shocks are predictable from U.S. bank surprises
observed before the Canadian bank earnings announcements. The U.S. shocks are
constructed by \cite{ottonello2025financial} using bank earnings-announcement
information and therefore provide a natural external benchmark for our
Canadian shock series. The exercise addresses the concern that Canadian bank
equity reactions may partly reflect global banking-sector information already
revealed by U.S. bank announcements, rather than new information revealed at
Canadian bank earnings announcements.

For each Canadian bank-shock date \(t\), we attach the last three U.S. bank
surprises observed strictly before \(t\). Let \(t_1(t)\), \(t_2(t)\), and
\(t_3(t)\) denote the most recent, second-most recent, and third-most recent
U.S. bank-shock dates before the Canadian announcement date. We then estimate
regressions of the form
\[
s^{CAN}_t
=
\alpha
+
\sum_{\ell=1}^{3}
\beta_\ell s^{US}_{t_\ell(t)}
+
u_t,
\]
where \(s^{CAN}_t\) is either the raw total Canadian bank-equity surprise or
the median-rotation Canadian credit-supply bank net-worth shock. We estimate
two versions of the test for the preferred Canadian credit-supply shock: one
using the raw U.S. bank surprise and one using the purged U.S. bank surprise.
The sample is restricted to Canadian bank-shock dates before October 1, 2020.
Standard errors are robust.

Table \ref{tab:predictability_us_bank_shocks} reports the results. Column
(1) shows that the raw total Canadian bank-equity surprise displays some mild
predictability from the most recent raw U.S. bank surprise. The coefficient
on the most recent U.S. surprise is positive and statistically significant,
and the joint test for the last three U.S. surprises has a \(p\)-value of
\(0.082\). This is consistent with the view that the raw Canadian bank-equity
surprise may contain a common banking-sector or global financial component.

The preferred Canadian credit-supply shock, however, does not display such predictability. In column (2), the median-rotation supply shock is regressed on the last three raw U.S. bank surprises. The individual coefficients are statistically insignificant, the adjusted \(R^2\) is small, and the joint \(p\)-value is \(0.292\). In column (3), the median-rotation supply shock is regressed on the last three purged U.S. bank surprises. The joint \(p\)-value is \(0.818\). Thus, the preferred Canadian bank net-worth shock is not predictable from recent U.S. bank earnings-news shocks observed before the Canadian earnings announcement.

\begin{table}[H]
\centering
\caption{Predictability of Canadian Bank Net-Worth Shocks from U.S. Bank Announcement Shocks}
\label{tab:predictability_us_bank_shocks}
\begin{threeparttable}
\footnotesize
\def\sym#1{\ifmmode^{#1}\else\(^{#1}\)\fi}
\begin{tabular}{lccc}
\toprule
                    &\multicolumn{1}{c}{(1)}
                    &\multicolumn{1}{c}{(2)}
                    &\multicolumn{1}{c}{(3)}\\
                    &\multicolumn{1}{c}{Total shock}
                    &\multicolumn{1}{c}{Supply shock}
                    &\multicolumn{1}{c}{Supply shock}\\
                    &\multicolumn{1}{c}{Raw U.S. shocks}
                    &\multicolumn{1}{c}{Raw U.S. shocks}
                    &\multicolumn{1}{c}{Purged U.S. shocks}\\
\midrule
Last U.S. bank surprise
                    &       0.357\sym{**}
                    &       0.173
                    &                     \\
                    &     (0.140)
                    &     (0.176)
                    &                     \\
\addlinespace
Second-last U.S. bank surprise
                    &      -0.085
                    &      -0.087
                    &                     \\
                    &     (0.139)
                    &     (0.104)
                    &                     \\
\addlinespace
Third-last U.S. bank surprise
                    &      -0.059
                    &      -0.123
                    &                     \\
                    &     (0.073)
                    &     (0.082)
                    &                     \\
\addlinespace
Last purged U.S. bank surprise
                    &
                    &
                    &       0.004         \\
                    &
                    &
                    &     (0.313)         \\
\addlinespace
Second-last purged U.S. bank surprise
                    &
                    &
                    &       0.062         \\
                    &
                    &
                    &     (0.154)         \\
\addlinespace
Third-last purged U.S. bank surprise
                    &
                    &
                    &      -0.094         \\
                    &
                    &
                    &     (0.105)         \\
\midrule
Observations
                    &         276
                    &         276
                    &         276         \\
Adjusted \(R^2\)
                    &       0.058
                    &       0.025
                    &      -0.003         \\
Joint F-stat
                    &       2.255
                    &       1.250
                    &       0.311         \\
Joint \(p\)-value
                    &       0.082
                    &       0.292
                    &       0.818         \\
\bottomrule
\end{tabular}
\begin{tablenotes}
\footnotesize
\item \textbf{Notes:} This table reports predictability regressions of
Canadian bank net-worth shocks on U.S. bank announcement shocks observed
strictly before the Canadian bank-shock date. For each Canadian shock date,
we attach the last three U.S. bank surprises from \cite{ottonello2025financial}
observed before the Canadian date. Column (1) uses the raw total Canadian
bank-equity surprise as the dependent variable. Columns (2) and (3) use the
median-rotation Canadian credit-supply bank net-worth shock. Columns (1) and
(2) use raw U.S. bank surprises as regressors, while column (3) uses purged
U.S. bank surprises. The sample is restricted to Canadian bank-shock dates
before October 1, 2020. Robust standard errors are reported in parentheses.
The joint \(p\)-value corresponds to the test that the three U.S. bank-shock
coefficients in each column are jointly equal to zero. \(\sym{*}\),
\(\sym{**}\), and \(\sym{***}\) denote statistical significance at the 10, 5,
and 1 percent levels, respectively.
\end{tablenotes}
\end{threeparttable}
\end{table}
%%%%%%%%%%%%%%%%%%%%%%%%%%%%%%%%%%%%%%%%%%%%%%%%%%%%%%%%%%%%%%%%%%%%%%

%%%%%%%%%%%%%%%%%%%%%%%%%%%%%%%%%%%%%%%%%%%%%%%%%%%%%%%%%%%%%%%%%%%%%%
%%%%%%%%%%%%%%%%%%%%%%%%%%%%%%%%%%%%%%%%%%%%%%%%%%%%%%%%%%%%%%%%%%%%%%
%%%%%%%%%%%%%%%%%%%%%%%%%%%%%%%%%%%%%%%%%%%%%%%%%%%%%%%%%%%%%%%%%%%%%%
\newpage
\section{Figures on Robustness Checks \& Additional Results}
\label{appendix:figures_robustness_additional}

%%%%%%%%%%%%%%%%%%%%%%%%%%%%%%%%%%%%%%%%%%%%%%%%%%%%%%%%%%%%%%%%%%%%%%
\subsection{Rotation Sensitivity}
\label{appendix:robustness_rotation}

This subsection reports impulse-response figures associated with the
rotation-sensitivity exercise discussed in Section
\ref{subsec:robustness_rotation}. Figure \ref{fig:appendix_rotation_p25}
reports responses using the \(w=0.25\) admissible rotation, while Figure
\ref{fig:appendix_rotation_p75} reports responses using the \(w=0.75\)
admissible rotation. The benchmark specification in the main text uses the
median admissible rotation, \(w=0.50\). Unless otherwise noted, both exercises
use the same local projection specification as the benchmark: six monthly lags
of Canadian macro-financial controls, a COVID-period dummy, and Newey--West
standard errors.

%%%%%%%%%%%%%%%%%%%%%%%%%%%%%%%%%%%%%%%%%%%%%%%%%%%%%%%%%%%%%%%%%%%%%%
\begin{figure}[p]
\centering
\includegraphics[width=0.95\textwidth]{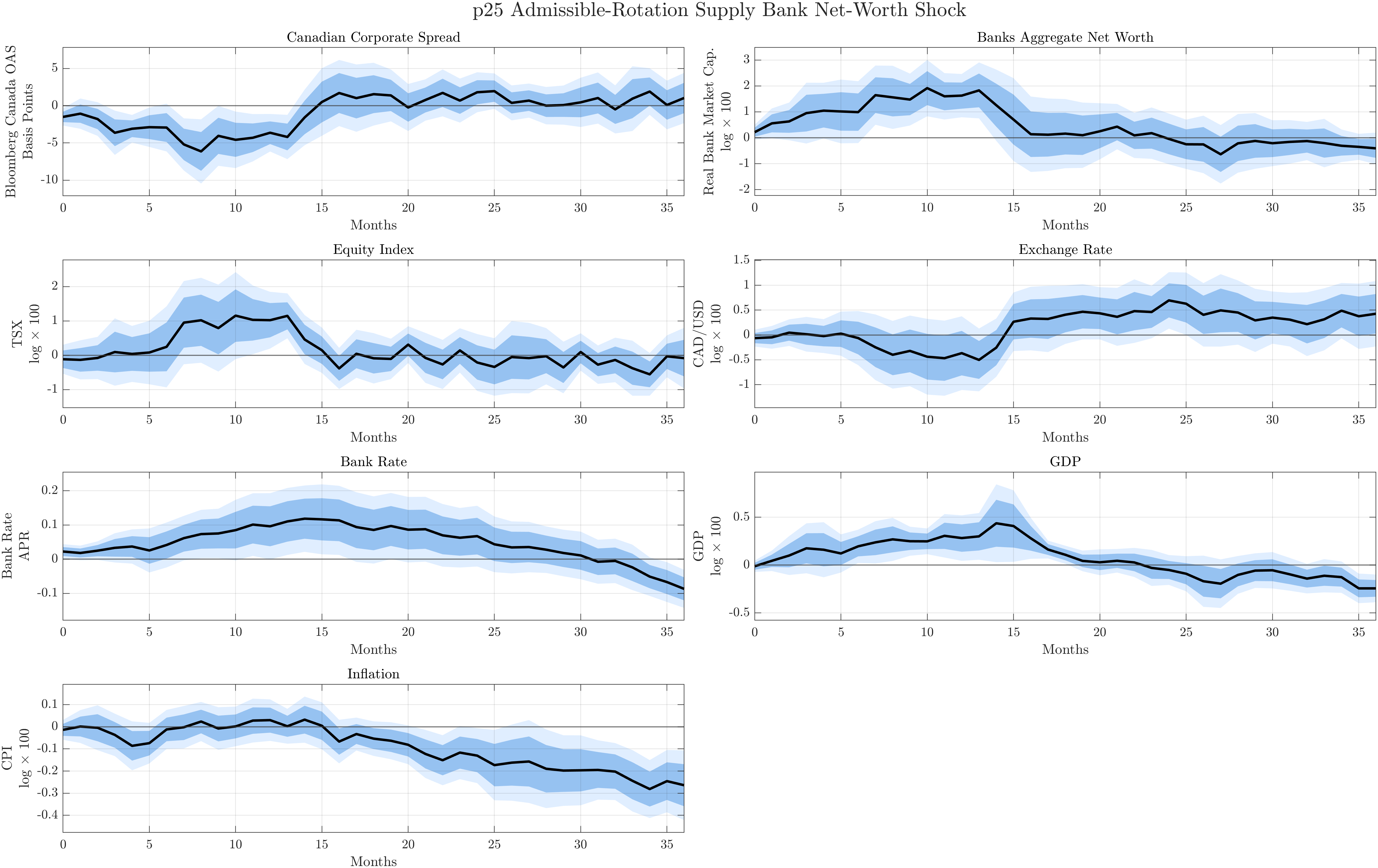}
\caption{Macroeconomic Effects of Bank Net-Worth Supply Shocks: \(w=0.25\) Admissible Rotation}
\label{fig:appendix_rotation_p25}
\floatfoot{\textbf{Notes:} This figure reports local projection impulse responses to a one-standard-deviation favorable credit-supply bank net-worth shock identified using the \(w=0.25\) admissible rotation. The shock is constructed from timing-adjusted earnings-announcement stock-price reactions and event-window changes in the Canadian corporate OAS, and then aggregated to the monthly frequency. The specification includes six monthly lags of Canadian macro-financial controls and a COVID-period dummy. For variables in logs, responses are cumulative changes relative to the month before the shock and are measured in percentage points. The Bank of Canada policy rate and the Canadian corporate OAS are reported in forward levels. Dark and light shaded areas denote 68 and 90 percent confidence intervals, respectively, computed using Newey--West standard errors.}
\end{figure}
%%%%%%%%%%%%%%%%%%%%%%%%%%%%%%%%%%%%%%%%%%%%%%%%%%%%%%%%%%%%%%%%%%%%%%

%%%%%%%%%%%%%%%%%%%%%%%%%%%%%%%%%%%%%%%%%%%%%%%%%%%%%%%%%%%%%%%%%%%%%%
\begin{figure}[p]
\centering
\includegraphics[width=0.95\textwidth]{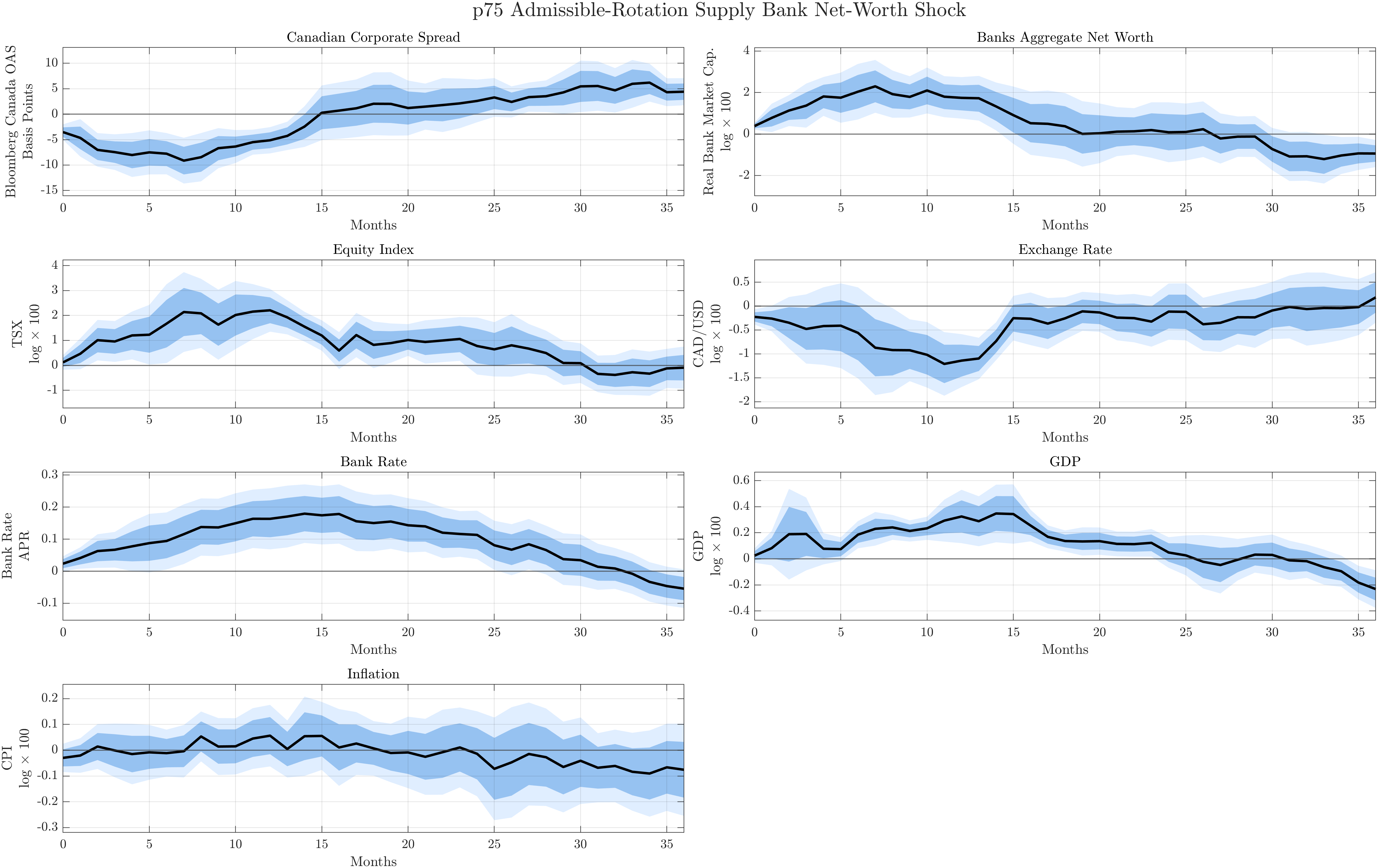}
\caption{Macroeconomic Effects of Bank Net-Worth Supply Shocks: \(w=0.75\) Admissible Rotation}
\label{fig:appendix_rotation_p75}
\floatfoot{\textbf{Notes:} This figure reports local projection impulse responses to a one-standard-deviation favorable credit-supply bank net-worth shock identified using the \(w=0.75\) admissible rotation. The shock is constructed from timing-adjusted earnings-announcement stock-price reactions and event-window changes in the Canadian corporate OAS, and then aggregated to the monthly frequency. The specification includes six monthly lags of Canadian macro-financial controls and a COVID-period dummy. For variables in logs, responses are cumulative changes relative to the month before the shock and are measured in percentage points. The Bank of Canada policy rate and the Canadian corporate OAS are reported in forward levels. Dark and light shaded areas denote 68 and 90 percent confidence intervals, respectively, computed using Newey--West standard errors.}
\end{figure}
%%%%%%%%%%%%%%%%%%%%%%%%%%%%%%%%%%%%%%%%%%%%%%%%%%%%%%%%%%%%%%%%%%%%%%

%%%%%%%%%%%%%%%%%%%%%%%%%%%%%%%%%%%%%%%%%%%%%%%%%%%%%%%%%%%%%%%%%%%%%%
\subsection{Specification Robustness}
\label{appendix:robustness_specification}

This subsection reports the impulse-response figures associated with the specification robustness checks discussed in Section \ref{subsec:robustness_specification}. Figure \ref{fig:appendix_robust_lags2} reports responses from a specification with two monthly lags of Canadian macro-financial controls instead of six. Figure \ref{fig:appendix_robust_trend} reports responses from a specification that augments the benchmark local projections with a linear deterministic trend. Unless otherwise noted, both exercises use the normalized median rotational-angle credit-supply bank net-worth shock and report 68 and 90 percent Newey--West confidence intervals.

%%%%%%%%%%%%%%%%%%%%%%%%%%%%%%%%%%%%%%%%%%%%%%%%%%%%%%%%%%%%%%%%%%%%%%
\begin{figure}[p]
\centering
\includegraphics[width=0.95\textwidth]{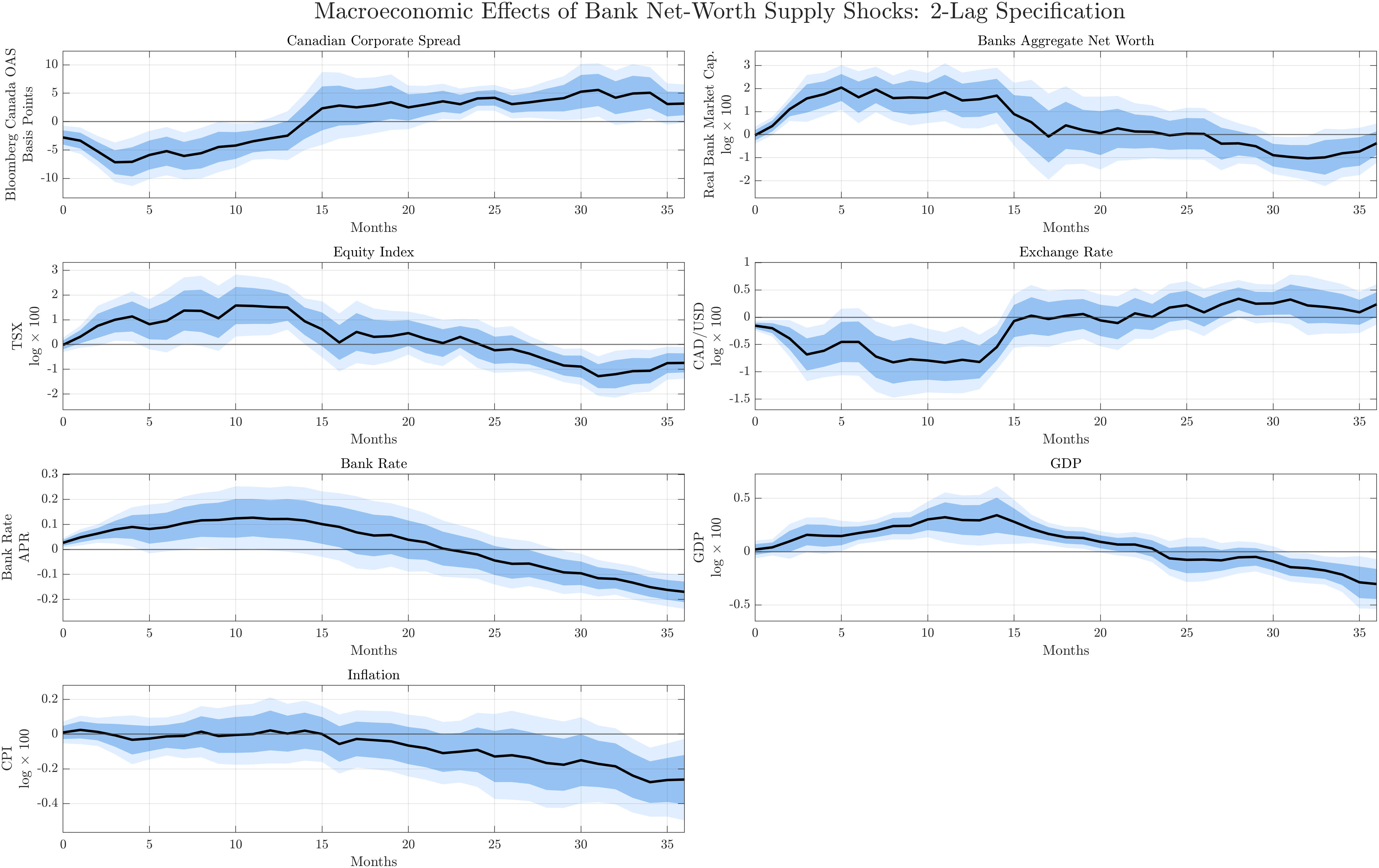}
\caption{Macroeconomic Effects of Bank Net-Worth Supply Shocks: Two-Lag Specification}
\label{fig:appendix_robust_lags2}
\floatfoot{\textbf{Notes:} This figure reports local projection impulse responses to a one-standard-deviation favorable credit-supply bank net-worth shock identified using the median rotational-angle procedure. The specification includes two monthly lags of Canadian macro-financial controls and a COVID-period dummy. The shock is constructed from timing-adjusted earnings-announcement stock-price reactions and aggregated using lagged market-capitalization shares. For variables in logs, responses are cumulative changes relative to the month before the shock and are measured in percentage points. The Bank of Canada policy rate and the Canadian corporate OAS are reported in forward levels. Dark and light shaded areas denote 68 and 90 percent confidence intervals, respectively, computed using Newey--West standard errors.}
\end{figure}
%%%%%%%%%%%%%%%%%%%%%%%%%%%%%%%%%%%%%%%%%%%%%%%%%%%%%%%%%%%%%%%%%%%%%%

%%%%%%%%%%%%%%%%%%%%%%%%%%%%%%%%%%%%%%%%%%%%%%%%%%%%%%%%%%%%%%%%%%%%%%
\begin{figure}[p]
\centering
\includegraphics[width=0.95\textwidth]{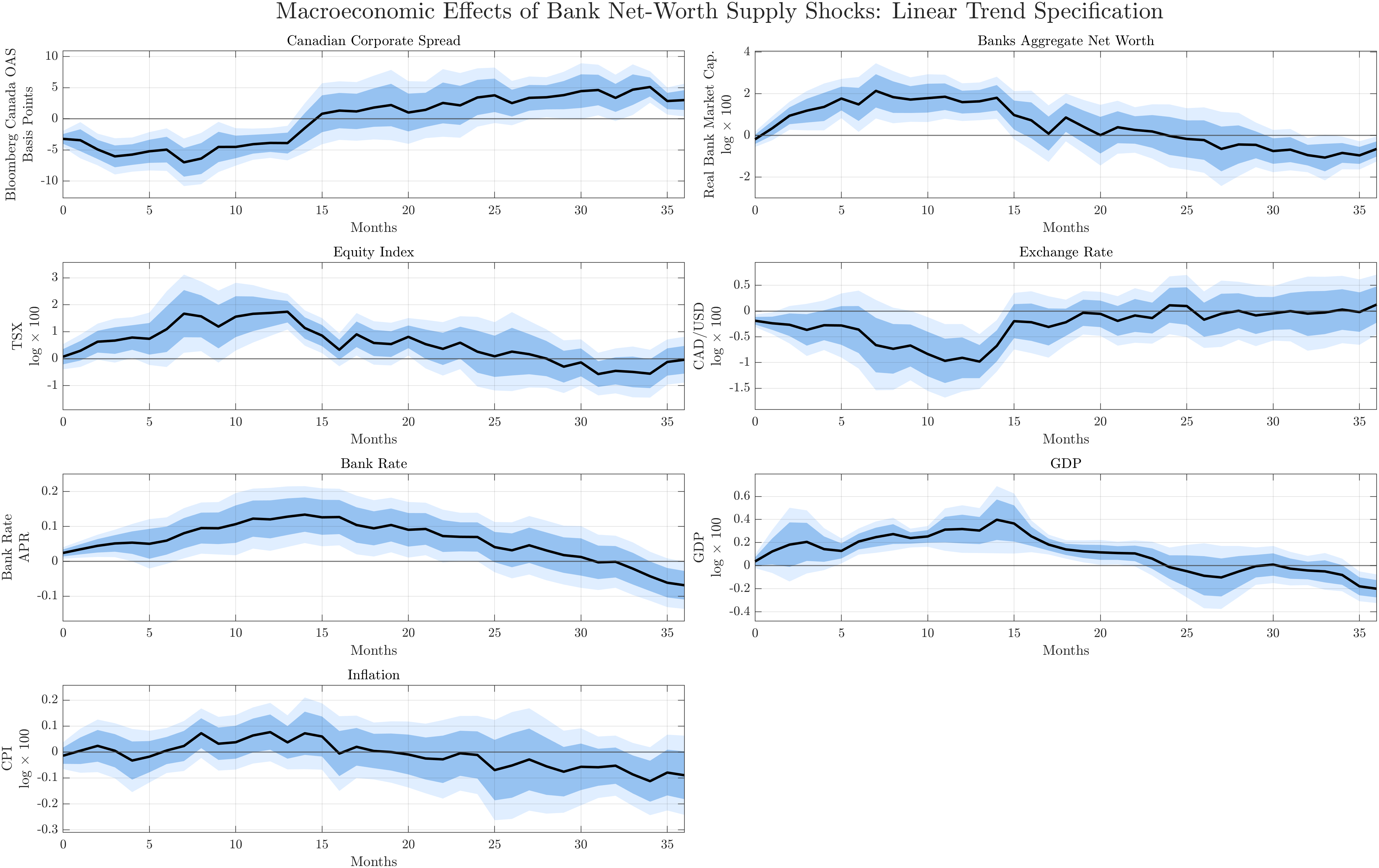}
\caption{Macroeconomic Effects of Bank Net-Worth Supply Shocks: Linear Trend Specification}
\label{fig:appendix_robust_trend}
\floatfoot{\textbf{Notes:} This figure reports local projection impulse responses to a one-standard-deviation favorable credit-supply bank net-worth shock identified using the median rotational-angle procedure. The specification includes six monthly lags of Canadian macro-financial controls, a COVID-period dummy, and a linear deterministic time trend. The shock is constructed from timing-adjusted earnings-announcement stock-price reactions and aggregated using lagged market-capitalization shares. For variables in logs, responses are cumulative changes relative to the month before the shock and are measured in percentage points. The Bank of Canada policy rate and the Canadian corporate OAS are reported in forward levels. Dark and light shaded areas denote 68 and 90 percent confidence intervals, respectively, computed using Newey--West standard errors.}
\end{figure}
%%%%%%%%%%%%%%%%%%%%%%%%%%%%%%%%%%%%%%%%%%%%%%%%%%%%%%%%%%%%%%%%%%%%%%

%%%%%%%%%%%%%%%%%%%%%%%%%%%%%%%%%%%%%%%%%%%%%%%%%%%%%%%%%%%%%%%%%%%%%%
%%%%%%%%%%%%%%%%%%%%%%%%%%%%%%%%%%%%%%%%%%%%%%%%%%%%%%%%%%%%%%%%%%%%%%
\subsection{Sample Robustness}
\label{appendix:robustness_sample}

This subsection reports the impulse-response figures associated with the sample robustness exercises discussed in Section \ref{subsec:robustness_sample}. Figure \ref{fig:appendix_robust_pre2020} reports responses estimated on the pre-2020 sample, excluding the COVID period and the most recent post-pandemic observations. Figure \ref{fig:appendix_robust_post2010} reports responses estimated on the post-2010 sample. Unless otherwise noted, both exercises use the normalized median rotational-angle credit-supply bank net-worth shock and report 68 and 90 percent Newey--West confidence intervals.

%%%%%%%%%%%%%%%%%%%%%%%%%%%%%%%%%%%%%%%%%%%%%%%%%%%%%%%%%%%%%%%%%%%%%%
\begin{figure}[p]
\centering
\includegraphics[width=0.95\textwidth]{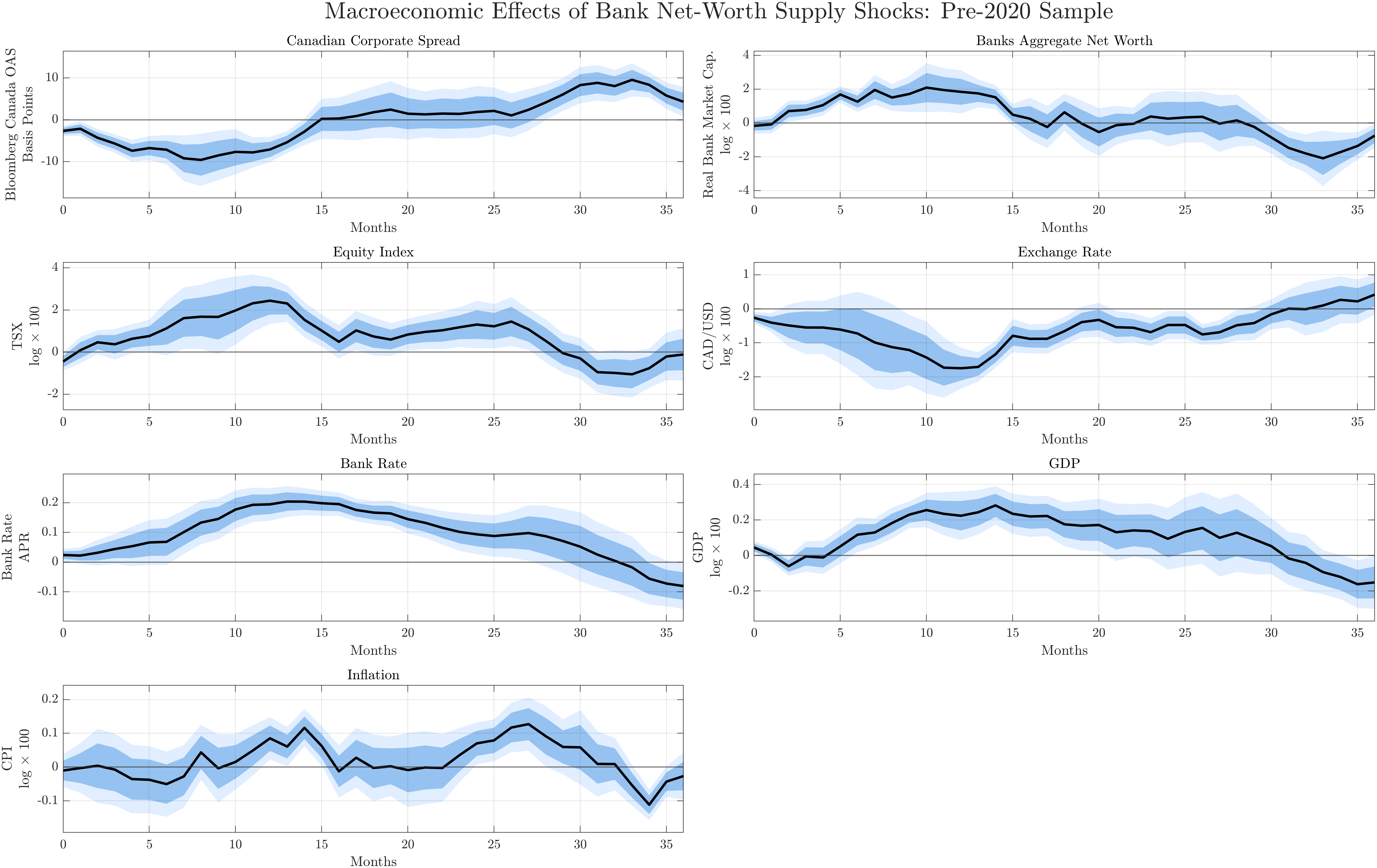}
\caption{Macroeconomic Effects of Bank Net-Worth Supply Shocks: Pre-2020 Sample}
\label{fig:appendix_robust_pre2020}
\floatfoot{\textbf{Notes:} This figure reports local projection impulse responses to a one-standard-deviation favorable credit-supply bank net-worth shock identified using the median rotational-angle procedure. The sample is restricted to observations before 2020. The specification includes six monthly lags of Canadian macro-financial controls. The shock is constructed from timing-adjusted earnings-announcement stock-price reactions and aggregated using lagged market-capitalization shares. For variables in logs, responses are cumulative changes relative to the month before the shock and are measured in percentage points. The Bank of Canada policy rate and the Canadian corporate OAS are reported in forward levels. Dark and light shaded areas denote 68 and 90 percent confidence intervals, respectively, computed using Newey--West standard errors.}
\end{figure}
%%%%%%%%%%%%%%%%%%%%%%%%%%%%%%%%%%%%%%%%%%%%%%%%%%%%%%%%%%%%%%%%%%%%%%

%%%%%%%%%%%%%%%%%%%%%%%%%%%%%%%%%%%%%%%%%%%%%%%%%%%%%%%%%%%%%%%%%%%%%%
\begin{figure}[p]
\centering
\includegraphics[width=0.95\textwidth]{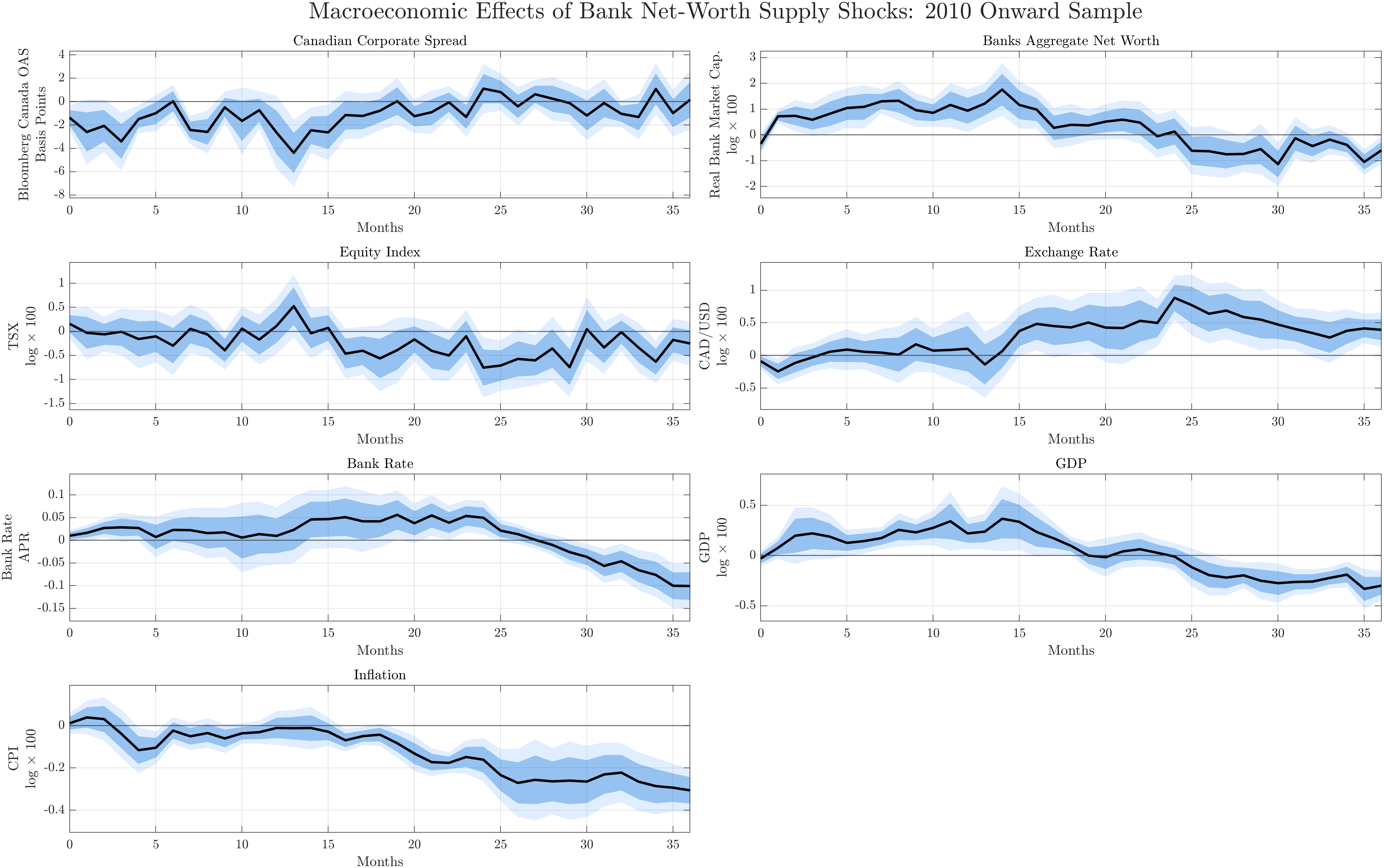}
\caption{Macroeconomic Effects of Bank Net-Worth Supply Shocks: Post-2010 Sample}
\label{fig:appendix_robust_post2010}
\floatfoot{\textbf{Notes:} This figure reports local projection impulse responses to a one-standard-deviation favorable credit-supply bank net-worth shock identified using the median rotational-angle procedure. The sample is restricted to observations from 2010 onward. The specification includes six monthly lags of Canadian macro-financial controls and a COVID-period dummy. The shock is constructed from timing-adjusted earnings-announcement stock-price reactions and aggregated using lagged market-capitalization shares. For variables in logs, responses are cumulative changes relative to the month before the shock and are measured in percentage points. The Bank of Canada policy rate and the Canadian corporate OAS are reported in forward levels. Dark and light shaded areas denote 68 and 90 percent confidence intervals, respectively, computed using Newey--West standard errors.}
\end{figure}
%%%%%%%%%%%%%%%%%%%%%%%%%%%%%%%%%%%%%%%%%%%%%%%%%%%%%%%%%%%%%%%%%%%%%%

%%%%%%%%%%%%%%%%%%%%%%%%%%%%%%%%%%%%%%%%%%%%%%%%%%%%%%%%%%%%%%%%%%%%%%
\subsection{Excluding the Largest Banks}
\label{appendix:robustness_large_banks}

This appendix reports a robustness exercise that assesses whether the
benchmark impulse responses are driven by the largest Canadian banks in the
sample. The benchmark shock is market-capitalization weighted, so larger
institutions receive larger weights in the aggregate bank net-worth surprise.
While this weighting is appropriate for constructing an aggregate measure of
bank balance-sheet news, it raises the possibility that the results may be
dominated by the largest banks.

We address this concern by reconstructing the shock series after excluding
large banks from the underlying announcement sample. We consider three
exercises. First, we remove RBC from the bank-level event sample, reconstruct
the credit-supply bank net-worth shock using the remaining banks, and
re-estimate the benchmark local projections. Second, we repeat the exercise
after removing TD. Third, we remove both RBC and TD simultaneously. In each
case, the shock is reconstructed before the macroeconomic analysis, so the
exercise changes the underlying high-frequency shock rather than simply
dropping observations from the local projection sample. The local projections
use the same specification as the benchmark: the same macro-financial
outcomes, normalization, controls, lag structure, COVID-period dummy, horizon,
and Newey--West inference procedure.

Figures \ref{fig:appendix_no_rbc}, \ref{fig:appendix_no_td}, and
\ref{fig:appendix_no_rbc_td} report the results. The main propagation
patterns are robust across the three exclusions. Excluding RBC leaves the
benchmark responses largely unchanged: corporate spreads fall after a
favorable bank net-worth shock, bank market capitalization and broader equity
prices rise, and GDP increases over the medium run. The response of the bank
rate is positive for much of the horizon, consistent with the baseline
monetary-policy response.

The results are similar when TD is excluded. The spread response remains
negative during the first year, bank valuations rise, the equity index
increases, and GDP responds positively over the medium run. Excluding both
RBC and TD also preserves the main qualitative findings. Even when the two
largest banks are removed from the shock-construction exercise, the remaining
bank announcements generate a credit-supply shock that lowers corporate
spreads, raises bank valuations and equity prices, and is followed by stronger
real activity.

Overall, these exercises show that the benchmark results are not driven by a
single large institution or by the joint influence of RBC and TD. The
estimated propagation mechanism remains present when the shock is constructed
from the remaining banks, supporting the interpretation that the benchmark
shock captures broad Canadian bank net-worth news.

%%%%%%%%%%%%%%%%%%%%%%%%%%%%%%%%%%%%%%%%%%%%%%%%%%%%%%%%%%%%%%%%%%%%%%
\begin{figure}[H]
\centering
\includegraphics[width=0.95\textwidth]{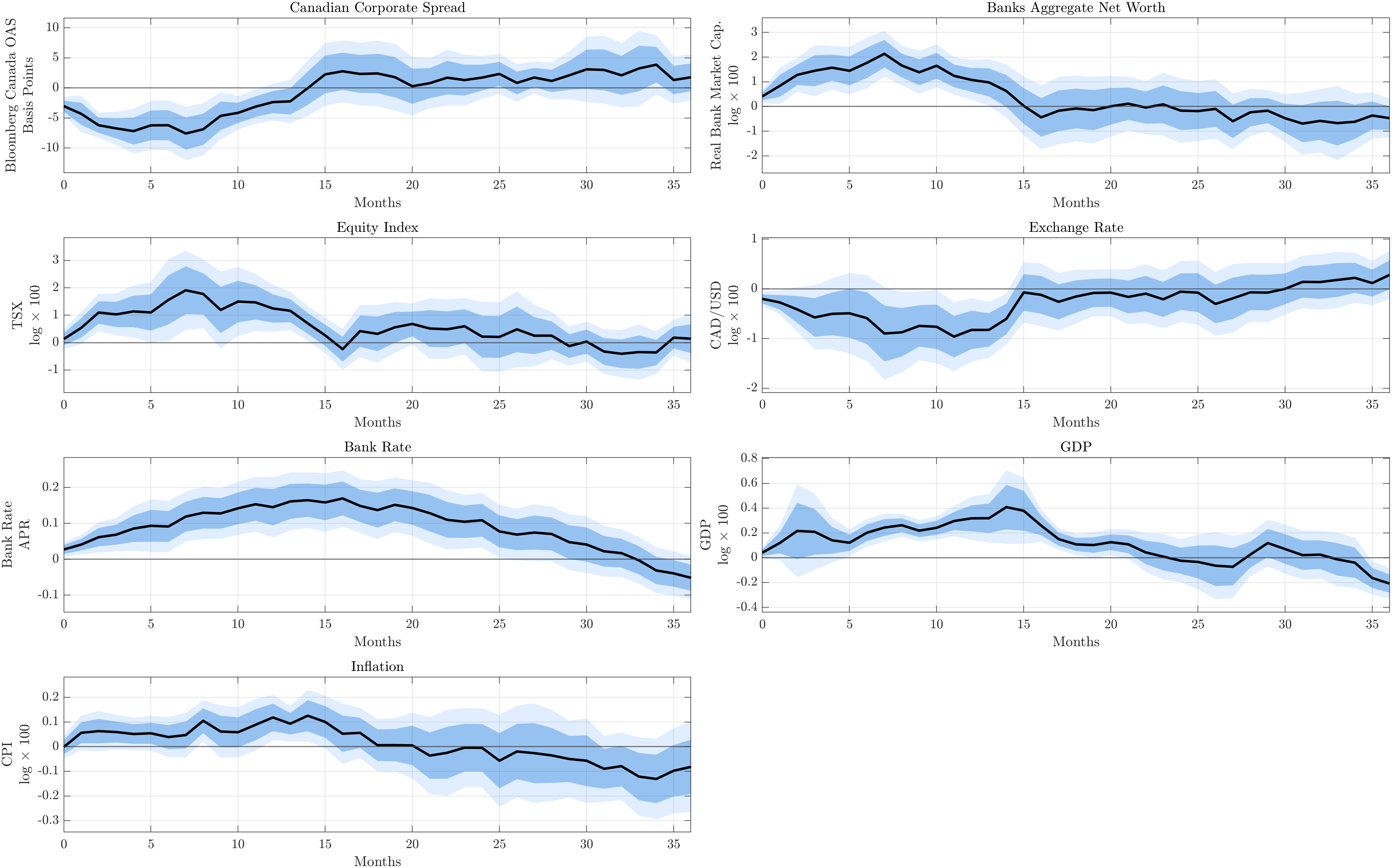}
\caption{Benchmark Responses Excluding RBC}
\label{fig:appendix_no_rbc}
\floatfoot{\textbf{Notes:} This figure reports impulse responses to the
median-rotation credit-supply bank net-worth shock after removing RBC from
the shock-construction sample. The shock is reconstructed using the remaining
banks and the benchmark local-projection specification is re-estimated. The
specification uses the same outcomes, controls, normalization, horizon,
COVID-period dummy, and Newey--West inference procedure as the benchmark
specification. The solid black line reports the point estimate. The dark and
light blue shaded areas report the 68 and 90 percent confidence intervals,
respectively.}
\end{figure}
%%%%%%%%%%%%%%%%%%%%%%%%%%%%%%%%%%%%%%%%%%%%%%%%%%%%%%%%%%%%%%%%%%%%%%

%%%%%%%%%%%%%%%%%%%%%%%%%%%%%%%%%%%%%%%%%%%%%%%%%%%%%%%%%%%%%%%%%%%%%%
\begin{figure}[H]
\centering
\includegraphics[width=0.95\textwidth]{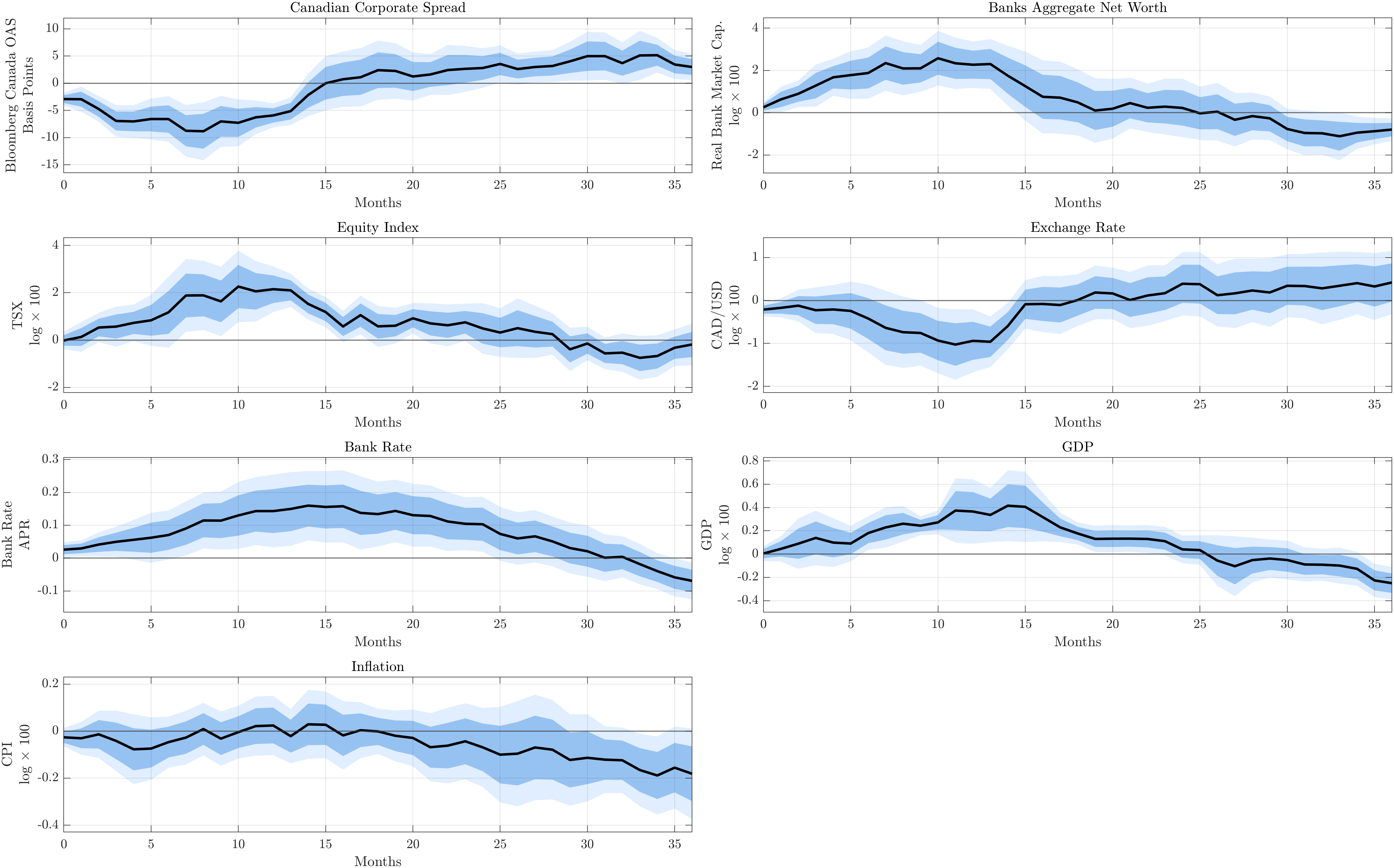}
\caption{Benchmark Responses Excluding TD}
\label{fig:appendix_no_td}
\floatfoot{\textbf{Notes:} This figure reports impulse responses to the
median-rotation credit-supply bank net-worth shock after removing TD from the
shock-construction sample. The shock is reconstructed using the remaining
banks and the benchmark local-projection specification is re-estimated. The
specification uses the same outcomes, controls, normalization, horizon,
COVID-period dummy, and Newey--West inference procedure as the benchmark
specification. The solid black line reports the point estimate. The dark and
light blue shaded areas report the 68 and 90 percent confidence intervals,
respectively.}
\end{figure}
%%%%%%%%%%%%%%%%%%%%%%%%%%%%%%%%%%%%%%%%%%%%%%%%%%%%%%%%%%%%%%%%%%%%%%

%%%%%%%%%%%%%%%%%%%%%%%%%%%%%%%%%%%%%%%%%%%%%%%%%%%%%%%%%%%%%%%%%%%%%%
\begin{figure}[H]
\centering
\includegraphics[width=0.95\textwidth]{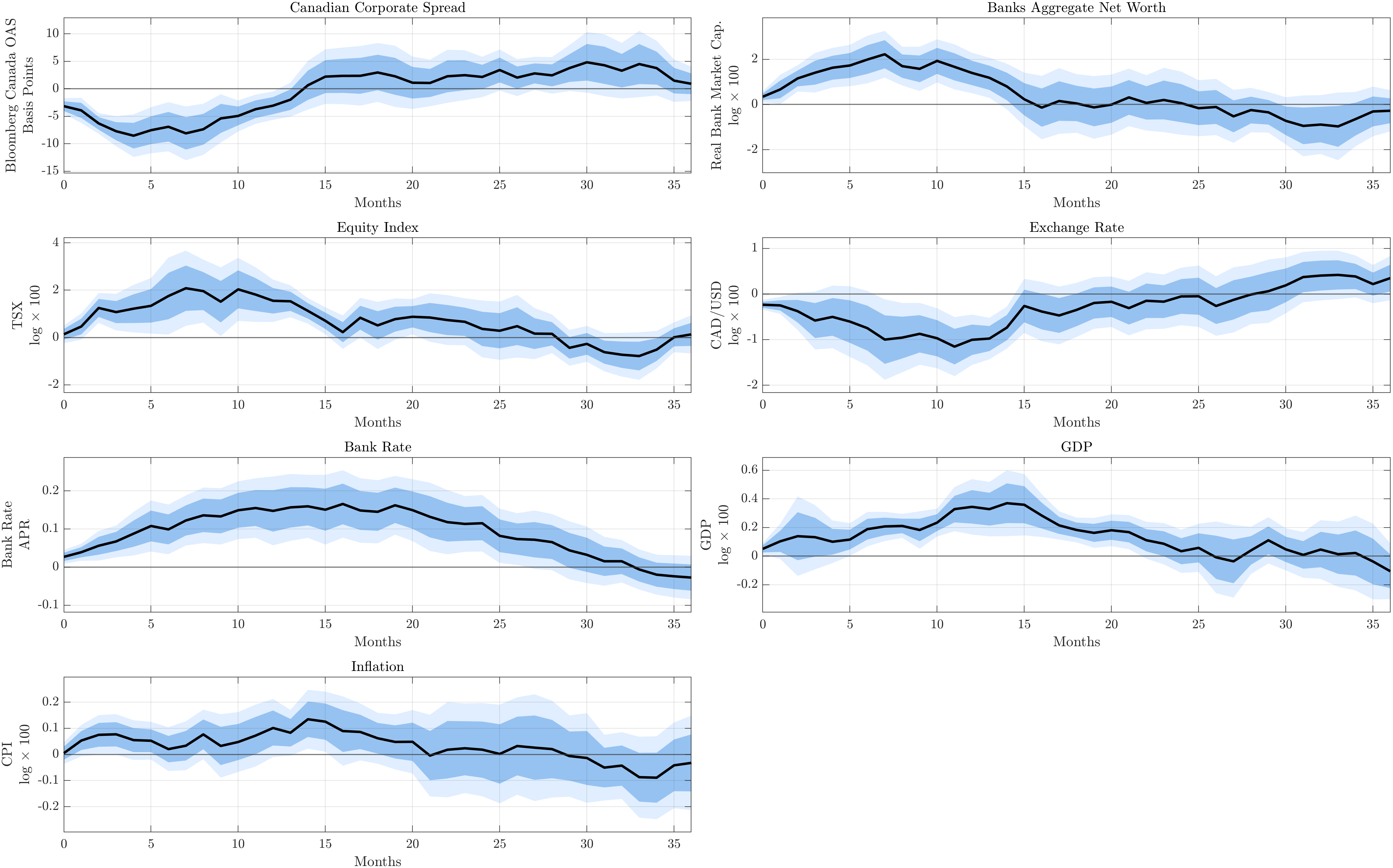}
\caption{Benchmark Responses Excluding RBC and TD}
\label{fig:appendix_no_rbc_td}
\floatfoot{\textbf{Notes:} This figure reports impulse responses to the
median-rotation credit-supply bank net-worth shock after removing both RBC
and TD from the shock-construction sample. The shock is reconstructed using
the remaining banks and the benchmark local-projection specification is
re-estimated. The specification uses the same outcomes, controls,
normalization, horizon, COVID-period dummy, and Newey--West inference
procedure as the benchmark specification. The solid black line reports the
point estimate. The dark and light blue shaded areas report the 68 and 90
percent confidence intervals, respectively.}
\end{figure}
%%%%%%%%%%%%%%%%%%%%%%%%%%%%%%%%%%%%%%%%%%%%%%%%%%%%%%%%%%%%%%%%%%%%%%

%%%%%%%%%%%%%%%%%%%%%%%%%%%%%%%%%%%%%%%%%%%%%%%%%%%%%%%%%%%%%%%%%%%%%%
\subsection{Sectoral Propagation}
\label{appendix:additional_sectoral_propagation}

This subsection reports the sectoral impulse-response figures associated with the additional propagation exercises discussed in Section \ref{subsec:robustness_sectoral}. Figure \ref{fig:appendix_goods_services} reports responses for durable goods, non-durable goods, and services. Figure \ref{fig:appendix_sectoral_gdp} reports responses for a more disaggregated set of sectors. Unless otherwise noted, all figures report responses to a one-standard-deviation favorable credit-supply bank net-worth shock identified using the median rotational-angle procedure. The regressions include the benchmark Canadian macro-financial controls, a COVID-period dummy, and six lags of the corresponding sectoral outcome.

%%%%%%%%%%%%%%%%%%%%%%%%%%%%%%%%%%%%%%%%%%%%%%%%%%%%%%%%%%%%%%%%%%%%%%
\begin{figure}[p]
\centering
\includegraphics[width=0.95\textwidth]{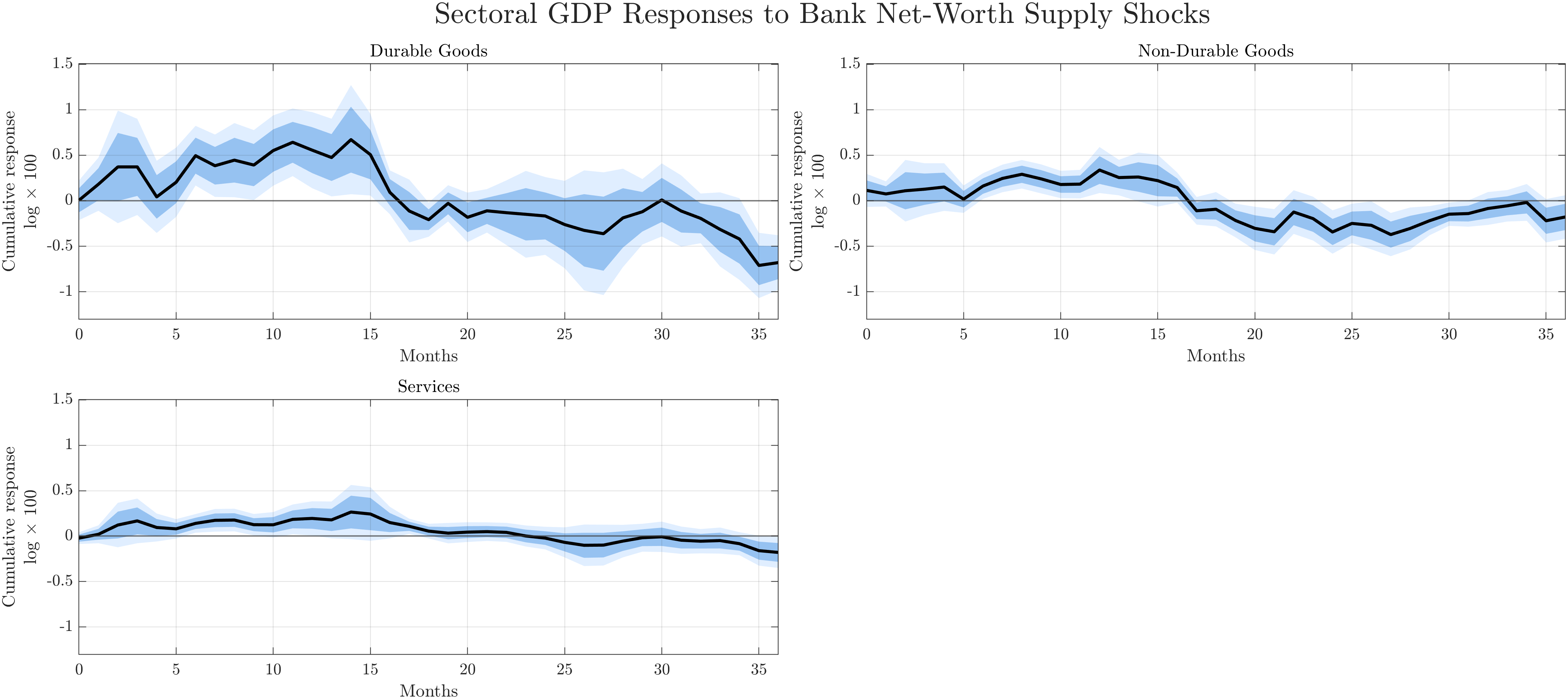}
\caption{Sectoral GDP Responses to Bank Net-Worth Supply Shocks: Goods and Services}
\label{fig:appendix_goods_services}
\floatfoot{\textbf{Notes:} This figure reports local projection impulse responses of durable goods, non-durable goods, and services to a one-standard-deviation favorable credit-supply bank net-worth shock identified using the median rotational-angle procedure. Sectoral outcomes are transformed as $100$ times log levels, and responses are cumulative changes relative to the month before the shock. The specification includes the benchmark Canadian macro-financial controls, six lags of the corresponding sectoral outcome, and a COVID-period dummy. Dark and light shaded areas denote 68 and 90 percent confidence intervals, respectively, computed using Newey--West standard errors.}
\end{figure}
%%%%%%%%%%%%%%%%%%%%%%%%%%%%%%%%%%%%%%%%%%%%%%%%%%%%%%%%%%%%%%%%%%%%%%

%%%%%%%%%%%%%%%%%%%%%%%%%%%%%%%%%%%%%%%%%%%%%%%%%%%%%%%%%%%%%%%%%%%%%%
\begin{figure}[p]
\centering
\includegraphics[width=0.95\textwidth]{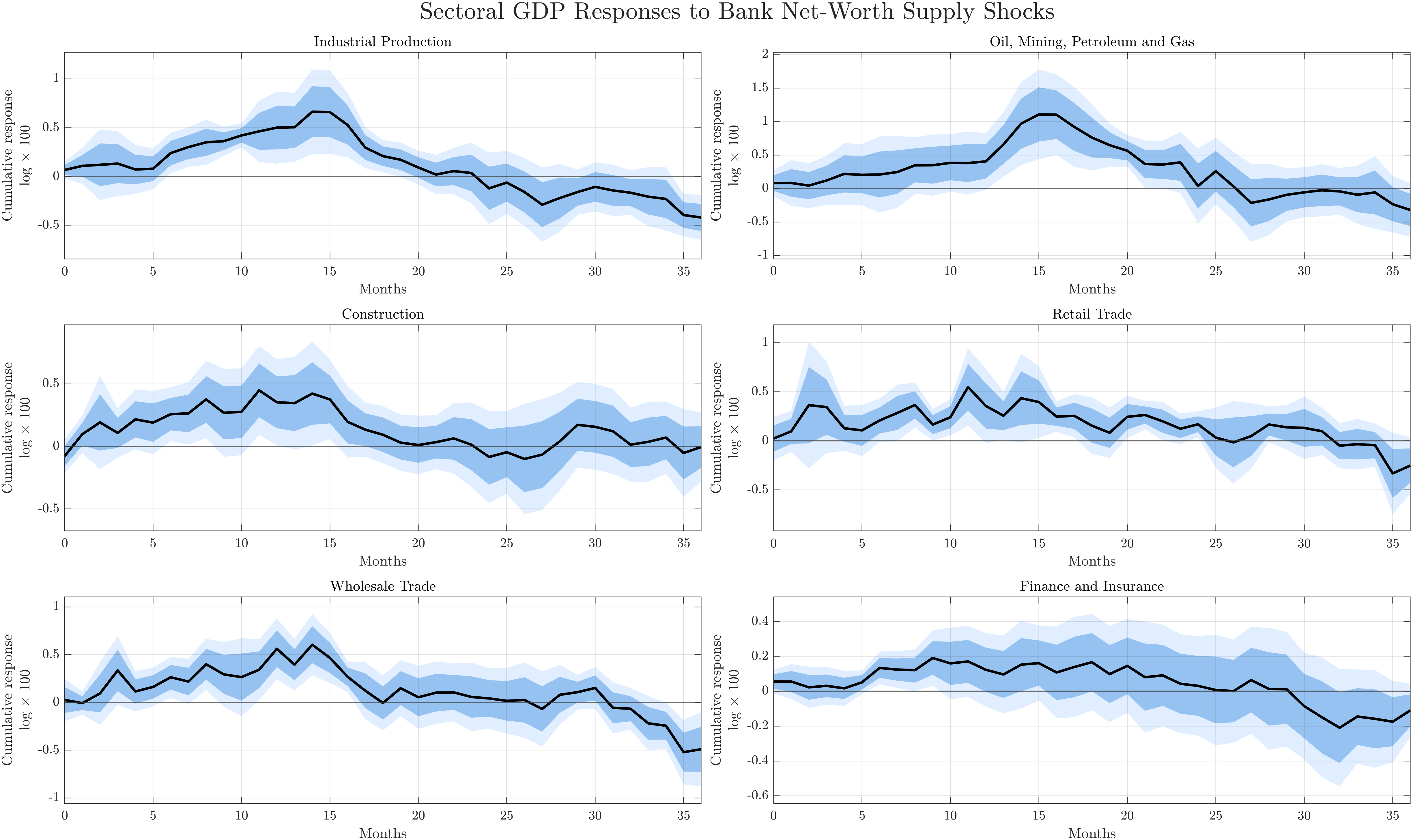}
\caption{Sectoral GDP Responses to Bank Net-Worth Supply Shocks: Disaggregated Sectors}
\label{fig:appendix_sectoral_gdp}
\floatfoot{\textbf{Notes:} This figure reports local projection impulse responses of industrial production; oil, mining, petroleum and gas; construction; retail trade; wholesale trade; and finance and insurance to a one-standard-deviation favorable credit-supply bank net-worth shock identified using the median rotational-angle procedure. Sectoral outcomes are transformed as $100$ times log levels, and responses are cumulative changes relative to the month before the shock. The specification includes the benchmark Canadian macro-financial controls, six lags of the corresponding sectoral outcome, and a COVID-period dummy. Dark and light shaded areas denote 68 and 90 percent confidence intervals, respectively, computed using Newey--West standard errors.}
\end{figure}
%%%%%%%%%%%%%%%%%%%%%%%%%%%%%%%%%%%%%%%%%%%%%%%%%%%%%%%%%%%%%%%%%%%%%%

%%%%%%%%%%%%%%%%%%%%%%%%%%%%%%%%%%%%%%%%%%%%%%%%%%%%%%%%%%%%%%%%%%%%%%
\subsection{Labor-Market Responses}
\label{appendix:additional_labor_market}

This subsection reports labor-market impulse responses associated with the additional propagation exercise discussed in Section \ref{subsec:robustness_labor_market}. Figure \ref{fig:appendix_labor_market} reports responses of total employment, the unemployment rate, and total hours worked to a one-standard-deviation favorable credit-supply bank net-worth shock identified using the median rotational-angle procedure. The regressions include the benchmark Canadian macro-financial controls, a COVID-period dummy, and six lags of the corresponding labor-market outcome.

%%%%%%%%%%%%%%%%%%%%%%%%%%%%%%%%%%%%%%%%%%%%%%%%%%%%%%%%%%%%%%%%%%%%%%
\begin{figure}[p]
\centering
\includegraphics[width=0.95\textwidth]{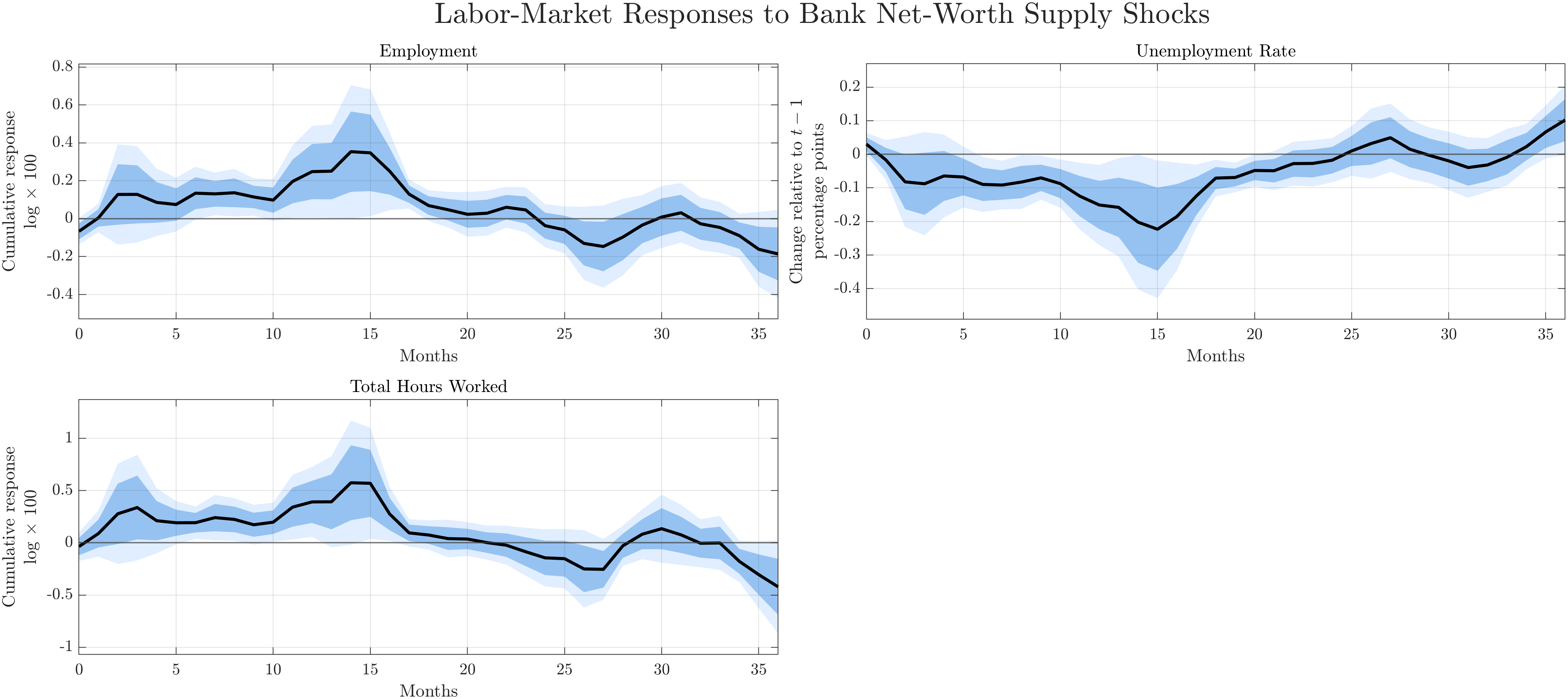}
\caption{Labor-Market Responses to Bank Net-Worth Supply Shocks}
\label{fig:appendix_labor_market}
\floatfoot{\textbf{Notes:} This figure reports local projection impulse responses of labor-market variables to a one-standard-deviation favorable credit-supply bank net-worth shock identified using the median rotational-angle procedure. Employment and total hours worked are transformed as $100$ times log levels, and responses are cumulative changes relative to the month before the shock. The unemployment-rate response is measured as the change in percentage points relative to the month before the shock. The specification includes the benchmark Canadian macro-financial controls, six lags of the corresponding labor-market outcome, and a COVID-period dummy. Dark and light shaded areas denote 68 and 90 percent confidence intervals, respectively, computed using Newey--West standard errors.}
\end{figure}
%%%%%%%%%%%%%%%%%%%%%%%%%%%%%%%%%%%%%%%%%%%%%%%%%%%%%%%%%%%%%%%%%%%%%%

\end{document}